\documentclass[prx,aps,floatfix,amsmath,superscriptaddress,twocolumn,longbibliography,nofootinbib]{revtex4-2}
\pdfoutput=1

\usepackage{amssymb,enumerate}
\usepackage{graphicx}
\usepackage{amsmath}
\usepackage{amsthm}
\usepackage{color}
\usepackage{dsfont}
\usepackage{hyperref}
\usepackage{boldline}
\usepackage{aliascnt}
\usepackage[nameinlink,capitalise]{cleveref}
\usepackage{tikz}
    \usetikzlibrary{math}
    \usetikzlibrary{decorations.pathreplacing}
    \usetikzlibrary{quantikz}
\usepackage{adjustbox}
\usepackage{xcolor}
\hypersetup{
    colorlinks,
    linkcolor={red!50!black},
    citecolor={blue!50!black},
    urlcolor={blue!80!black}
}
\usepackage{xfrac}
\usepackage{empheq}
\usepackage{enumitem}
\usepackage{tabularray}

\usepackage{array}
\usepackage{graphicx,graphics,epsfig,times,bm,bbm,amssymb,amsfonts,mathrsfs}
\usepackage[normalem]{ulem}
\usepackage{subfigure}
\usepackage{dsfont}
\usepackage{braket}
\usepackage{natbib}
\usepackage{chngcntr}
\usepackage{nameref}
\usepackage{comment}

\newcommand{\UU}{\text{U}}
\newcommand{\uu}{\mathfrak{u}}
\newcommand{\sv}{\mathfrak{sv}}
\newcommand{\U}{\mathcal{T}}
\newcommand{\V}{\mathcal{V}}
\newcommand{\SU}{\text{SU}}
\newcommand{\Sn}{\mathbb{S}_n}
\newcommand{\su}{\mathfrak{su}}
\newcommand{\g}{\mathfrak{g}}
\newcommand{\R}{\mathbb{R}}
\newcommand{\C}{\mathbb{C}}

\newcommand{\Tr}{\operatorname{Tr}}

\newcommand{\<}{\langle}
\renewcommand{\>}{\rangle}
\renewcommand{\H}{\mathcal{H}}
\newcommand{\Zhat}{J_z}
\newcommand{\Nhat}{a^{\dagger}a}

\newcommand{\HTC}{{H}_{\text{TC}}}
\newcommand{\UTC}{V_{\text{TC}}}
\newcommand{\UZ}{R_z}

\newcommand{\HHO}{\mathcal{H}_{\text{osc}}}
\newcommand{\qmax}{q_{\text{max}}}
\newcommand{\alg}{\mathfrak{alg}}
\newcommand{\Hsym}{\H_{\text{sym}}}
\newcommand{\SV}{\mathcal{SV}}

\newcommand{\jmin}{j_\text{min}}

\newcommand{\qjn}{_{q,j}^{(n)}}
\newcommand{\qj}{_{q,j}}
\newcommand{\Vswap}{V_{\text{qub.}\leftrightarrow\text{osc.}}}
\newcommand{\Uo}{\text{U(1)}}
\newcommand{\Vsym}{\V^{\Uo\times\Sn}}
\newcommand{\SVsym}{\SV^{\Uo\times\Sn}}
\newcommand{\vac}{\ket{0}_{\text{osc}}}

\newcommand{\gTC}{g_{\text{TC}}}
\newcommand{\kosc}{\ket{k}_{\text{osc}}}
\newcommand{\PiSym}{\Pi_{\text{sym}}^{\qmax}}
\newcommand{\vz}{\mathfrak{v}_z}
\newcommand{\wz}{\mathfrak{w}_z}

\newcommand{\1}{\mathbb{I}}
\newcommand{\qeq}{\quad=\quad}
\newcommand{\eq}{\,=\,}
\newcommand{\s}{\text{Span}}
\newcommand{\p}{\,+\,}

\newcommand{\ipo}[3]{\left\langle #1 \middle| #2 \middle| #3 \right\rangle}
\newcommand{\pure}[1]{|#1\rangle\langle #1|}

\newcommand{\bes} {\begin{subequations}}
\newcommand{\ees} {\end{subequations}}
\newcommand{\bea} {\begin{eqnarray}}
\newcommand{\eea} {\end{eqnarray}}
\newcommand{\be} {\begin{equation}}
\newcommand{\ee} {\end{equation}}

\def\al{\alpha}

\def\>{\rangle}
\def\<{\langle}
\def\Tr{\operatorname{Tr}}

\newcommand{\Var}{\operatorname{Var}}

\newcommand{\ignore}[1]{}

\newtheorem{theorem}{Theorem}

\newaliascnt{acknowledgement}{theorem}

\aliascntresetthe{acknowledgement}

\newaliascnt{algorithm}{theorem}

\aliascntresetthe{algorithm}

\newaliascnt{axiom}{theorem}

\aliascntresetthe{axiom}

\newaliascnt{claim}{theorem}

\aliascntresetthe{claim}

\newaliascnt{conclusion}{theorem}

\aliascntresetthe{conclusion}

\newaliascnt{condition}{theorem}

\aliascntresetthe{condition}

\newaliascnt{conjecture}{theorem}

\aliascntresetthe{conjecture}

\newaliascnt{corollary}{theorem}
\newtheorem{corollary}[corollary]{Corollary}
\aliascntresetthe{corollary}

\newaliascnt{criterion}{theorem}

\aliascntresetthe{criterion}

\newaliascnt{definition}{theorem}

\aliascntresetthe{definition}

\newaliascnt{example}{theorem}

\aliascntresetthe{example}

\newaliascnt{exercise}{theorem}

\aliascntresetthe{exercise}

\newaliascnt{lemma}{theorem}
\newtheorem{lemma}[lemma]{Lemma}
\aliascntresetthe{lemma}

\newaliascnt{notation}{theorem}

\aliascntresetthe{notation}

\newaliascnt{problem}{theorem}

\aliascntresetthe{problem}

\newaliascnt{proposition}{theorem}
\newtheorem{proposition}[proposition]{Proposition}
\aliascntresetthe{proposition}

\newaliascnt{remark}{theorem}

\aliascntresetthe{remark}

\newaliascnt{solution}{theorem}

\aliascntresetthe{solution}

\newaliascnt{summary}{theorem}

\aliascntresetthe{summary}

\newcounter{step}
\newenvironment{step}
  {\par\vspace{4pt}\noindent\refstepcounter{step} \textbf{Step \arabic{step}.}}
  {\vspace{4pt}}

\crefname{tab}{Table}{Tables}
\Crefname{tab}{Table}{Tables}

\crefname{eq}{Eq.}{Eqs.}
\Crefname{eq}{Eq.}{Eqs.}

\crefname{figure}{Figure}{Figures}
\Crefname{figure}{Figure}{Figures}

\crefname{section}{Sec.}{Secs.}
\Crefname{section}{Sec.}{Secs.}

\crefname{appendix}{Appendix}{Appendices}
\Crefname{appendix}{Appendix}{Appendices}

\crefname{theorem}{Theorem}{Theorems}
\Crefname{theorem}{Theorem}{Theorems}

\crefname{acknowledgement}{Acknowledgement}{Acknowledgements}
\Crefname{acknowledgement}{Acknowledgement}{Acknowledgements}

\crefname{algorithm}{Algorithm}{Algorithms}
\Crefname{algorithm}{Algorithm}{Algorithms}

\crefname{axiom}{Axiom}{Axioms}
\Crefname{axiom}{Axiom}{Axioms}

\crefname{claim}{Claim}{Claims}
\Crefname{claim}{Claim}{Claims}

\crefname{conclusion}{Conclusion}{Conclusions}
\Crefname{conclusion}{Conclusion}{Conclusions}

\crefname{condition}{Condition}{Conditions}
\Crefname{condition}{Condition}{Conditions}

\crefname{conjecture}{Conjecture}{Conjectures}
\Crefname{conjecture}{Conjecture}{Conjectures}

\crefname{corollary}{Corollary}{Corollaries}
\Crefname{corollary}{Corollary}{Corollaries}

\crefname{criterion}{Criterion}{Criteria}
\Crefname{criterion}{Criterion}{Criteria}

\crefname{definition}{Definition}{Definitions}
\Crefname{definition}{Definition}{Definitions}

\crefname{example}{Example}{Examples}
\Crefname{example}{Example}{Examples}

\crefname{exercise}{Exercise}{Exercises}
\Crefname{exercise}{Exercise}{Exercises}

\crefname{lemma}{Lemma}{Lemmas}
\Crefname{lemma}{Lemma}{Lemmas}

\crefname{notation}{Notation}{Notations}
\Crefname{notation}{Notation}{Notations}

\crefname{problem}{Problem}{Problems}
\Crefname{problem}{Problem}{Problems}

\crefname{proposition}{Proposition}{Propositions}
\Crefname{proposition}{Proposition}{Propositions}

\crefname{remark}{Remark}{Remarks}
\Crefname{remark}{Remark}{Remarks}

\crefname{solution}{Solution}{Solutions}
\Crefname{solution}{Solution}{Solutions}

\crefname{summary}{Summary}{Summaries}
\Crefname{summary}{Summary}{Summaries}

\crefname{step}{Step}{Steps}
\Crefname{step}{Step}{Steps}

\crefname{enumi}{Part}{Parts}
\Crefname{enumi}{Part}{Parts}

\crefformat{table}{Table #2#1#3}
\crefformat{equation}{Eq.~(#2#1#3)}
\crefformat{figure}{Figure #2#1#3}
\crefformat{section}{Sec.~#2#1#3}
\crefformat{appendix}{Appendix #2#1#3}

\crefformat{theorem}{Theorem #2#1#3}
\crefformat{lemma}{Lemma #2#1#3}
\crefformat{proposition}{Proposition #2#1#3}
\crefformat{corollary}{Corollary #2#1#3}
\crefformat{claim}{Claim #2#1#3}
\crefformat{definition}{Definition #2#1#3}
\crefformat{remark}{Remark #2#1#3}
\crefformat{example}{Example #2#1#3}
\crefformat{conjecture}{Conjecture #2#1#3}
\crefformat{step}{Step #2#1#3}

\makeatletter 
\renewcommand\onecolumngrid{%
  \do@columngrid{one}{\@ne}%
  \def\set@footnotewidth{\onecolumngrid}%
  \def\footnoterule{\kern-6pt\hrule width 1.5in\kern6pt}%
}
\makeatother

\begin{document}

\title{Global Control with the Tavis-Cummings Interaction}
\author{Plato Deliyannis}
\email{plato.deliyannis@duke.edu}
\affiliation{Duke Quantum Center and Department of Physics, Duke University, Durham, NC 27708, USA}
\author{Iman Marvian}
\email{iman.marvian@duke.edu}
\affiliation{Duke Quantum Center and Department of Physics, Duke University, Durham, NC 27708, USA}\affiliation{Department of Electrical and Computer Engineering, Duke University, Durham, NC 27708, USA}

\begin{abstract}
We study the controllability of a system of qubits under global control, where
control pulses act identically on all qubits.
Specifically, we consider a collection of qubits identically coupled to a single
bosonic mode, or harmonic oscillator, via the Jaynes-Cummings interaction.
This collective coupling, known as the Tavis-Cummings (TC) interaction, has been
realized in several quantum computing platforms, including superconducting and
atomic qubit systems.
Although the qubits do not interact directly with one another, they can become
entangled through their common coupling to the bosonic mode.
We characterize the group of unitaries that can be implemented on the joint Hilbert space of the qubits and bosonic mode using the TC interaction together with a global $z$ field $J_z$, corresponding to identical $z$ rotations on all qubits.
We show that for $n\ge 3$ qubits the set of realizable unitaries is restricted by an ``accidental'' symmetry of the TC Hamiltonian, distinct from its ``standard'' U(1) and permutational symmetries.
On the other hand, we find that the Hamiltonian $J_z^2$ breaks this accidental symmetry and, together with the TC interaction and $J_z$, achieves semi-universality: it allows the implementation of arbitrary unitaries that respect permutational and U(1) symmetry, up to certain constraints on the center of the group.
In a companion paper, we further analyze this remarkable accidental symmetry and show that it can be understood through Schwinger's bosonic model of angular momentum.
\end{abstract}

\maketitle

\section{Introduction}
Global control is a paradigm in quantum computing and quantum control in which individual qubit addressing is minimized, and control is instead implemented through collective fields that act uniformly on all qubits.
It is a promising route toward scaling quantum computers to larger system sizes \cite{Lloyd_1993,Benjamin_2001,Benjamin_2000}, as well as toward efficiently realizing multi-qubit operations \cite{Fitzsimons_etal_2007,Raussendorf_2005,Kay_Pachos_2004}, both of which are prerequisites for realizing virtually any potential advantage of quantum computing \cite{cesa2023universal,menta_etal_actuators_2026}.
A canonical example is a collection of qubits coupled identically to a single bosonic mode, or quantum harmonic oscillator.
Such systems arise naturally in several quantum computing platforms and provide a basic setting in which effective interactions between qubits are mediated by
the bosonic mode.
Examples include ion qubits coupled to a vibrational mode in trapped-ion systems \cite{Leibfried_2003_trappedion,Haffner_2008_trappedion}, superconducting qubits in circuit QED \cite{Blais_2004_cQED,Wallraff_2004_cQED}, and atoms in cavity QED \cite{Raimond_2001_cavity,Walther_2006_cavity, Varcoe_2000_cavity}.

Often, the interaction between a single qubit and the bosonic mode is well approximated by the Jaynes-Cummings (JC) interaction \cite{Jaynes_Cummings_1963,JC_history}, which has been widely used as a tool for realizing universal quantum computation \cite{Childs_Chuang_2000,Yuan_Lloyd_2007,Cirac_Zoller_1995}.
In the special case where all qubits couple identically to the bosonic mode, so that the interaction is invariant under permutations of the qubits, the resulting model is known as the Tavis-Cummings (TC) model \cite{tavis_1968_exact_solut,TC2_1969}.
The TC model has been studied extensively in contexts ranging from atomic physics to quantum control \cite{Bashir_Abdalla_1995,Bogoliubov_etal_1996,Rybin_etal_1998,Tessier_etal_2003,Vadeiko_etal_2003,Genes_etal_2003,Fink_etal_2009,Agarwal_etal_2012,Zou_etal_2013,Keyl_2014_control}.
The TC Hamiltonian can be realized, for example, through the coupling of a linear trapped-ion chain to its center-of-mass vibrational mode \cite{Retzker_2007, Molmer_Sorensen}, or in superconducting-qubit platforms \cite{Fink_etal_2009}.
\begin{figure*}[htp]
    \centering
    \includegraphics[width=0.9\textwidth]{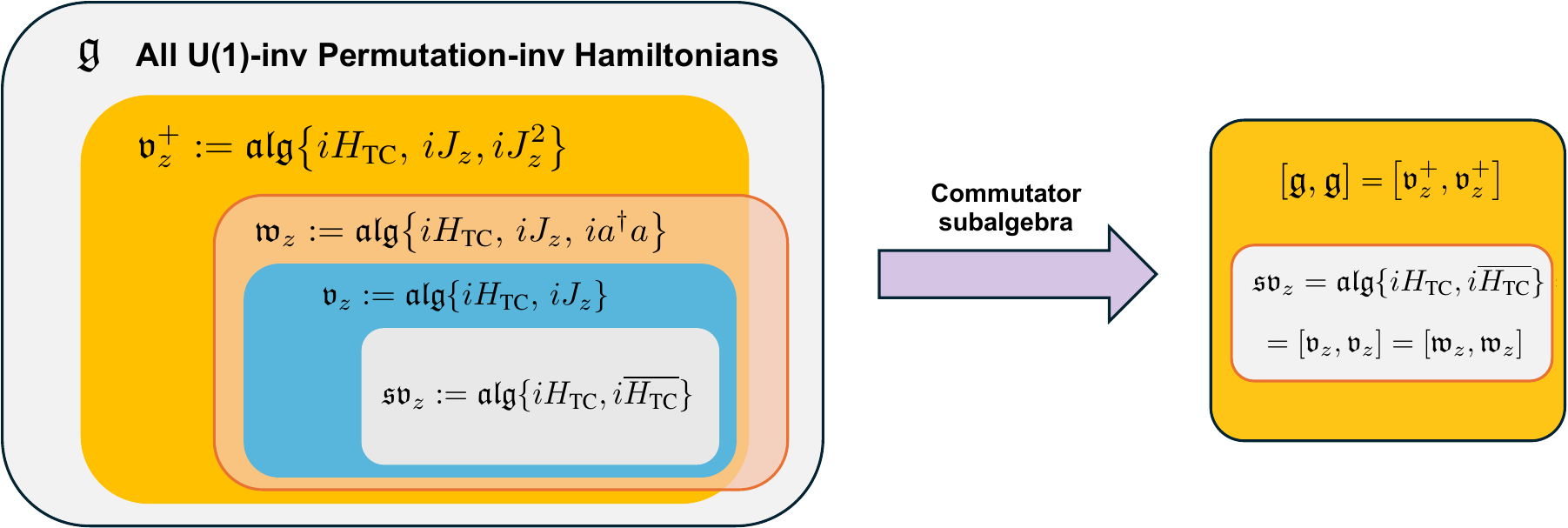}
    \caption{\textbf{First main result (\cref{thm:Thm0} and \cref{thm:full_characterization}).} Schematic relationship between different Lie algebras generated by various subsets of $\HTC\propto ({J}_+{a}+{J}_-{a}^{\dag})$, $\overline{\HTC}\propto i({J}_+{a}-{J}_-{a}^{\dag})$, $J_z$, $J_z^2$, and $a^\dag a$.
    All these algebras are subalgebras of $\mathfrak{g}$, the Lie algebra of all PI, U(1)-invariant Hamiltonians. 
    Although the algebras $\mathfrak{sv}_z$, $\vz$, $\wz$, $\vz^+$, and $\mathfrak{g}$ are all distinct, after removing their centers, i.e., after passing to their commutator subalgebras, they reduce to only two possibilities: $[\mathfrak{g},\mathfrak{g}]=[\vz^+,\vz^+]$ and $\mathfrak{sv}_z=[\vz,\vz]=[\wz,\wz]$. 
    The latter algebra is restricted by the accidental symmetry of $\HTC$, which, to the best of our knowledge, has not been identified before. The centers removed by passing to the commutator subalgebras induce relative phases between sectors with different U(1) charges or different angular momenta. 
    See \cref{thm:Thm0} for precise  relationships, and \cref{sec:accidental_sym} for further discussion of the accidental symmetry, which is investigated in more detail in a companion paper \cite{symmetry_paper}.
    For a system with $n$ qubits, the difference between the dimensions of the commutator algebras $[\mathfrak{g},\mathfrak{g}]=[\vz^+,\vz^+]$ and $\mathfrak{sv}_z=[\vz,\vz]=[\wz,\wz]$ is $n(n-2)(n+2)(n+4)/48$
    when $n$ is even, and is $n(n-2)(n+2)(n+4)/48 + (12n+3)/48$ when $n$ is odd.
    See also \cref{thm:full_characterization} for a full characterization of the associated Lie groups, in particular the Lie group of unitaries that can be realized using the Hamiltonians $\HTC$ and $J_z$.
    }
 \label{fig:Lie}
\end{figure*}

In this work, we study a system of qubits coupled to a bosonic mode through the permutation-invariant (PI) TC interaction, together with global uniform $z$ field that respect the same permutational symmetry.
While controllability and universality of quantum systems in the absence of such symmetries are well understood, characterizing the operations accessible under these constraints remains substantially more challenging.
Indeed, recent work has revealed unexpectedly strong restrictions on the set of realizable unitaries in the presence of symmetries (See, e.g., \cite{marvian_sym_loc_2022, marvian_SU(2), HLM_2024_SU(d),marvian_abelian_2024}).
\footnote{For instance, a recent no-go theorem shows that generic $n$-qubit unitaries respecting a given symmetry cannot be implemented using only $2$-qubit, or more generally $k<n$-qubit, gates that respect the same symmetry \cite{marvian_sym_loc_2022}.}

Specifically, we investigate and characterize the group $\V_z$ of all unitary time evolutions realizable on the combined qubit-bosonic system using the TC interaction together with a global $z$ field.
Our main characterization is stated in \cref{thm:full_characterization}.
To obtain this characterization, we apply a recent result of \cite{HLM_2024_SU(d)}, which provides simple criteria for determining whether a given set of symmetric Hamiltonians is universal within the symmetry-respecting subspace, i.e., whether it generates all unitaries compatible with the relevant symmetries.

In the present setting, the TC interaction and the global $z$ field respect both permutation invariance and a U(1) symmetry associated with conservation of the total excitation number, generated by $a^{\dag}a+J_z$.
Here, $a^\dag a$ is the harmonic-oscillator Hamiltonian, and $J_z$ is the total angular momentum of the qubits along the $z$ direction; see \cref{sec:setup}.
Equivalently, these Hamiltonians respect a $\Uo\times\Sn$ symmetry.
Consequently, every unitary in $\V_z$ is block diagonal with respect to the decomposition of the combined qubit-bosonic Hilbert space into sectors carrying inequivalent irreducible representations of $\Uo\times\Sn$.
These sectors are labeled by the conserved U(1) charge and the total angular momentum associated with the permutational symmetry.
This decomposition reduces the infinite-dimensional joint Hilbert space to an infinite direct sum of finite-dimensional invariant subspaces.
Accordingly, instead of characterizing $\mathcal{V}_z$ directly on the full Hilbert space, we characterize its projection to any finite, but arbitrarily large, collection of invariant sectors.

The group of all PI, U(1)-invariant unitaries would allow independent unitary transformations on each of these sectors.
The central question is therefore whether the TC interaction and a global $z$ field generate this entire symmetry-allowed group, or whether they impose additional restrictions.

In \cref{thm:full_characterization}, we fully characterize the group $\V_z$  of unitaries realizable with these two Hamiltonians, and find that there are exactly two types of restrictions on the components of unitaries in $\V_z$ across different sectors.

\vspace{0.2em}
\noindent $\bullet$ \textbf{Central constraints}: The center of $\V_z$ is strictly smaller than the center of the full group of unitaries respecting the $\Uo\times\Sn$ symmetry, namely $\Vsym$. 
The latter center consists of relative phases between sectors carrying inequivalent representations of $\Uo\times\Sn$; see \cref{sec:beyond_sym} for a precise definition.

\noindent $\bullet$\textbf{Constraints from ``accidental'' symmetry:} We identify a different symmetry of TC interaction, which to the best of our knowledge has not been identified before. For $n\geq 3$ qubits, this symmetry forces the unitaries realized in certain invariant sectors to be identical.
In our companion paper \cite{symmetry_paper}, we discuss this unexpected symmetry further and show that it can be explained using Schwinger's bosonic model of angular momentum. 
Here, we focus on its consequences for controllability.

\vspace{0.2em}
While the first type of constraint is generic and unavoidable for any finite set of Hamiltonians respecting the $\Uo\times\Sn$ symmetry, the second type arises from the specific form of  matrix elements of $\HTC$ and $J_z$.
Indeed, we show that these accidental-symmetry constraints can be removed by adding the Hamiltonian $J_z^2$, which is PI and also respects the U(1) symmetry.
Specifically, in \cref{sec:Jz2_proof}, we show that $\HTC$, $J_z$, and $J_z^2$ generate all unitary transformations compatible with the permutational and U(1) symmetries, up to the generic central constraints described above.

\cref{thm:Thm0} summarizes the relationship between the Lie algebras generated by various combinations of $\HTC$, $J_z$, $J_z^2$, and $a^\dag a$, and the Lie algebra $\g$ of all PI, U(1)-invariant Hamiltonians acting on $n$ qubits coupled to a bosonic mode.
The corresponding inclusions and separations are illustrated in \cref{fig:Lie}.
We note that related controllability questions were previously studied by Keyl, Zeier, and Schulte-Herbrüggen \cite{Keyl_2014_control} in the totally symmetric subspace of the qubits, corresponding to the highest-angular-momentum irreducible representation (irrep)  $j=n/2$.
In particular, restricted to a single U(1) charge sector within this totally symmetric subspace, their results imply that $\HTC$ and $J_z$ generate all unitaries on that sector.
However, this earlier work did not address the simultaneous action on multiple charge and angular-momentum sectors.
As a result, it did not capture the additional constraints we identify here, including the accidental-symmetry constraints that lock the unitaries realized in different invariant sectors together.

\vspace{0.5em}
Our second main result, presented in \cref{thm:general_k_ancilla} and illustrated in \cref{fig:VU_gate}, builds on the above characterization and addresses a fundamental question with broad applications in quantum computing and control: Suppose that, as depicted in \cref{fig:VU_gate}, the bosonic mode is used as an ancillary system: it is initialized in a fixed energy eigenstate $\ket{k}_{\text{osc}}$ and returned to the same state at the end of the protocol. 
What unitary transformations $U$ can then be realized on the qubits using only the Hamiltonians $\HTC$ and $J_z$?

\begin{figure}[htp]
    \centering
    \includegraphics[width=0.7\linewidth]{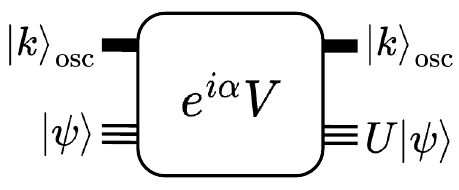}
    \caption{\textbf{Second main result (\cref{thm:general_k_ancilla}).}
    The goal is to implement a unitary $U$ on the qubits by coupling them to a bosonic mode initialized in the eigenstate $\ket{k}$ of its intrinsic Hamiltonian $a^\dag a$, with eigenvalue $k$.
    After the operation, the bosonic mode returns to the same pure state and hence remains unentangled with the qubits.
    This is achieved by implementing a unitary $V$ using the Hamiltonians $\HTC$ and $J_z$.
    The phase $e^{i\alpha}$ denotes the freedom to choose an overall global phase; see \cref{eq:VU_ancilla_k}.
    \cref{thm:general_k_ancilla} determines which unitary transformations $U$ can be realized in this way.}
    \label{fig:VU_gate}
\end{figure}

We fully address this question in \cref{thm:general_k_ancilla}.
Remarkably, we find that, in addition to the permutational and U(1) symmetries inherited from the Hamiltonians $\HTC$ and $J_z$, the resulting qubit unitaries must satisfy additional constraints, which are characterized in the theorem.

Building on these results, in the companion paper \cite{circuit_paper}, which focuses on quantum circuits, we discuss applications of \cref{thm:full_characterization} to quantum computation.
In particular, as briefly discussed in \cref{sec:qubit_unis}, we prove in \cite{circuit_paper} that \emph{any PI unitary transformation on $n$ qubits can be realized using the TC interaction together with global $z$ and $x$ fields, assuming that the qubits couple identically to an bosonic mode initially prepared in the vacuum state of its intrinsic harmonic-oscillator Hamiltonian}.

This result, which follows as a corollary of \cref{thm:general_k_ancilla}, shows that this globally controlled qubit-bosonic system is sufficient to implement arbitrary PI gates on the qubits, including unitaries such as $\exp({i\theta Z^{\otimes n}}):\theta\in[0,2\pi)$ and multi-qubit controlled-$Z$ gates.
The latter are equivalent, up to single-qubit Hadamard gates on the target qubit, to multi-qubit Toffoli gates, which are among the most useful primitive operations in quantum computation.

\vspace{1em}
The rest of this paper is organized as follows:
\begin{itemize}
\item \cref{sec:setup} introduces the TC interaction and discusses its permutational and U(1) symmetries.
Then, we discuss different variants of control systems that combine this interaction with other PI Hamiltonians.
\item \Cref{sec:commutator_subgroups} introduces commutator subgroups and central constraints, which are used to characterize the groups of unitaries realized for these control systems.
\item \Cref{sec:summary} provides a summary of the main results; \cref{thm:Thm0} describes the relationships between the groups of unitaries realized for each control system in terms of their Lie algebras.
\item \cref{sec:sym_subspace} characterizes the group $\V_z$ of unitaries realized using the TC interaction and a global $z$ field, restricted to the totally symmetric subspace of the qubits, corresponding to the highest angular-momentum sector $j=n/2$ (see \cref{lem:sym_semi_universality}).
In \cref{sec:Vswap}, we also prove the existence of a universal qubit-bosonic SWAP gate that exchanges the state of the qubits in the totally symmetric subspace with the state of the bosonic mode.
\item \cref{sec:beyond_sym} generalizes the results of the previous section beyond the highest angular-momentum sector $j=n/2$ to arbitrary $j$, in order to fully characterize $\V_z$.
\Cref{sec:full_char_proof} proves the main results of \cref{sec:sym_subspace,sec:beyond_sym}.
\item \cref{sec:accidental_sym} describes an unexpected symmetry of the TC Hamiltonian, which is unrelated to its permutational and U(1) symmetries (see \cref{prop:accidental_symmetry_matrix}).
\item Finally, \cref{sec:qubit_unis} briefly discusses some implications of \cref{thm:full_characterization} for implementing qubit unitaries by coupling qubits via the TC interaction to a bosonic mode initialized in an eigenstate of its intrinsic Hamiltonian $a^{\dag}a$: see \cref{thm:general_k_ancilla}.
These applications are further discussed in our companion paper, \cite{circuit_paper}.
\end{itemize}

\section{Setup: Tavis-Cummings Interaction}
\label{sec:setup}
Consider a system of $n$ qubits, or two-level systems, coupled identically to bosonic mode, or quantum harmonic oscillator.
In this paper, we use ``bosonic mode'' and ``oscillator'' interchangeably.
The total qubit-oscillator interaction is described by the Tavis-Cummings (TC) Hamiltonian \cite{tavis_1968_exact_solut}
\begin{align}
 \begin{split}
    \HTC
    \,:=&\,\, \frac{\gTC}{2}\sum_{i=1}^{n} \Big(\sigma_+^{(i)}a + \sigma_-^{(i)}a^{\dag}\Big) \\
    \,=&\,\, \gTC\,\big({J}_+a + {J}_-a^{\dag}\big)\,,
 \end{split}
 \label{eq:TCham}
\end{align}
where $\gTC$ is the coupling strength, $a$ is the annihilation operator on the oscillator Hilbert space satisfying $[a,a^\dag]=\mathbb{I}$, and
\begin{align}
    {J}_{\pm}:=J_x\pm iJ_y\,,
\end{align}
where
\begin{align}
    J_w \,:=\, \frac{1}{2}\sum_{i=1}^{n} {\sigma}_w^{(i)},\qquad
    w=x,y,z ,
 \label{eq:Zham}
\end{align}
are the total angular momentum operators of the $n$ qubits.
The Hamiltonian acts on the joint qubit-oscillator Hilbert space
\begin{align}
    \H_{\text{qubits}}\otimes\H_{\text{osc}}
    = (\C^2)^{\otimes n}\otimes\mathcal{L}^2(\R).
\end{align}
In addition to $\HTC$, we consider the PI qubit Hamiltonian $\omega_z(t)J_z$, which can be realized by applying uniform $z$ field to all qubits. 
In general, a collective rotation about the axis $w=x,y,z$ by an angle $\phi\in[-2\pi,2\pi)$ is described by
\begin{align}
    R_w(\phi):=\exp(-i\phi J_w).
\end{align}
Since 
\begin{align}
    R_z(\pi)^{\dag}J_wR_z(\pi)=-J_w \quad:\quad w=x,y \,,
\end{align}
it follows that
\begin{align}
    R_z(\pi)^{\dag}\HTC R_z(\pi)=-\HTC\,.
\end{align}
Thus, the ability to perform collective $z$ rotations allows us to reverse the sign of the TC interaction. 
We may therefore take $\gTC>0$ without loss of generality. More generally, by interleaving evolution under $\HTC$ with $\pi$ rotations $R_z(\pi)$ at appropriate times, one can engineer an effective time-dependent interaction of the form
\begin{align}
    f(t)\HTC\,,\quad |f(t)|\leq 1\,.
\end{align}
Hence, in the following, we assume Hamiltonians $\HTC$ and $J_z$ can be turned on and off at will. 
The harmonic oscillator has an intrinsic Hamiltonian $H_\text{osc}=\nu a^\dag a$, which is typically fixed and time independent.
We therefore work in a rotating frame, or interaction picture, in which the intrinsic evolution generated by $H_\text{osc}$ is removed.
More specifically, we use the rotating frame generated by $\nu(a^\dag a+J_z)$, corresponding to the transformation
\begin{align}
    \ket{\psi(t)}_{\rm Rot}=
    \exp\!\big(i\nu(a^\dag a+J_z)t\big) \ket{\psi(t)}_{\rm Lab}\,.
\end{align}
In this frame, the total Hamiltonian takes the form
\begin{align} \label{eq:totH}
    H(t) = f(t)\HTC + \omega_z(t)J_z\,,
\end{align}
where the real control amplitudes $f(t)$, $\omega_z(t)$ may be taken to be piecewise constant. 
(Indeed, it suffices to consider the special case where, at each moment $t$ in time, only one of the control functions $f(t)$ or $z(t)$ is $\pm1$,  while the other is zero.)
Here, $\omega_z(t)$ denotes the detuning from the rotating-frame frequency, i.e., the laboratory-frame $z$-field strength shifted by $\nu$.

We consider unitary transformations $V(T)$ that can be realized using Hamiltonians in \cref{eq:totH} under the Schr\"odinger equation 
\begin{align}
    \frac{d}{dt}V(t) \eq -i H(t) V(t)\quad:\quad0 \le t\le T \,,
\end{align}
for arbitrary time $T$, with the initial condition $V(0)=\mathbb{I}$.
We denote the group of such unitary transformations by $\mathcal{V}_z$.
Equivalently, $\V_z$ is the group of unitaries that can be realized as finite sequences of unitaries $R_z(\phi): \phi\in[-2\pi,2\pi)$ and 
\begin{align}
    \UTC(r)&:=\exp(-i r  \HTC /\gTC)\,,
\end{align}
where $r\in\mathbb{R}$ and $\phi\in [-2\pi,2\pi)$ are dimensionless parameters, i.e.,
\begin{align}
    \V_z\,=\,\big\langle \exp({-it\HTC}),\,\exp(-i t J_z) \,:\, t\in\mathbb{R} \big\rangle\,.
\end{align}
In other words, there are finite $L\in\mathbb{N}$ and $T\in\R$ such that every $V\in\V_z$ can be decomposed as
\begin{align}
    V \eq \prod_{\ell=1}^{L} \UTC(r_{\ell}) R_{z}(s_{\ell})\,\, : \,\,r_{\ell},s_{\ell}\in[-T,T]\,.
 \label{eq:Vz_circ_def}
\end{align}

A third, equivalent way to define $\V_z$ is as the connected group obtained from the real Lie algebra generated from Hamiltonians $\HTC$ and $J_z$, that is
\begin{align}
    \V_z = \exp(\vz)\,\,:\,\,\vz=\alg_{\R}\big\{i\HTC, i\Zhat\big\}\,.
\end{align}
As we further discuss in \cref{sec:commutator_subgroups}, $\vz$ is an infinite-dimensional Lie algebra. 
In general, analyzing such algebras can be substantially more difficult
than analyzing finite-dimensional ones. 
In the present case, however, the analysis is greatly simplified by a fundamental property of the TC interaction: its U(1) symmetry.

\subsection{Permutational and U(1) symmetries of the TC interaction}
\label{sec:TC_and_z}
A characteristic feature of the TC Hamiltonian, which plays a crucial role in this work, is that in addition to permutational symmetry, it also respects a U(1) symmetry associated with the conserved operator $a^\dag a + J_z$.
In particular, we define the \textit{charge} operator
\begin{align}
    Q \,:=\, a^\dag a + J_z + \frac{n}{2} \,\,=\,\, \sum_{q=0}^\infty q\,\Pi_q \,,
 \label{eq:charge}
\end{align}
where the second equation is the spectral decomposition of $Q$: the integers $q \ge 0$ are its eigenvalues, and $\Pi_q$ is the projector onto the corresponding eigen-subspace.
The shift by $n/2$ is chosen so that the spectrum of $Q$ is the non-negative integers $\{0,1,2,\dotsc\}$.
(Note that, as is common, we often suppress tensor products with identity operators.)
This Hermitian operator defines a unitary representation of the U(1) symmetry as
\begin{align}\label{eq:U(1)}
 \begin{split}
    \exp(i\theta Q)=e^{i\theta n/2}&\,\exp(i\theta J_z)\otimes\exp(i\theta\Nhat)\\[4pt]
    &\qquad:\,\,\theta\in[0,2\pi)\,.
 \end{split}
\end{align} 
Clearly, the $z$-field Hamiltonian $J_z$ also conserves this charge, whereas the $x$-field Hamiltonian $J_x$ does not; that is,
\begin{align}
    [\HTC,Q]=[J_z,Q]=0\,,
\end{align}
but $[J_x,Q]\neq 0$.
The presence of this U(1) symmetry, which has been extensively exploited since the original work of Jaynes, Tavis and Cummings \cite{Jaynes_Cummings_1963,tavis_1968_exact_solut}, provides a powerful tool for understanding the dynamics of this model.
In particular, thanks to this fundamental property of the TC interaction, the infinite-dimensional Hilbert space $\H_{\text{qubits}}\otimes\H_{\text{osc}}$ decomposes into an infinite direct sum of finite-dimensional subspaces.
With respect to this decomposition, the Lie algebra $\vz$ and the group $\mathcal{V}_z$ are block diagonal.

It is worth emphasizing that adding a global $x$ field qualitatively changes this structure: from any initial state, the reachable subspace is generally infinite dimensional.
In particular, the oscillator can leak into arbitrarily high energy levels, which substantially complicates the analysis. 

Another crucial property of $\HTC$ is its permutational symmetry, which follows from the fact that all qubits are coupled identically to the oscillator.
In particular, $\HTC$ commutes with $\textbf{P}_{i,j}$, the SWAP unitary acting on any pair of qubits $i$ and $j$, also known as a transposition.
Global fields such as $J_z$ also have this property, and hence so does the charge $Q$:
\begin{align}
    [\HTC, \textbf{P}_{i,j}]=[J_z, \textbf{P}_{i,j}]=[Q, \textbf{P}_{i,j}]=0 \,,
\end{align}
for all $i$ and $j$.

Clearly, $\V_z$ only contains unitaries that respect both the permutational and U(1) symmetries.  
The question is whether the converse holds: namely, whether every unitary with these symmetries belongs to $\V_z$, or equivalently, can be realized using $\HTC$ and global rotations about the $z$-axis.
Before addressing this central question, we discuss several variants of the problem.

\subsection{Other variants of the problem}
\subsubsection{Phase-shifted TC Interaction}
Instead of turning the TC interaction on and off and interleaving these intervals with $z$-rotations of the qubits as in \cref{eq:Vz_circ_def}, one can simply apply phase-shifted versions of the TC interaction, namely,
\begin{align}
 \begin{split}
    \HTC(\phi)\,:=&\,\,R_z^{\dag}(\phi) \HTC R_z(\phi)\\[4pt]
    =& \,\,\gTC\,\big(e^{i\phi}{J}_+{a} + e^{-i\phi}{J}_-{a}^{\dag}\big)\,,
 \end{split}
 \label{eq:HTC_phase_shifted}
\end{align}
for $\phi\in[0,2\pi)$, and move the overall rotations around $z$ to the end of the process. 
Furthermore, $\HTC(\phi)$ can be written as a linear combination of $\HTC$ and the Hamiltonian
\begin{align}
 \begin{split}
    \overline{\HTC}\,\,:=\HTC\left(\frac{\pi}{2}\right) \,\,&=\,\, i \gTC\,({J}_+{a} - {J}_-{a}^{\dag}) \\[4pt]
    \,\,&=\,\, i[J_z,\HTC]\,.
 \end{split}
 \label{eq:HTCbar}
\end{align}
Therefore, the unitaries realized by such Hamiltonians form the group
\begin{align}
 \begin{split}
    \mathcal{SV}_z:=&\,\,\big\langle \exp\big({-it\HTC(\phi)\big)} \,:\, t\in\mathbb{R},\phi\in[0,2\pi) \big\rangle\\[4pt]
    =&\,\,\big\langle \exp\big({-it\HTC}\big),\,\exp\big({-it\overline{\HTC}}\big) \,:\, t\in\mathbb{R}\big\rangle\,,
 \end{split}
\end{align}
with the corresponding Lie algebra
\begin{align}
 \begin{split}
    \sv_z :=&\,\, \alg_{\R}\big\{ i\HTC(\phi): \phi\in[0,2\pi)\big\}\\[4pt]
    =&\,\,\alg_{\R}\big\{ i\HTC, i \overline{\HTC}\big\}\,.
 \end{split}
\end{align}
Clearly, $\mathcal{SV}_z$ and $\sv_z$ are, respectively, a  subgroup and subalgebra of $\V_z$ and $\vz$.

\vspace{2mm}
It is also worth noting that $\mathcal{V}_z$ is generated by the subgroup
$\mathcal{SV}_z$ together with the one-parameter family of U(1) rotations
$R_z(\phi)$, with $\phi\in[0,4\pi)$. 
Indeed, the above argument shows that
$\mathcal{SV}_z$ is a normal subgroup of $\mathcal{V}_z$: for any $\widetilde{V}\in\mathcal{SV}_z$ and any $\phi$, $R_z(\phi) \widetilde{V} R_z(\phi)^\dag \in \mathcal{SV}_z$. 
(As we explain further in \cref{sec:commutator_subgroups}, $\mathcal{SV}_z$ is in fact a special normal subgroup, namely the commutator subgroup of $\mathcal{V}_z$).
Therefore, every unitary $V\in\mathcal{V}_z$ can be written in the form
\be\label{eq:Dec}
V=R_z(\phi) \widetilde{V}\,,
\ee
for some $\widetilde{V}\in\mathcal{SV}_z$ and $\phi\in[0,4\pi)$.

\subsubsection{Varying the Oscillator Hamiltonian }
Another variant of 
Hamiltonian \cref{eq:totH} is
\begin{align}
    H_{\text{Lab}}(t) = f(t) \HTC + \omega_z(t) J_z + \nu\Nhat\,\,:\,\,0\le t\leq T\, ,
 \label{eq:Hlab}
\end{align}
which includes the time-independent harmonic oscillator Hamiltonian $\nu a^\dag a$. 
This Hamiltonian is indeed physically relevant as it describes the evolution of the system in the lab frame. 
However, as explained at the beginning of this section, one can  eliminate this constant time-independent term by going to the rotating frame defined by the transformation $\ket{\psi(t)}_{\text{Rot}} = \exp(i \nu (\Nhat + \Zhat) t) \ket{\psi(t)}_{\text{Lab}} \,$.  
Then, it is straightforward to see that the time evolution of the system will be described by a  Hamiltonian in the form of \cref{eq:totH} with a modified $\omega_z(t)$. 
 
One can also consider another variant of this Hamiltonian in which one has control over the oscillator Hamiltonian, namely
\begin{align}
    \widetilde{H}_{\text{Lab}}(t) = f(t)\HTC + \omega_z(t) \Zhat+ n(t)  \nu \Nhat\,,
 \label{eq:H_controlling_N}
\end{align}
with a real function $n(t)$. 
It turns out that this does not change the set of realizable unitaries. 
That is, even when restricted to time $t\ge 0$, as in \cref{eq:Hlab}, which defines a direction of time, the presence of time-independent (drift) Hamiltonian $\nu a^\dag a$, does \textit{not} restrict the set of realizable unitaries.
This is a consequence of periodicity of this Hamiltonian -- if we set $f(t)=\omega_z(t)=0$, and let the oscillator freely evolve, we can realize both $\exp(-i t a^\dag a)$ and $\exp(i t a^\dag a)$ for arbitrary $t\ge 0$, which is equivalent to having control over $n(t)$ in \cref{eq:H_controlling_N}.

The Hamiltonian $\nu a^\dag a$, whether controllable and time-dependent or not, extends the set of realizable unitaries to a larger group, that contains $\V_z$, namely
\begin{align}
 \begin{split}
   \mathcal{W}_z:=\big\langle \exp({-it\HTC}),\,&\exp({-i t J_z}),\\
   &\hspace{-10pt}\exp({-ita^\dag a})\,:\, t\in\mathbb{R} \big\rangle\,,
 \end{split}
\end{align}
with the corresponding Lie algebra
\begin{align}
 \wz \,:=\,\, \alg_{\R}\big\{i\HTC,\,i\Zhat,\,ia^{\dag}a\big\}\,.
\end{align}
Note that $[a^\dag a, J_z]=0$ and
\begin{align}\label{eq:ab}
    \big[a^\dag a, \HTC\big]=\left[Q-J_z-\frac{n}{2},\,\HTC\right]=-\big[J_z, \HTC\big]\,,
\end{align}
which, as we further explain in the next section, implies that the only difference between $\wz$ and $\vz$ is their center (see \cref{eq:Lie_vw}).
Furthermore, this implies that any unitary $W\in\mathcal{W}_z$, has decompositions as
\begin{align}
    W \eq e^{-i \theta a^\dag a} V \eq R_z(\phi)e^{-i t a^\dag a} \widetilde{V}\,, 
\end{align}
for proper $\theta\in [0,2\pi)$, $V\in \mathcal{V}_z$,  $\widetilde{V}\in \mathcal{SV}_z$, and $\phi\in[0,4\pi)$, where to obtain the second decomposition we have used \cref{eq:Dec}.

\subsubsection{Anti-Tavis-Cummings Hamiltonian}
Finally, we note that all the results of this paper have an immediate counterpart for the so-called anti-Tavis-Cummings (anti-TC) Hamiltonian, which is obtained from the TC Hamiltonian by exchanging $J_+$ and $J_-$, or equivalently by exchanging $a$ and $a^\dag$.
Equivalently, it is related to $\HTC$ by a global $\pi$ rotation of the qubits about the $x$ axis:
\begin{align} 
    H_{\text{anti-TC}}=g_{\text{anti-TC}} (J_+ a^\dag+J_- a) = e^{i \pi J_x} \HTC e^{-i \pi J_x} \,. 
\end{align}
Thus, the anti-TC model is unitarily equivalent to the TC model. 
In trapped-ion terminology, the TC and anti-TC interactions correspond, respectively, to the red- and blue-sideband spin-motion couplings.

This unitary equivalence implies that the Lie algebra generated by $H_{\mathrm{anti\text{-}TC}}$ and $J_z$ is isomorphic to the Lie algebra
generated by $\HTC$ and $J_z$, after conjugation by the same global rotation.
In particular, since $e^{i \pi J_x} J_z e^{-i \pi J_x}=-J_z$, the conserved charge for the anti-TC interaction is obtained by applying this rotation to the TC charge $Q$:
\begin{align}
    Q_{\mathrm{anti\text{-}TC}}
    \,:=\, e^{i\pi J_x} Q e^{-i\pi J_x}
    = a^\dag a - J_z+\frac{n}{2}\mathbb{I}.
\end{align}
Therefore, the charge sectors, the corresponding block decompositions, and the representation-theoretic analysis are mapped directly from the TC case to the anti-TC case by this conjugation.
Consequently, all characterization results established above for $\HTC$ apply without modification to $H_{\mathrm{anti\text{-}TC}}$, after replacing the conserved charge $Q$ by $Q_{\mathrm{anti\text{-}TC}}$ and conjugating the relevant operators by $e^{i\pi J_x}$.

\section{Centers, Commutator Subalgebras, and Commutator Subgroups}
\label{sec:commutator_subgroups}
In this section, we further discuss the relations among these groups and their corresponding Lie algebras, and show how these relations can be understood in terms of the notions of centers and commutator subgroups and subalgebras.

\subsection{The Lie algebra of PI, U(1)-invariant Hamiltonians}
By definition, the above groups satisfy
\begin{align}
    \mathcal{SV}_z\subset \V_z \subset \mathcal{W}_z \subset \Vsym,
 \label{eq:Lie_Groups}
\end{align}
where
\begin{align}
 \begin{split}
    \Vsym = \Big\{V:\,
    [V, Q]&=[V, \textbf{P}_{i,j}]=0 :\\
    &i\neq j\in \{1,\,\dotsc,\, n\}\Big\}
 \end{split}
\end{align}
is the group of \textit{all} PI, U(1)-invariant unitaries, where $\textbf{P}_{i,j}$ is the SWAP unitary acting on qubits $i$ and $j$.

Similarly, their corresponding Lie algebras satisfy
\begin{align}
    \sv_z \subset\vz \subset \wz \subset \g\,,
 \label{eq:Lie_Algebras}
\end{align}
where $\g$ is the Lie algebra corresponding  to $\Vsym$, which can be easily seen to be the set of all PI, U(1)-invariant skew-Hermitian operators,
\begin{align}
    \g := \Big\{&A\in \mathcal{L}\big((\C^2)^{\otimes n}\otimes\mathcal{L}^2(\R)\big): A=-A^\dag,\\
    &[A, Q]=[A, \textbf{P}_{i,j}]=0 :i\neq j\in \{1,\,\dotsc,\,n\} \Big\}\,. \nonumber
\end{align}
Using the irreducible representation (irrep) basis of the permutation group $\mathbb{S}_n$ and the eigen-subspaces of charge operator $Q$, we find a simple characterization of group $\Vsym$ and algebra $\g$ (see \cref{eq:g-char0,eq:g-char}).
In particular, subspaces with inequivalent irreps of $\mathbb{S}_n$ can be labeled by the eigenvalues of the angular momentum squared operator $J^2$, which has the spectral decomposition 
\begin{align}
    J^2=J_x^2+J_y^2+J_z^2 \,=\,\sum_{j=j_\text{min}}^{n/2} j(j+1) P_j\,,
\end{align}
where 
\begin{align}
    j(j+1) \,\, : \quad j=j_\text{min},\dotsc, \frac{n}{2} 
\end{align}
are the eigenvalues of $J^2$, with $j_\text{min}=1/2$ for odd $n$ and $j_\text{min}=0$ for even $n$.  
Similarly, the spectral decomposition of the charge operator $Q$ takes the form
\begin{align}
    Q \eq \sum_{q=0}^\infty q \Pi_q\,,
\end{align}
where eigenvalues $q$ are non-negative integers.

A useful implication of Schur-Weyl duality (see \cref{sec:schur_weyl}), is that every PI operator commutes with $J^2$.
That is,
\begin{align}
    \big[Q, J^2\big]=\big[\HTC, J^2\big]=\big[J_z, J^2\big]=0\,.
\end{align}
Since $Q$ and $J^2$ commute, their eigen-projectors also commute, which means
\begin{align}
    \Pi\qj \,:=\,P_j\Pi_q=\Pi_q P_j
 \label{eq:Piqj}
\end{align}
is a Hermitian projector.
Each projector $\Pi\qj$ is clearly PI and U(1)-invariant, and therefore $i\Pi\qj \in \mathfrak{g}$.
Furthermore, $i\Pi\qj$ commutes with every element of $\mathfrak{g}$, and hence belongs to the center of $\mathfrak{g}$ (i.e., the set of elements of $\mathfrak{g}$ that commute with all elements of $\mathfrak{g}$). 
Indeed, as we further explain in \cref{sec:CC_PI_unitaries} (see \cref{eq:center}), the center of $\mathfrak{g}$ is exactly the span of these projectors; namely,
\begin{align}
    \mathfrak{c}=\s_{\R}\Big\{i\Pi\qj\,:\,q\ge 0\,;\,\, j=j_\text{min},\dotsc,n/2 \Big\}\,.
\end{align}
These elements of $\mathfrak{g}$ correspond to Hamiltonians that generate \textit{relative phases} between subspaces carrying different $q$ and $j$, i.e., unitaries in the form
\begin{align}\label{Eq:phases}
    V=\sum_{q,j} \exp(i \theta_{q,j})\,\Pi_{q,j}\quad:\quad\theta_{q,j}\in[0,2\pi)\,,
\end{align}
which constitute exactly the center of $\Vsym$.

A general element of $\mathfrak{g}$ has a unique decomposition as 
\begin{align}\label{eq:dec0}
    A \eq A_\perp+\sum_{q,j} \Tr(A \Pi_{q,j})\frac{\Pi_{q,j}}{\Tr(\Pi_{q,j})}\,,
\end{align}
where the second term is simply the projection of $A$ onto the center $\mathfrak{c}$, while $A_\perp$ is orthogonal to $\mathfrak{c}$. That is,
\begin{align}
    \Tr(A_\perp\Pi_{q,j})=0 \qquad \forall \,\,q\geq0\,\,\,\text{and}\,\,
    \,j=\jmin,\dotsc,\frac{n}{2}\,.
\end{align}
Furthermore, $A_\perp$ belongs to the commutator subalgebra of $\mathfrak{g}$, denoted by $[\g,\g]$, which is the subalgebra generated by elements of the form $[h_1,h_2]$ with $h_1,h_2\in \g$. 
In other words, the decomposition in \cref{eq:dec0} corresponds to the decomposition
\begin{align}
   \g=[\g, \g]\oplus \mathfrak{c}\,,
\end{align}
where $\mathfrak{c}$ is the center and is orthogonal to the commutator sub-algebra. 
This commutator subalgebra is semi-simple; namely, it is an infinite direct sum of $\mathfrak{su}(d_n(q,j))$ for all $q$ and $j$, where $d_n(q,j)$ is the dimension associated with the sector labeled by $(q,j)$, given in \cref{eq:Hqjn_dimension}.

\subsection{Central constraints}
\label{sec:central_constraints}
We are interested in determining the subalgebra of $\g$ generated by the TC interaction and global $z$ rotations, namely $\vz=\mathfrak{alg}\{i\HTC, iJ_z\}$.
While the center $\mathfrak{c}$ of $\mathfrak{g}$ is infinite-dimensional, it turns out that $\vz$ contains only a one-dimensional subspace of $\mathfrak{c}$.

This is a consequence of the following more general fact: any subalgebra
\begin{align}
    \mathfrak{h} \eq \mathfrak{alg}_{\mathbb{R}}\{iH_1,\,\dotsc,\,iH_r\}
\end{align}
generated by a finite number of Hamiltonians $iH_1,\dotsc,iH_r\in \g$ contains only a finite-dimensional subspace of $\mathfrak{c}$. 
More precisely, the projection of $\mathfrak{h}$ onto $\mathfrak{c}$ is equal to the subspace spanned by the projections of $iH_1,\dotsc,iH_r$ onto $\mathfrak{c}$, that is,
\begin{align}
    \text{span}_{\mathbb{R}} \left\{\sum_{q,j} \Tr(iH_s \Pi_{q,j}) \frac{\Pi_{q,j}}{\Tr(\Pi_{q,j})}\,\,:\,\, s=1,\dotsc, r\right\}\,,
\end{align}
whose dimension is clearly bounded by $r$ \cite{PhysRevA.92.042309, marvian_sym_loc_2022}.

It turns out that the projections of $\HTC$ and $\overline{\HTC}$ onto this subspace are zero, i.e.,
\begin{align}
    \Tr\big(\HTC \Pi_{q,j}\big) \eq \Tr\big(\overline{\HTC} \Pi_{q,j}\big)=0
\end{align}
for all $q$ and $j$, whereas $J_z$ has a nonzero projection. 
It follows that the Lie algebra $\sv_z=\mathfrak{alg}\{i\HTC,i\overline{\HTC}\}$ does not contain any element of $\mathfrak{c}$, whereas $\vz=\mathfrak{alg}\{i\HTC,iJ_z\}$ contains only a one-dimensional subspace of the infinite-dimensional center $\mathfrak{c}$. 
More precisely, its projection onto $\mathfrak{c}$ corresponds to the projection of $iJ_z$ onto $\mathfrak{c}$, i.e., the one-dimensional subspace of operators proportional to
\begin{align*}
    \sum_{q,j} \Tr(i J_z \Pi_{q,j})\frac{\Pi_{q,j}}{\Tr(\Pi_{q,j})}\,.
\end{align*}
Indeed, it turns out that $\sv_z$ is centerless, i.e., it does not contain any nonzero element that commutes with all other elements of $\sv_z$ (see \cref{thm:Thm0} below).

\subsection{Commutator subgroups and semi-universality}\label{sec:semi-univ}
The commutator subalgebra is the Lie algebraic version of the notion of commutator subgroup.
For any group $\V$, its commutator subgroup is defined as the subgroup of $\V$ generated by elements of the form $V_1^{-1} V_2^{-1} V_1 V_2$ for $V_1, V_2\in\V$, i.e.
\begin{align}
    [\V,\V]\,:=\,\big\langle V_1^{-1} V_2^{-1} V_1 V_2\,:\, V_1, V_2\in\V\big\rangle\,.
\end{align}
If $\V$ is a connected Lie group with Lie algebra $\mathfrak{v}$, its commutator subgroup is also a connected Lie group whose associated Lie algebra is the commutator subalgebra $[\mathfrak{v},\mathfrak{v}]$ of $\mathfrak{v}$.  
For instance, 
\begin{align}
    \mathcal{S}\Vsym:=\Big[\Vsym, \Vsym\Big]\,,
\end{align}
is the commutator subgroup of $\Vsym$, and has a simple characterization given in \cref{eq:sv1}. This is the connected Lie group associated to the Lie algebra $[\mathfrak{g},\mathfrak{g}]$. \\

In what follows, to set aside the central constraints on relative phases discussed above and focus on other possible restrictions, we use the notion of \emph{semi-universality}, introduced in \cite{marvian_abelian_2024}. 
We say that a subgroup $\mathcal{W}\subset\V$ is \emph{semi-universal in} $\V$ if their commutator subgroups are equal, i.e.,
\begin{align}
    [\mathcal{W}, \mathcal{W}] = [\V, \V] \,.
\end{align}
In this way, the relative phases between sectors corresponding to different irreducible representations disappear within the commutator subgroup. 
From a Lie-algebraic perspective, this corresponds to removing the center of the Lie algebra under consideration and focusing on its commutator subalgebra.

\section{Summary of Main Result 1}\label{sec:summary}
The central constraints discussed above imply that the relevant nontrivial question is not whether $\mathfrak{g}$ is equal to $\vz$ or $\sv_z$; clearly, it is not.
Rather, the important question is whether $\sv_z$ is equal to $[\mathfrak{g},\mathfrak{g}]$, the commutator subalgebra of $\mathfrak{g}$.
In physical terms, this asks whether $\HTC$ and $J_z$ allow us to realize all PI, U(1)-invariant unitaries, apart from possible central constraints, namely relative phases between sectors with different $q$ and $j$.
\vspace{1em}

One of the key observations revealed by our analysis of TC interactions is that the answer is negative for $n\ge 3$ qubits; in fact,
\begin{align*}
    \sv_z \subsetneq [\mathfrak{g},\mathfrak{g}]\,.
\end{align*}
This failure is a consequence of the aforementioned accidental symmetry of $\HTC$.
At the same time, we find that this accidental symmetry is broken by the Hamiltonian $J_z^2$.
Indeed, we prove that, when restricted to any subspace with bounded charge $Q$, the Lie algebra generated by $i\HTC$, $iJ_z$, and $iJ_z^2$ contains $[\mathfrak{g},\mathfrak{g}]$. 

We postpone the characterization of $\sv_z$ to \cref{sec:full_char}, after introducing the irrep-basis decomposition and the notion of accidental symmetry; see \cref{prop:accidental_symmetry_matrix,prop:quasi_semi_universality}.
Here, we state the above relationship and also further clarify the relations among the commutator subalgebras of the different Lie algebras discussed above. 
For technical reasons, we introduce a  cutoff integer $\qmax \ge 0$, and 
\begin{align}
    \Pi^{\qmax}\,:=\,\sum_{q=0}^{q_{\text{max}}} \Pi_q\,,
\end{align}
which is the projector onto the subspace with charge less than or equal to the arbitrary maximum charge $\qmax\geq0$.
\begin{table*}[htp]
    \centering
    \renewcommand{\arraystretch}{2}
    \setlength{\tabcolsep}{6pt}
    \begin{tabular}{l|l|l}
        \textbf{5 qubits, $\qmax=7$} & \text{Decomposition into simple Lie algebra factors}  & \text{Dimension} \\\hline
        $[\vz,\vz]\Pi^{\qmax}$ & $\Big(\su(6)^{\oplus3} \oplus\su(4)^{\oplus4}\oplus\su(2)^{\oplus5}\Big)\oplus\Big(\su(5)\oplus\su(3)^{\oplus2}\Big)$ & $220$\\
        $[\g,\g]\Pi^{\qmax}$ & $[\vz,\vz]\Pi^{\qmax}\oplus\Big(\su(4)\oplus\su(2)^{\oplus2}\Big)$ & $241$\\
        &&\vspace{-12pt}\\
        \textbf{6 qubits, $\qmax=8$} & \\\hline
        $[\vz,\vz]\Pi^{\qmax}$ & $\Big(\su(7)^{\oplus3} \oplus\su(5)^{\oplus4}\oplus\su(3)^{\oplus5}\Big)\oplus\Big(\su(6)\oplus\su(4)^{\oplus2}\oplus\su(2)^{\oplus3}\Big)$ & $354$\\
        $[\g,\g]\Pi^{\qmax}$ & $[\vz,\vz]\Pi^{\qmax} \oplus \Big(\su(5)\oplus\su(3)^{\oplus2}\Big)$ & $394$
    \end{tabular}
    \caption{Comparison of the Lie algebras $[\g,\g]\Pi^{\qmax}$ and $\sv_z\Pi^{\qmax}$ for the examples of $n=5$ and $n=6$ qubits.
    Here we choose the cutoff $\qmax=n+2$, which ensures that all restrictions imposed by the accidental symmetry are present.
    For further details, see \cref{app:simple_lie_decomp}.}
    \label{tab:commutator_algebras}
\end{table*}

We prove the following theorem:
\begin{theorem}\label{thm:Thm0}
Consider the following Lie subalgebras of $\mathfrak{g}$, the Lie algebra of PI, U(1)-invariant skew-Hermitian operators:
\begin{align}\label{Eq:All}
 \begin{split}
    \sv_z&:=\mathfrak{alg}_{\mathbb{R}}\{i\HTC,i\overline{\HTC}\}\,,\\[4pt]
    \vz&:=\alg_{\R}\!\left\{i\HTC,\, i\Zhat\right\}\,,\\[4pt]
    \wz&:=\alg_{\R}\!\left\{i\HTC,\, i\Zhat,\, ia^{\dag}a\right\}\,,\\[4pt]
    \mathfrak{v}^+_z&:=\alg_{\R}\!\left\{i\HTC,\, i\Zhat, i\Zhat^2\right\}\,.
 \end{split}
\end{align}

The following statements hold:
\begin{enumerate}
\item \label{part:1} The Lie algebra $\sv_z$ is centerless and is equal to the commutator subalgebras of $\vz$ and $\wz$, that is
\begin{align}
    \sv_z= [\vz,\vz] = [\wz,\wz] \,.
 \label{eq:comm_subgroup_relations}
\end{align}
Furthermore,
\begin{align}
    \vz &\eq \sv_z\oplus \mathfrak{u}(1)\\[4pt]
    \wz &\eq \sv_z\oplus \mathfrak{u}(1)^{\oplus 2}\,,
\end{align}
where $\mathfrak{u}(1)$ and $\mathfrak{u}(1)^{\oplus 2}$  correspond to the centers of $\vz$ and $\wz$, respectively.
See also \cref{cor:lie_algebras}.
\item \label{part:2} For $n\ge 3$ qubits, $\sv_z$ is strictly contained in the
commutator subalgebra of $\mathfrak{g}$, i.e. 
\begin{align}
    \sv_z \subsetneq\, [\mathfrak{g}, \mathfrak{g}]\,,
\end{align}
whereas for $n=2$ qubits, they are equal.
\item \label{part:3.5} For sufficiently large charge cutoff, namely $\qmax \geq n+\lceil\frac{n}{2}-1\rceil$, the difference between dimensions of $[\g,\g]\Pi^{\qmax}$ and $\sv_z\Pi^{\qmax}$ is independent of $\qmax$, and is given by
\begin{align} \label{eq:dimension_diff}
 \begin{split}
    &\dim\big([\g,\g]\Pi^{\qmax}\big)
    -\dim\big(\sv_z\Pi^{\qmax}\big)\\[4pt]
    &\hspace{20pt}\eq \sum_{j=j_{\min}}^{n/2} \sum_{j'<j} \left((2j'+1)^2-1\right)\\[4pt]
    &\hspace{20pt}\eq \frac{1}{48}\Bigg(n(n-2)(n+2)(n+4)\\[2pt]
    &\hspace{70pt}+
    \begin{cases}
        0, & n \text{ even},\\[4pt]
        12n+3, & n \text{ odd}
    \end{cases}
    \Bigg)\,.
 \end{split}
\end{align}
\item \label{part:3} For any cutoff integer $\qmax \ge 0$, the projections of the commutator subalgebras of $\vz^+$ and $\mathfrak{g}$   onto the subspace with charge $q\leq\qmax$  are equal, i.e.,
\begin{align}
    \big[\vz^+,\vz^+\big]\Pi^{\qmax}=[\mathfrak{g}, \mathfrak{g}]\Pi^{\qmax}\,.
\end{align}
\item \label{part:4} When restricted to any single eigensubspace of the charge operator $Q$, with projector $\Pi_q$, the Lie algebra $\sv_z$ is equal to the commutator subalgebra of $\mathfrak{g}$; that is,
\begin{align}
    \sv_z \Pi_q \eq [\mathfrak{g}, \mathfrak{g}] \Pi_q\,.
 \label{eq:comm_subgroup_q}
\end{align}

\item \label{part:5} Similarly, when restricted to any single eigensubspace of the squared angular momentum operator $J^2$, with projector $P_j$, and to the subspace with charge at most $\qmax$, the projections of $\sv_z$ and the commutator subalgebra of $\mathfrak{g}$ are equal; that is,
\begin{align}\label{eq:semi_universal_j}
    \sv_z \Pi^{\qmax} P_j \eq [\mathfrak{g}, \mathfrak{g}] \Pi^{\qmax} P_j\,.
\end{align}
\end{enumerate}
\end{theorem}
\cref{fig:Lie} gives a schematic depiction of the relations among the different Lie algebras in this theorem.
Also, \Cref{tab:commutator_algebras} illustrates \cref{part:3.5}, for the cases of $n=5$ and $n=6$ qubits.
Furthermore, in \cref{sec:full_char}, we present a full characterization of $\sv_z$ and the other Lie algebras listed above when restricted to a truncated subspace with bounded charge $\qmax$ (see \cref{thm:full_characterization}).
We prove \cref{thm:full_characterization} and \Cref{part:1,part:2,part:3,part:3.5,part:4,part:5} of \cref{thm:Thm0} in \cref{sec:full_char_proof}.

\subsection{Remarks on \cref{thm:Thm0}}
A few remarks are in order: \\

\noindent \textbf{(1)} The charge cutoff $q_{\text{max}}$ is imposed only for technical reasons: it ensures that the relevant Hilbert space is finite-dimensional and avoids the difficulties associated with infinite-dimensional spaces and Lie algebras. 
We conjecture that this cutoff is not necessary.

\vspace{0.5em}
\noindent \textbf{(2)} The accidental symmetry only relates pairs of sectors with distinct labels $(q,j)$ and $(q',j')$, with $q\neq q'$ and $j\neq j'$.
In particular, for each pair of angular momenta $j$ and $j'$, in the interval $0<j'<j\le n/2$, there exists a pair of charges $q$ and $q'$ given by \cref{eq:q_q'}, such that the unitaries realized by $\HTC$ and $\overline{\HTC}$ in the sectors labeled by $(q,j)$ and $(q',j')$ are identical.
This explains \cref{eq:dimension_diff}, which gives the difference between the dimensions of $[\g,\g]\Pi^{\qmax}$ and $\sv_z\Pi^{\qmax}$.
For sufficiently large $\qmax$, all constraints imposed by the accidental symmetry are present.
For each pair $j'<j$, the $(2j'+1)\times(2j'+1)$ special unitary realized in the sector with charge $q'$ and angular momentum $j'$ is also realized in the corresponding subspace of the sector with charge $q$ and angular momentum $j$.
Therefore, each such correlated pair removes exactly $(2j'+1)^2-1$ real parameters from $[\g,\g]$, which do not appear in $\sv_z$.

This also explains  \Cref{part:4,part:5} of the theorem: the constraint imposed by the accidental symmetry disappears once one restricts either to a single eigensubspace of $Q$ or to a single eigensubspace of $J^2$, because if either $q$ or $j$ is fixed, no such paired sectors remain within the restricted subspace.

\vspace{0.5em}
\noindent \textbf{(3)} In the case $j=0$, both $\mathfrak{sv}_z P_{j}$ and  $[\g,\g] P_{j}$ vanish.  First, the $j=0$ subspace is annihilated by $\HTC$ and $\overline{\HTC}$, as well as by all collective spin operators $J_w$ with $w=x,y,z$.
Hence, the restrictions of $\mathfrak{sv}_z$ and $\vz$ to this sector are trivial.

By contrast, a general PI, U(1)-invariant Hamiltonian, i.e., an element of $\mathfrak{g}$, need not vanish on the $j=0$ sector.
Since this sector carries the trivial representation of $\SU(2)$, permutation invariance implies that, on this sector, such a Hamiltonian acts on the qubit degrees of freedom as $P_0$.
Moreover, U(1)-invariance implies that its bosonic part commutes with $a^\dag a$.
Therefore, when projected to  the $j=0$ sector, any element  $A\in\g$  takes the form
\begin{align}
    A P_0 \eq i\left(\sum_k a_k \pure{k}_{\text{osc}}\right) P_0\,,
\end{align}
where $\ket{k}_{\text{osc}}$ is the eigenvector of $a^\dag a$ with eigenvalue $k$, and $a_k\in\mathbb{R}$.

Thus, although when projected to $j=0$ sector, elements of $\mathfrak{g}$ can be non-zero, they form an Abelian subalgebra, and therefore do not contribute to the commutator subalgebra:
\begin{align}\label{eq:dsa}
    \mathfrak{sv}_zP_0=[\mathfrak{g},\mathfrak{g}]P_0 = 0\,.
\end{align}

\vspace{0.5em}
\noindent \textbf{(4)} As we discuss further below in \cref{Sec:2qubits}, the case of $n=2$ qubits is special. In this case, the accidental symmetry imposes no nontrivial constraint, since there is only one nontrivial angular-momentum sector. Indeed, for $n=2$ the only angular-momentum sectors are $j=0$ and $j=1$, and the $j=0$ sector is dynamically trivial in the sense explained above. Consequently, for any charge cutoff $\qmax$,
\begin{align}
    \mathfrak{sv}_z \Pi^{\qmax}
    = \mathfrak{sv}_z P_1 \Pi^{\qmax}
    = [\mathfrak{g},\mathfrak{g}] P_1 \Pi^{\qmax}
    = [\mathfrak{g},\mathfrak{g}]\Pi^{\qmax}\,.
\end{align}
Here, the first equality follows from \cref{eq:dsa}, the second from \cref{eq:semi_universal_j}, and the third again from \cref{eq:dsa}.

\vspace{0.5em}
\noindent \textbf{(5)} For $n\geq 3$ qubits, if
$q_{\text{max}}\geq \big\lfloor n/2 \big\rfloor+3$, which is required for the accidental symmetry to be relevant (see \cref{eq:acc_sym_minimal}), then the restriction of $\sv_z$ to the truncated charge subspace remains strictly smaller than the corresponding restriction of the commutator subalgebra $[\mathfrak{g},\mathfrak{g}]$. Namely,
\begin{align}
    \sv_z \Pi^{\qmax}
    \subsetneq
    [\mathfrak{g}, \mathfrak{g}]\Pi^{\qmax}\,.
\end{align}

\vspace{0.5em}
\noindent \textbf{(6)}  We emphasize that $\vz$ and $\wz$ have the same commutator subalgebra, namely $\sv_z$, and differ only in their centers. 
Specifically, $\vz$ has a one-dimensional center, whereas $\wz$ has a two-dimensional center. 
These centers are spanned by the projections of $J_z$ onto $\mathfrak{c}$, and of $J_z$ and $a^\dag a$ onto $\mathfrak{c}$, respectively.
Furthermore,
\begin{align}\label{eq:Lie_vw}
 \begin{split}
    \wz 
    \,:=&\,\, \alg_{\R}\big\{i\HTC, i\Zhat, ia^{\dag}a\big\} \\[4pt]
    =&\,\,\vz \oplus \s_{\R}\left\{i \left(Q-\tfrac{n}{2}\mathbb{I}\right)\right\} \\[4pt]
    =&\,\,\vz \oplus \s_{\R}\{ia^{\dag}a\}.
 \end{split}
\end{align}
The second equality follows from the identity $a^\dag a=Q-J_z-\frac{n}{2}\mathbb{I}$, together with the fact that $\vz$ contains $iJ_z$ and that $Q-\frac{n}{2}\mathbb{I}$ commutes with $\vz$, and hence lies in the center of $\wz$. 
The third equality follows similarly from the fact that $\vz$ contains $iJ_z$.

\vspace{0.5em}
\noindent \textbf{(7)} The family of unitaries realized by Hamiltonians $\HTC$ and $\overline{\HTC}$ is $\SV_z$, which is the connected Lie group associated to $\mathfrak{sv}_z$.
\Cref{part:1} of \cref{thm:Thm0} implies that
\begin{align}
 \begin{split}
    \SV_z \,:=&\,\,\left\langle e^{-i t \HTC},\,e^{-i t \overline{\HTC}}\,:\, t\in\mathbb{R}\right\rangle\\[4pt]
    =&\,\,\big[\V_z,\V_z\big] \,=\,\big[\mathcal{W}_z,\mathcal{W}_z\big]\,,
 \end{split}
 \label{eq:Lie_Group_Identities}
\end{align}
Additionally, \cref{part:2} of the theorem implies that for $n\ge 3$ ,
\begin{align}
    \SV_z \,\subsetneq\, \mathcal{S}\Vsym \quad:\quad(n\geq3)\,.
\end{align}
Furthermore, \cref{part:3} of the theorem implies that for any cutoff integer $\qmax\ge 0$,
\begin{align}
    \mathcal{SV}_z^+\Pi^{\qmax}=\mathcal{S}\Vsym \Pi^{\qmax}\,,
\end{align}
where
\begin{align}
    \mathcal{SV}_z^+=[\mathcal{V}_z^+, \mathcal{V}_z^+]\,.
\end{align}
\begin{table}[h]
\centering
    \renewcommand{\arraystretch}{2.0}
    \setlength{\tabcolsep}{6pt}
    \begin{tabular}{l|l}
        \textbf{Lie Algebra} & \textbf{Lie Group} \\\hline
         $\g$ & $\Vsym$
         \\\hline
         $\vz = \mathrm{alg}_{\mathbb{R}}\big\{ i\HTC,\, i\Zhat \big\}$ & $\mathcal{V}_z$\\\hline
         $\begin{matrix*}[l]
            \sv_z &= \mathrm{alg}_{\mathbb{R}}\big\{ i\HTC,\, i\overline{\HTC} \big\} \\[-4pt]
            &= [\vz,\vz] = [\wz,\wz]
            \end{matrix*}$ & $\begin{matrix*}[l]
                \mathcal{SV}_z &= [\V_z,\V_z] \\[-4pt]
                &= [\mathcal{W}_z,\mathcal{W}_z]
            \end{matrix*}$\\\hline
         $\vz^+ = \mathrm{alg}_{\mathbb{R}}\big\{ i\HTC,\, i\Zhat ,\,i\Zhat^2\big\}$& $\mathcal{V}^+_z$
         \\\hline
         $\wz = \mathrm{alg}_{\mathbb{R}}\big\{ i\HTC,\, i\Zhat, ia^\dag a \big\}$ & $\mathcal{W}_z$ \\\hline
    \end{tabular}
\caption{Summary of the Lie groups describing the sets of unitaries that can be realized by Hamiltonians $\HTC$, $\overline{\HTC}$, $J_z$, and $a^\dag a$, together with the corresponding Lie algebras. 
All these groups are subgroups of $\Vsym$, the group of all PI, U(1)-invariant unitaries on the combined Hilbert space of $n$ qubits and the harmonic oscillator.}
\label{tab:Lie_Algebras}
\end{table}

\subsection{Special case of $n=2$ qubits coupled to a bosonic mode}\label{Sec:2qubits}
Our results reveal that, in the special case of $n=2$ qubits, which is particularly important for applications in the context of quantum computing gates, the accidental symmetry described in \cref{sec:accidental_sym} does not appear.
In other words, 
\begin{align}
    \sv_z\Pi^{\qmax}=[\g,\g]\Pi^{\qmax} \quad:\quad(n=2)\,,
\end{align}
for arbitrary $\qmax\ge 1$, which means, from \cref{cor:lie_algebras},
\begin{align}
    \vz\Pi^{\qmax} = \mathfrak{su}(2)\oplus  \mathfrak{su}(3)^{\oplus (\qmax-1)} \oplus \mathfrak{u}(1)\,.
\end{align}
In particular, the $\mathfrak{su}(2)$ factor corresponds to the traceless Hamiltonians in the 2D subspace spanned by
\[
    \ket{11}\otimes\ket{1}_{\text{osc}} \quad \text{and}\quad  \ket{\Psi^+}\otimes\ket{0}_{\text{osc}}\,.
\]
Here, $\ket{0}$ and $\ket{1}$ denote the computational basis states of a qubit, i.e., the eigenvectors of $\sigma_z$ with eigenvalues $+1$ and $-1$, respectively.
Furthermore,
\begin{align}
    \ket{\Psi^{\pm}}=\frac{\ket{01}\pm \ket{10}}{\sqrt{2}} \, ,
\end{align}
are the states with eigenvalue $m=0$ of $J_z$ in the sectors with $j=1$ for $\Psi^+$ and $j=0$ for $\Psi^-$.
Furthermore, each $\mathfrak{su}(3)$ factor corresponds to the Hamiltonians in the 3D subspaces spanned by
\begin{align*}
    \ket{11}\otimes\ket{q}_{\text{osc}}\,,\,\,\ket{\Psi^+}\otimes\ket{q-1}_{\text{osc}}\,,\,\, \ket{00}\otimes\ket{q-2}_{\text{osc}}\,,
\end{align*}
for $q\geq2$.

An interesting subalgebra of $\sv_z$ is obtained by restricting to Hamiltonians that conserve $J_z$, or equivalently, conserve $a^\dag a$. Such Hamiltonians can be written as
\begin{align} \label{eq:2qub_Ham}
    H \eq \sum_{k\geq0}^{k_{\max}} \Big(a_k \pure{00}&+b_k\pure{\Psi^+}\\\nonumber
    &\hspace{10pt}+c_k\pure{11} \Big)\otimes |k\rangle\langle k|_\text{osc}\,,
\end{align}
where $a_k$ , $b_k$, and $c_k$ satisfy
\begin{align}
    c_k + b_{k-1} + a_{k-2} &\eq 0\,,
\end{align}
using the convention that the coefficients are zero outside the boundary $0\leq k\leq k_{\text{max}}$.
This condition follows directly from the fact that $\sv_z$ is traceless within each sector $\H\qj$, i.e. $\Tr(\Pi_q P_j H)=0$.
Hamiltonians of the form in \cref{eq:2qub_Ham} include, for example,
\begin{align}
    \pure{\Psi^+}\otimes\pure{0}_{\text{osc}} - \pure{11}\otimes\pure{1}_{\text{osc}}\,,
\end{align}
which can generate entanglement between the two qubits, e.g. by initializing the system in state $\ket{01}\otimes\ket{0}_{\text{osc}}$.

In our companion paper \cite{circuit_paper}, we present methods for implementing several useful 2-qubit entangling gates,  with corresponding Hamiltonians of the form in \cref{eq:2qub_Ham}, via the TC interaction, using an oscillator initialized in its vacuum state $\vac$ as an ancillary system.

\section{Restriction to the Symmetric Subspace}
\label{sec:sym_subspace}
As mentioned below \cref{thm:Thm0}, the constraints imposed by the accidental symmetry do not appear when one restricts to a single angular momentum sector. 
Therefore, to simplify the problem, we begin our study of $\V_z$, the group of unitaries that can be realized using the TC interaction and global $z$ rotations on the qubits, with the particularly important case $j=n/2$.
Equivalently, we restrict attention to joint qubit-oscillator states whose qubit component is supported on the symmetric subspace $\Hsym\subset(\mathbb{C}^2)^{\otimes n}$, and later, in \cref{sec:beyond_sym}, we go beyond this and discuss the general case.

More precisely, $\Hsym$ is the $(n+1)$-dimensional subspace of vectors that remain invariant under all permutations of qubits.
The tensor product of this subspace with the Hilbert space of the oscillator decomposes as
\begin{align}\label{eq:dec1}
    \Hsym\otimes  \mathcal{H}_\text{osc}\cong \bigoplus_{q=0}^{\infty} \mathcal{H}_{q}\,,  
\end{align}
where $\mathcal{H}_{q}$ is the eigen-subspace of charge operator $Q$ with eigenvalue $q$ within $\Hsym\otimes  \mathcal{H}_\text{osc}$. 
Recall that in the angular momentum $j$ subspace, $J_z$ is non-degenerate with eigenvalues
\begin{align}
    m=-j,\,-j+1,\,\dotsc,\,j\,.
\end{align}
Furthermore, since eigenvalues of $a^\dag a$ are non-negative integers, in a sector with eigenvalue $q$ of $Q=a^\dag a+J_z+n/2$, the maximum allowed value of $m$ is $q-n/2$, which means
\begin{align}\label{eq:bound1}
    -j\leq m \leq \min\left\{j,\,q-\frac{n}{2}\right\}\,.
\end{align}
For $j=n/2$, this gives a subspace of dimension
\begin{align}
    \text{dim}(\mathcal{H}_{q})=\min\{n+1,\,q+1\}\,,
 \label{eq:dimHq_j=n/2}
\end{align}
which grows linearly with charge $q$, from one for $q=0$ until it reaches a maximum value of $n+1$.

The permutational symmetry of Hamiltonian $H(t)$ in \cref{eq:totH} implies that it leaves $\Hsym\otimes \mathcal{H}_\text{osc}$ invariant. 
Furthermore, because $H(t)$ conserves charge $Q$, it is block-diagonal with respect to the decomposition in \cref{eq:dec1}. 
This, in turn, implies that restrictions of unitaries in $\V_z$ to $\Hsym\otimes\mathcal{H}_\text{osc}$ are block-diagonal with respect to this decomposition:
\begin{align}
 \begin{split}
   \big(P_{j=n/2}\big)\V_z  \,\subset\, \big(P_{j=n/2}\big)\Vsym  &\,\cong\, \bigoplus_{q=0}^{\infty} \UU\big(\mathcal{H}_{q}\big)\,,
 \end{split}
\end{align}
where $P_{j=n/2}$ is the projector to the qubits' symmetric subspace and $\UU(\mathcal{H}_{q})$ denotes the group of unitary transformations on $\mathcal{H}_{q}$.

Does the above inclusion hold as equality?
The previous work of Keyl et al. in \cite{Keyl_2014_control} establishes this when restricted to any eigen-subspace of the charge operator $Q$, i.e., they showed
\begin{align}
 \label{eq:Vz_sym_sector}
    \Pi_q\big(P_{j=n/2}\big) \V_z\,\cong\,\UU\big(\mathcal{H}_{q}\big)\, ,
\end{align}
where $\Pi_q(P_{j=n/2})=(P_{j=n/2})\Pi_q$ is the projector to the eigen-subspace of $Q$ inside $\Hsym\otimes\mathcal{H}_\text{osc}$ with eigenvalue $q$.

However, in principle, choosing a desired unitary in a given charge sector $q$ may impose constraints on unitaries realizable in other sectors. 
Indeed, it is known that in general, the phases of unitaries implemented in different sectors cannot be chosen independently \cite{PhysRevA.92.042309, Keyl_2014_control, marvian_sym_loc_2022}. 
More precisely, these constraints can be expressed in terms of the determinants of the unitaries realized in the various charge sectors (see \cref{eq:sym_phase_constraint} below). 
From a Lie algebraic perspective, this corresponds to restrictions on the center of the dynamical Lie algebra (see also \cite{marvian_sym_loc_2022}).

As discussed in \cref{sec:semi-univ}, we can isolate such phase-related constraints and focus instead on other structural limitations by studying the commutator subgroup of $\V_z$, whose projection to any particular sector is (from \cref{eq:Vz_sym_sector}) the special unitary group on that sector:
\begin{align}
    \Pi_q(P_{j=n/2})[\V_z,\V_z] \,\cong\, \SU(\H_q)\,.
\end{align}
Still, even in this setting, selecting a desired unitary in one sector can, in principle, constrain the unitaries realizable in other sectors, and indeed \cite{Keyl_2014_control} suggests that this might be the case. 
However, in the following lemma, we show that when truncated to any finite charge sector cutoff $q_{\text{max}}\ge 0$, there are no additional constraints on the set of realizable unitaries.  
In the following lemma, with $\Pi^{\qmax}=\sum_{q=0}^{q_{\text{max}}} \Pi_q$ defining the projector to the subspace of $\H_{\text{qubits}}\otimes\mathcal{H}_\text{osc}$ of all states with charge $q\le q_{\text{max}}$,
\begin{align}
    \PiSym:=&\,\,\Pi^{\qmax} \big(P_{j=n/2}\big)\,\,=\,\,\big(P_{j=n/2}\big)\Pi^{\qmax}
\end{align}
defines the projection of $\Pi^{\qmax}$ to the symmetric subspace.
This projector truncates the infinite-dimensional Hilbert space to a finite-dimensional subspace, whose dimension can be arbitrarily large depending on $q_{\text{max}}$.

\begin{lemma}\label{lem:sym_semi_universality}
For any integer $q_{\text{max}}>0 $, the projection of $\V_z $ to the subspace of $\Hsym\otimes  \mathcal{H}_\text{osc}$ with charge $q\le q_{\text{max}}$ contains the subgroup $\bigoplus_{q=0}^{q_{\text{max}}} \SU\big(\mathcal{H}_{q}\big)$, i.e., 
\begin{align}
    \bigoplus_{q=0}^{q_{\text{max}}} \SU\big(\mathcal{H}_{q}\big) \,\subset\, \V_z \PiSym\,. 
\end{align}
Furthermore, for any choice of unitaries $v_q\in \UU\big(\mathcal{H}_{q}\big): q=0, \dotsc, q_{\text{max}}\,$, $\V_z$ contains a unitary $V$ that realizes them, such that
\begin{align}
    V\PiSym = \bigoplus_{q=0}^{\qmax} v_q
\end{align}
if, and only if, there exists $\theta_z\in[-2\pi,2\pi)$, such that $\theta_q:=\arg\det(v_q)$ satisfies
\begin{align}
 \begin{split}
    \theta_{q} &\eq \theta_z\Tr\big(\Pi_qP_{j=n/2}J_z\big) \\[8pt]
    &\eq\begin{cases}
        \dfrac{(q+1)(q-n)\theta_z}{2} & :\,\,q\leq n \\[6pt]
        0 & :\,\,q>n\,,
    \end{cases}
 \end{split}
 \label{eq:sym_phase_constraint}
\end{align}
where these equations hold modulo $2\pi$.
\end{lemma}
This implies that when projected to $\PiSym$, the only restrictions on $\V_z$ are on the relative phases between sectors with different charges.
Following the terminology introduced in \cite{HLM_2024_SU(d)}, and discussed \cref{sec:semi-univ}, in this means $\V_z$ is semi-universal in $\Vsym$.
Equivalently, in terms of Lie algebras this means
\begin{align}
 \begin{split}
    \sv_z \PiSym &\eq [\g,\g]\PiSym \\[4pt]
    &\cong\bigoplus_{q=0}^{q_{\text{max}}   } \mathfrak{su}\big(\min\{n+1, \,q+1\}\big)\,.
 \end{split}
\end{align}
Is the assumption that the state of the qubits is restricted to the symmetric subspace important for this property? 
We return to this question in \cref{sec:beyond_sym}, where \cref{thm:full_characterization} gives a full generalization of \cref{lem:sym_semi_universality} to arbitrary qubit states.
Indeed, \cref{lem:sym_semi_universality} follows as the special case of \cref{thm:full_characterization}, when restricted to the symmetric subspace.

\subsection{State convertibility in the symmetric subspace}
\label{sec:state_convertability}

An immediate corollary of the above result is the following:
\begin{corollary}\label{cor:state_prep}
Consider a pair of states $|\Psi\rangle, |\Phi\rangle \in \Hsym\otimes  \mathcal{H}_\text{osc}$, which   both have support on a finite number of eigen-subspaces of  $Q$.  There exists a unitary transformation $V\in \V_z$ that converts $|\Psi\rangle$ to $|\Phi\rangle$, such that $V|\Psi\rangle=|\Phi\rangle$ if, and only if,
\begin{align}\label{eq:cons}
    \forall q\geq0\,\,:\,\,\langle\Psi|\Pi_q|\Psi\rangle=\langle\Phi|\Pi_q|\Phi\rangle\,,
\end{align}
where $\Pi_q$ is the projector to the eigen-subspace of $Q$ with eigenvalue $q$.
\end{corollary}
Note that this condition can be equivalently stated as
\begin{align}
    \langle\Psi|e^{i\theta Q}|\Psi\rangle=\langle\Phi|e^{i\theta Q}|\Phi\rangle\,\,:\,\, \theta\in[0,2\pi)\,.
\end{align}
Also, note that \cref{eq:cons} is necessary for the existence of a general PI, U(1)-invariant unitary $V\in \Vsym$ that realizes the desired state conversion. Therefore, this corollary means that the restrictions imposed by \cref{eq:sym_phase_constraint} do not forbid pure state conversions inside $\Hsym\otimes  \mathcal{H}_\text{osc}$.

\begin{proof}
Necessity is an immediate consequence of the U(1) symmetry; since every $V\in\mathcal{V}_z$ commutes with $Q$, it also commutes with each spectral projector $\Pi_q$.
We now prove sufficiency.
By assumption, the two projected vectors $\Pi_q\ket{\Psi}$ and $\Pi_q\ket{\Phi}$ have equal norms for every $q$.
Since the states have finite support over the charge sectors, only finitely many of these vectors are nonzero.
First consider the one-dimensional sector with $q=0$, which is spanned by $\ket{1}^{\otimes n}\ket{0}_{\rm osc}$.
If the common norm $\|\Pi_0\ket{\Psi}\|=\|\Pi_0\ket{\Phi}\|$ is nonzero, then the two projected vectors differ only by a phase.
We choose a global $z$ rotation $R_z(\theta_z)$ so that this phase is matched, namely so that
\begin{align}
    \Pi_0 R_z(\theta_z)\ket{\Psi}
    =
    \Pi_0\ket{\Phi}\,.
\end{align}
If the common norm is zero, we simply take $\theta_z=0$.
Define $\ket{\Psi'}:=R_z(\theta_z)\ket{\Psi}$.
Since $R_z(\theta_z)$ preserves each charge sector, we still have
\begin{align}
    \|\Pi_q\ket{\Psi'}\|
    =
    \|\Pi_q\ket{\Phi}\|
    \qquad \forall q\,.
\end{align}

For every sector $q\ge 1$ with nonzero support, the Hilbert space $\mathcal{H}_q:=\Pi_q(\Hsym\otimes\mathcal{H}_{\rm osc})$ has dimension larger than one.
Therefore, for each such $q$, there exists a special unitary $v_q\in\SU(\mathcal{H}_q)$, such that
\begin{align}
    v_q\,\Pi_q\ket{\Psi'}
    =
    \Pi_q\ket{\Phi}\,.
\end{align}
Indeed, in any dimension $d>1$, two vectors in $\mathbb{C}^d$ can be related by an element of $\SU(d)$ if and only if they have the same norm.\footnote{The fact that there exists a general unitary in $\UU(d)$ with this property is trivial.
To see why there exists also an element of $\SU(d)$ with the desired property, note that there exists special unitary $U$ such that $U |\eta\rangle=e^{i\gamma} |\eta\rangle$ for arbitrary phase $|\gamma\rangle$, namely $U=e^{-i\gamma / (d-1)}\exp(i\gamma d/(d-1)|\eta\rangle\langle \eta)$.} 
By \cref{lem:sym_semi_universality}, $\mathcal{V}_z$ contains a unitary whose action on each of the finitely many relevant sectors is given by the corresponding $v_q$, while its action on the remaining sectors is arbitrary.
Combining this unitary with the initial $z$ rotation, we obtain a unitary $V\in\mathcal{V}_z$ satisfying 
$V\ket{\Psi}=\ket{\Phi}$.
This proves sufficiency and completes the proof.
\end{proof}

\subsection{Application: qubit-oscillator SWAP gate and preparing arbitrary oscillator states}
\label{sec:Vswap}
Here, we demonstrate how the result from the previous section leads to a novel method for preparing the oscillator in desired states.
In particular, \cref{lem:sym_semi_universality} implies that there exists a \textit{fixed} unitary $\Vswap\in\V_z$ implementing the map
\begin{align}
 \begin{split}
    \Vswap&\left(\left|j=\frac{n}{2},m\right\rangle\vac\right) \\[4pt]
    &\hspace{14pt}=\,\left|j=\frac{n}{2},-\frac{n}{2}\right\rangle \left|m+\frac{n}{2}\right\rangle_{\text{osc}}\,,
 \end{split}
 \label{eq:rq1}
\end{align}
for each $m=-n/2,\dots,n/2$ simultaneously, where $|j=\frac{n}{2},m\rangle$ is the eigenstate of $J^2$ and $J_z$  with eigenvalues $\frac{n}{2}(\frac{n}{2}+1)$ and $m$, respectively.
Note that the final qubit state, $|j=\frac{n}{2},-\frac{n}{2}\rangle =|1\rangle^{\otimes n}$, is independent of $m$.
Therefore, any such unitary maps an arbitrary state $\ket{\psi}$ in the symmetric subspace of $n$ qubits to a corresponding state $\ket{\Psi}$ of the oscillator, according to
\begin{align}\label{eq:SWAP}
    \Vswap\big(\ket{\psi}\otimes\vac\big) \eq |1\rangle^{\otimes n}\otimes\ket{\Psi}_{\text{osc}}\, ,
\end{align}
as illustrated by the circuit in \cref{fig:QubOsc_swap}.
\begin{figure}[htp]
    \centering
    \includegraphics[width=0.75\linewidth]{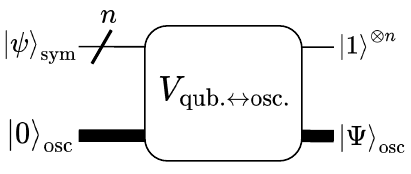}
    \caption{Qubit-Oscillator SWAP unitary. See \cref{eq:SWAP}.}
    \label{fig:QubOsc_swap}
\end{figure}

In particular, for state $\ket{\psi}\in\mathcal{H}_{\text{sym}}$, decomposed as
\begin{align}
    \ket{\psi}=\sum_{m=-n/2}^{n/2} \psi_m \left|j=\frac{n}{2},m\right\rangle \,,
\end{align}
the corresponding oscillator state $\ket{\Psi}\in\mathcal{H}_{\text{osc}}$ is
\begin{align}
    \ket{\Psi}_{\text{osc}}=\sum_{k=0}^{n} \psi_{k-\frac{n}{2}} \ket{k}_{\text{osc}}\,,
\end{align}
where $\ket{k}_{\text{osc}}$ is the eigenvector of the number operator $a^\dag a$ with eigenvalue $k$.

To see how \cref{lem:sym_semi_universality} implies the existence of such a unitary transformation $\Vswap$, note that for all $m = -n/2, \ldots, n/2$, the vectors $|j=\frac{n}{2}, m\rangle \otimes \vac$ and $|j=\frac{n}{2},-\frac{n}{2}\rangle \otimes |m + \frac{n}{2}\rangle_{\text{osc}}$ on the two sides of \cref{eq:rq1} belong to the same charge sector, $q=m+\frac{n}{2}$.
Furthermore, these vectors are orthogonal to each other, except for $m = -n/2$, where they are identical. 
It follows that for each $m$, there exists an element of $\SU(\mathcal{H}_q)$ that maps the state $|j=\frac{n}{2}, m\rangle \otimes \vac$ to $|j=\frac{n}{2},-\frac{n}{2}\rangle  \otimes |m + \frac{n}{2}\rangle_{\text{osc}}$. 
Then, by \cref{lem:sym_semi_universality}, we conclude that there exists an element of $\V_z$ that implements this transformation. 
This protocol is exact and consists solely of unitary time evolution under Hamiltonians $\HTC$ and $J_z$.
See \cite{circuit_paper} for an explicit construction of the 2-qubit version of $\Vswap$.

One application of this technique is state preparation for the bosonic mode. 
In \cite{circuit_paper}, we show that starting from the state $|0\rangle^{\otimes n}$ using $\HTC$ along with global $z$ and $x$ fields, we can generate any state in the symmetric subspace $\mathcal{H}_{\text{sym}}$. 
When combined with the universal qubit-oscillator SWAP gate described above, this means that the oscillator can be prepared in an arbitrary state with up to $n$ excitations.

\section{PI, U(1)-invariant unitaries beyond the symmetric subspace} 
\label{sec:beyond_sym}
\subsection{Irreducible representations of the permutation group via Schur-Weyl duality}
\label{sec:schur_weyl}

To go beyond the symmetric subspace, we first recall some useful facts about the representation theory of $\mathbb{S}_n$, the group of permutations of $n$ objects. 
According to the Schur-Weyl duality, there is a canonical one-to-one correspondence between the irreps of SU(2) and  $\mathbb{S}_n$ that appear in $(\mathbb{C}^2)^{\otimes n}$.
For instance, the symmetric subspace $\Hsym$ corresponds to the highest angular momentum $j=n/2$.
More generally, the subspaces with inequivalent irreps of $\mathbb{S}_n$ can be labeled by the eigenvalues of the total angular momentum (also known as the Casimir) operator $J^2=J_x^2+J_y^2+J_z^2$, which are in the form  $j(j+1)$ for $j=j_\text{min},j_\text{min}+1, \dotsc, n/2$, where $j_\text{min}=1/2$ when $n$ is odd and is zero when $n$ is even.

Then, under the action of permutations, the total Hilbert space of $n$ qubits decomposes as
\begin{align}\label{eq:dec-n}
    (\C^{2}&)^{\otimes n} \,\cong\, 
    \bigoplus_{j=\jmin}^{n/2}\big(\mathbb{C}^{M(n,j)} \otimes \mathbb{C}^{2j+1}\big)\,,
\end{align}
where the permutation group $\mathbb{S}_n$ acts irreducibly on $\mathbb{C}^{M(n,j)}$, with
\begin{align}
    M(n,j) \,:=\, \left(\begin{matrix}n\\\frac{n}{2}-j\end{matrix}\right)\frac{2j+1}{\frac{n}{2}+j+1}\,\,,
 \label{eq:spin_j_multiplicity}
\end{align}
and this irrep appears with multiplicity $2j+1$ \cite{MLH_2024,Bartlett_etal_2007}.
The group of unitaries $U^{\otimes n}: U\in \SU(2)$  acts irreducibly on the corresponding subsystem $\mathbb{C}^{2j+1}$, as the spin $j$ representation of SU(2). 
Then, Schur's lemma implies that a general PI operator $A$ is block-diagonal with respect to this decomposition, and inside each $j$-sector acts trivially on $\mathbb{C}^{M(n,j)}$, whereas it can be arbitrary on $\mathbb{C}^{2j+1}$, i.e., $A=\bigoplus_{j} \mathbb{I}_{M(n,j)}\otimes A_j$. 

This decomposition, in turn, implies that with respect to the group $\Uo\times\Sn$, the joint Hilbert space of $n$ qubits and the oscillator decomposes as 
\begin{align}
    (\C^{2})^{\otimes n}&\otimes\mathcal{L}^2(\R) \,\cong\, \bigoplus_{j=\jmin}^{n/2}\C^{M(n,j)}\otimes \Bigg(\bigoplus_{q=n/2-j}^{\infty}\H\qj\Bigg)
    \nonumber\\
    \nonumber\\
    &=\bigoplus_{q=0}^{\infty}\,\bigoplus_{j=\max(\jmin,\,\frac{n}{2}-q)}^{n/2}\left(\C^{M(n,j)} \otimes\H\qj\right)
    \,,
\label{eq:full_block_structure2}
\end{align}
where $\H\qj$ is the eigen-subspace of $Q=a^\dag a+J_z+\frac{n}{2}$ within $\mathbb{C}^{2j+1}\otimes \mathcal{L}^2(\R)$ with eigenvalue $q$, and has dimension
\begin{align}
 \begin{split}
    d_n(q,j) :=&\,\, \text{dim}(\H\qj)\\[4pt]
    =&\,\, \min\left\{2j+1\,,\,\,q+1+j-\frac{n}{2}\right\}\,.
 \end{split}
 \label{eq:Hqjn_dimension}
\end{align}
This is the generalization of \cref{eq:dimHq_j=n/2} to arbitrary $j$, beyond the symmetric subspace ($j=n/2$), and similar to \cref{eq:dimHq_j=n/2} follows immediately from \cref{eq:bound1}.

Note that when $d_n(q,j)=2j+1$, the space $\mathcal{H}_{q,j}$ carries the full irrep of SU(2) with angular momentum $j$. We refer to such sectors as \emph{filled} sectors.
On the other hand, when $d_n(q,j)<2j+1$, the space $\mathcal{H}_{q,j}$ is a truncated subspace of this irrep. This truncation occurs when $q-n/2<j$: in this case, the constraint on the total excitation number prevents $m$ from taking all values in the spin-$j$ irrep, up to $m=j$.
We refer to such sectors as \emph{unfilled} sectors.

Furthermore, note that $\H\qj$ is non-zero if, and only if
\begin{align}
    j \geq \frac{n}{2}-q\,,
\end{align}
which explains why, in \cref{eq:full_block_structure2}, the summation over $q$ starts at $q = n/2-j$. 
Roughly speaking, this is because achieving lower values of $q < n/2$ requires more qubits to be in the state $\ket{1}$, meaning they tend to be aligned, which corresponds to a higher value of $j$. 
For $q>n/2-j$, the dimension of $\mathcal{H}_{q,j}$ grows linearly with $q$ until it reaches a maximum value of $2j + 1$.

\subsection{PI, U(1)-invariant unitaries}
\label{sec:CC_PI_unitaries}
By Schur's Lemma, any PI, U(1)-invariant unitary $V\in\Vsym$ is block-diagonal with respect to the decomposition in \cref{eq:full_block_structure2}.
That is, it is in the form
\begin{align}
    V \eq \bigoplus_{q=0}^{\infty}\,\bigoplus_{j=\max(\jmin,\,\frac{n}{2}-q)}^{n/2}\Big(\1_{M(n,j)} \otimes v\qj\Big)\,,
 \label{eq:unitary_symmetric_form}
\end{align}
where $v\qj$ is a unitary operator on $\H\qj$. 
To be more precise, $\Vsym$ is exactly the group of all unitaries of the form in \cref{eq:unitary_symmetric_form}:
\begin{align}\label{eq:g-char0}
    \Vsym \,\,=\,\, \bigoplus_{q=0}^{\infty}\, \bigoplus_{j=\max(\jmin,\,\frac{n}{2}-q)}^{n/2}\Big(\1_{M(n,j)} \otimes \UU(\H\qj)\Big)\,,
\end{align}
and the corresponding Lie algebra is
\begin{align}\label{eq:g-char}
    \g \,\,=\,\,\bigoplus_{q=0}^{\infty}\, \bigoplus_{j=\max(\jmin,\,\frac{n}{2}-q)}^{n/2}\Big(\1_{M(n,j)} \otimes \uu(\H\qj)\Big)\,.
\end{align}
Here, we point out that, since the only skew-Hermitian operators that commute with \textit{all} other skew-Hermitian operators are imaginary multiples of the identity, the center of $\g$ is
\begin{align}\label{eq:center}
    \mathfrak{c}\,\,&=\Bigg\{\bigoplus_{q=0}^{\infty}\, \bigoplus_{j=\max(\jmin,\,\frac{n}{2}-q)}^{n/2}\hspace{-4mm}ic\qj\big(\1_{M(n,j)} \otimes \1_{\H\qj}\big)\,:\,c\qj\in\R\Bigg\}\,.
\end{align}
In other words, the center is spanned by the projectors $\Pi\qj=\big(\1_{M(n,j)} \otimes \1_{\H\qj}\big)$, as claimed in \cref{sec:central_constraints}.

\vspace{0.5em}
As we saw in \cref{lem:sym_semi_universality}, when projected to the symmetric subspace ($j=n/2$), the only restrictions on realizable unitaries are constraints on the relative phases between different sectors.
To set aside these relative-phase constraints and focus on other possible restrictions, we use the notion of semi-universality defined in \cref{sec:semi-univ} \cite{marvian_abelian_2024}.
Using the fact that the commutator subgroup of the unitary group is the special unitary group,
\begin{align}
    \big[\text{U}(\mathcal{H}),\text{U}(\mathcal{H})\big]=\text{SU}(\mathcal{H}),
\end{align}
together with the decomposition in \cref{eq:unitary_symmetric_form}, we find that the commutator subgroup of the group of all PI, U(1)-invariant unitaries is
\begin{align}\label{eq:sv1}
    \mathcal{S}\Vsym&:=\Big[\Vsym, \Vsym\Big]\\[4pt]\nonumber
    &\hspace{-10pt}=\bigoplus_{q=0}^{\infty}\, \bigoplus_{j=\max(\jmin,\,\frac{n}{2}-q)}^{n/2}\Big(\1_{M(n,j)} \otimes \SU(\H\qj)\Big)\,,
\end{align}
with the corresponding Lie algebra
\begin{align}
    \big[\g,\g\big]\,\,=\,\, \bigoplus_{q=0}^{\infty}\,\bigoplus_{j=\max(\jmin,\,\frac{n}{2}-q)}^{n/2} \1_{M(n,j)}\otimes \su(\H\qj)\,.
\end{align}
Therefore, the first statement in \cref{lem:sym_semi_universality} says that $\mathcal{SV}_z\PiSym$ is semi-universal in $\mathcal{S}\Vsym \PiSym$.\\

In the following, for any operator $V\in \Vsym $, we define
\begin{align}
    \pi\qj(V)=v_{q,j}\,:=\,\frac{\Tr_{\mathbb{C}^{M(n,j)}}(\Pi_{q,j} V \Pi_{q,j})}{M(n,j)}\,,
 \label{eq:vqj_component}
\end{align}
obtained from the component of $V$ in the above equation, where the partial trace is over the multiplicity factor $\mathbb{C}^{M(n,j)}$ in \cref{eq:full_block_structure2}.  
We can think of $\pi_{q,j}(V)$ as the explicit matrix obtained from the matrix representation of $V$ relative to any basis that is consistent with the irrep decomposition in \cref{eq:full_block_structure2}.
For concreteness, we pick the basis obtained from 
\begin{align}\label{eq:basis}
    \ket{j,m,\alpha} = c_{j,m} J^{j-m}_- \ket{j,j,\alpha}\,,
\end{align}
where
$c_{j,m}$ is a normalization factor, 
and $|j,j,\alpha\rangle$ is a highest-weight vector, e.g.
\begin{align}
    \ket{j,j,\alpha_0}= \ket{\Psi^-}^{\otimes (\frac{n}{2}-j)}\otimes \ket{0}^{\otimes 2j}\,,
\label{eq:jjalpha0}
\end{align}
where $\ket{\Psi^-}$ denotes a 2-qubit singlet state.
Then, $\{\ket{j,m,\alpha}: m=-j,\dotsc, j\}$ defines a $(2j+1)$ dimensional irrep of SU(2), corresponding to angular momentum $j$, which satisfies
\begin{align}
 \begin{split}
    J^2 \ket{j,m,\alpha}&=j(j+1) \ket{j,m,\alpha}\\[4pt]
    J_z \ket{j,m,\alpha} &= m \ket{j,m,\alpha}\,.
 \end{split}
\end{align}
We define matrix elements of $\pi_{q,j}(V)$ with respect to this basis, indexed by $r:=j+m$, i.e.
\begin{align}
 \begin{split}
    &\ket{j,r-j,\alpha}\otimes\left|\left(q+j-\frac{n}{2}\right)-r\right\rangle \\[6pt]
    &\hspace{30pt}:\quad r=0,\dotsc,\dim(\H\qj)-1\,.
 \end{split}
 \label{eq:matrix_basis}
\end{align}
via
\begin{align}
 \label{eq:proj}
    [&\pi_{q,j}(V)]_{r_1,r_2} = \\[4pt]\nonumber
    &=\big(\langle j,r_1-j,\alpha|\otimes \langle k_1|_{\text{osc}}\big)\,V\,\big(|j,r_2-j,\al\rangle\otimes |k_2\rangle_{\text{osc}}\big)\,,
\end{align}
where $k_{1,2}=\big(q+j-\frac{n}{2}\big)-r_{1,2}$.
Note that for any PI operator $V$, these matrix elements are independent of the index $\alpha$ that labels the multiplicity of SU(2).

\section{An accidental Symmetry of the TC interaction }\label{sec:accidental_sym}
\begin{figure}[htp]
    \centering
    \includegraphics[width=0.96\linewidth]{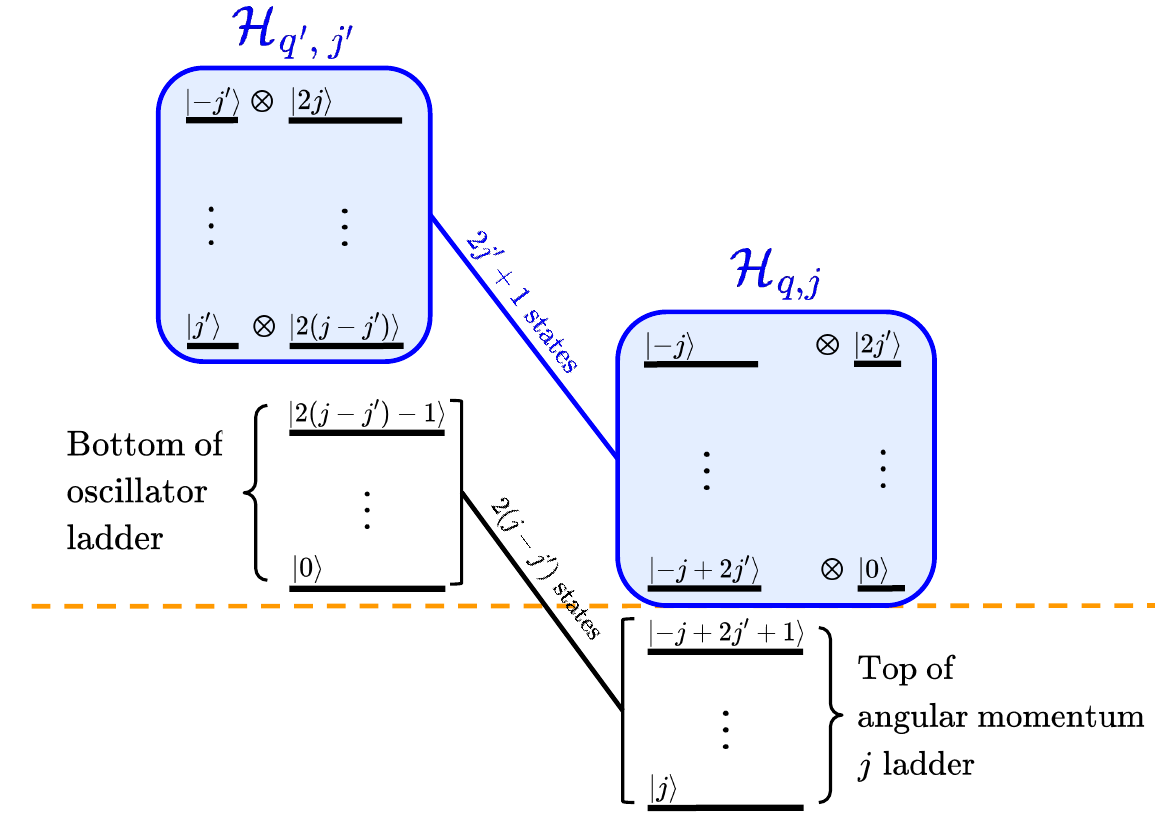}
    \caption{\textbf{General accidental symmetry pairs:}
    The matrix elements of $\HTC$ are identical in the two illustrated sectors, provided that $(q,q')$ and $(j,j')$ satisfy \cref{eq:q_q'}.
    Note that in general, these sectors appear with nontrivial multiplicities within $(\C^2)^{\otimes n}\otimes\mathcal{L}^2(\R)$. 
    Figure from \cite{symmetry_paper}.}
    \label{fig:accidental_symmetry}
\end{figure}
In \cref{sec:beyond_sym}, we saw that, after projecting to $\Hsym\otimes\mathcal{H}_\text{osc}$, the commutator subgroup of the unitaries generated by $\HTC$ and $J_z$ coincides with the commutator subgroup of all PI, U(1)-invariant unitaries. 
In particular, the special-unitary transformations acting on different charge sectors can be chosen independently.
One may expect that a similar independence property should hold between multiple angular momentum $j$ sectors.
However, perhaps surprisingly, we find that because of a peculiar ``accidental'' symmetry of the TC Hamiltonian, which is independent of the permutational and U(1) symmetries discussed in \cref{sec:beyond_sym}, this is not the case.

In the companion paper \cite{symmetry_paper}, we provide a detailed description of this accidental symmetry; in particular, we explain how the symmetry can be understood in terms of Schwinger's oscillator model of angular momentum, and discuss its simplest manifestation in the case of $n=3$ qubits.
Here, we present a formal statement of this accidental symmetry, which we then use to fully characterize $\V_z$ in \cref{thm:full_characterization}.

In \cref{app:AccSym_Prop} we prove that
\begin{proposition}
\label{prop:accidental_symmetry_matrix}
\textbf{(Accidental Symmetry of $\HTC$)}
For each pair of angular momenta $(j,j')$ satisfying
\begin{align}
    0\,<\,j'\,<\,j\,\leq\,\frac{n}{2}\,,
\end{align}
there is a unique pair of distinct sectors, $\H\qj$ and $\H_{q',j'}$, namely sectors with charges 
\begin{align}\label{eq:q_q'}
 \begin{split}
       q  \eq \frac{n}{2}+2j'-j \quad\text{and}\quad
       q' \eq \frac{n}{2}-j'+2j\,,
 \end{split}
\end{align}
for which both of the following equalities are satisfied:
\begin{align}
 \label{eq:dimension_equality}
 \begin{split}
    \text{Equality of dimension:}&\\
    \dim(\mathcal{H}_{q,j})&=\dim(\mathcal{H}_{q',j'})
 \end{split}\\[6pt]
 \begin{split}
    \text{Equality of second moment:}\hspace{-20pt}&\\
    \Tr(\pi\qj(H_{\text{TC}})^2 )&=\Tr(\pi_{q',j'}(H_{\text{TC}})^2)\,.
 \end{split}
 \label{eq:variance_equality}
\end{align}
In particular, both of these sectors have dimension $\dim(\mathcal{H}_{q,j})=\dim(\mathcal{H}_{q',j'})=2j'+1$.
Furthermore, in these two sectors, the matrices of $H_{\text{TC}}$ are identical, with respect to the basis $\ket{j,r-j,\al}\otimes\ket{k}_{\text{osc}}$ defined in \cref{eq:matrix_basis}, where $r:=j+m$.
In particular, the only nonzero matrix elements are
\begin{align}
 \begin{split}
    &\big[\pi\qj(\HTC)\big]_{r,r-1} \,\,=\,\,\big[\pi_{q',j'}(\HTC)\big]_{r,r-1} \,\,=\,\,\\[4pt]
    &\hspace{30pt}=\sqrt{k+1}\times \sqrt{(j+m)(j-m+1)}\,,
 \label{eq:HTC_elements}
 \end{split}
\end{align}
for $r=1,\dotsc,2j'$.
\end{proposition}
Equivalently, these constraints can be stated in terms of the partial isometry
\begin{align} \label{eq:Sjj}
 \begin{split}
    S_{j,j';\,\al,\al'}:=\sum_{m'=-j'}^{j'} &\ket{j',m',\alpha'}\bra{j,m'-j+j', \alpha}\\
    &\otimes\ket{2j-j'-m'}\bra{j'-m'}_{\text{osc}} \,,
 \end{split}
\end{align}
which unitarily maps one copy of $\mathcal{H}_{q,j}$ to a copy of $\mathcal{H}_{q',j'}$, where $j>j'$.
\footnote{The corresponding partial isometry from $\mathcal{H}_{q',j'}$ to $\mathcal{H}_{q,j}$ is $S_{j',j;\al',\al} := S^\dag _{j,j';\al,\al'}$.}
In the following, we often suppress the multiplicity labels $\alpha$ and $\alpha'$, because they are fixed.
In particular, $\ket{j,m,\alpha}$ and $\ket{j',m',\alpha'}$, are constructed from fixed but otherwise arbitrary highest-weight vectors $\ket{j,j,\alpha}$ and $\ket{j',j',\alpha'}$ via \cref{eq:basis}.
For example, one may choose $\ket{j,j,\alpha_0}$ as in \cref{eq:jjalpha0}.
With this convention, we omit the multiplicity indices and simply write the operator as $S_{j,j'}$. 
See \cite{symmetry_paper} for further discussion on the properties of this operator. 

Under these conventions, the TC Hamiltonian $\HTC$ and its phase-shifted variants $\HTC(\phi)$ -- including $\overline{\HTC}=\HTC(\pi/2)$ -- satisfy
\begin{align}\label{eq:sds}
    \big[\HTC , S_{j,j'}\big]= \big[\HTC(\phi), S_{j,j'}\big]=0\,,
\end{align}
for all $\phi\in[0,2\pi)$.
Furthermore, while $J_z$ and $a^{\dag}a$ break this symmetry, they do so in a very specific form:
\begin{align}
 \begin{split}
    \big[J_z,S_{j,j'}\big] &\eq (j-j')S_{j,j'}\\[4pt]
    \big[a^\dag a,S_{j,j'}\big] &\eq 2(j-j') S_{j,j'}\,.
 \end{split}
 \label{eq:S_comms_Jz_adaga}
\end{align}
Equivalently, in terms of the components defined in \cref{eq:vqj_component},
\begin{align}
 \begin{split}
    \pi\qj(J_z) &\eq \pi_{q',j'}(J_z) +(j'-j)\1 \\[4pt]
    \pi\qj(a^{\dag}a) &\eq \pi_{q',j'}(a^{\dag}a) + 2(j'-j)\1\,.
 \end{split}
 \label{eq:acc_sym_Jz_matrix}
\end{align}

\subsection*{Remarks on the accidental symmetry}
Three important remarks on \cref{prop:accidental_symmetry_matrix} must be emphasized:

\vspace{0.5em}
\noindent \textbf{(1)} Note that in \cref{prop:accidental_symmetry_matrix}, we excluded sectors with $j=0$, each of which is one-dimensional.
Indeed, as we further discuss in \cite{symmetry_paper}, in this case the accidental symmetry associates pairs of one-dimensional sectors and implies that the eigenvalues of $\HTC$ in these two sectors are identical.
Of course, these one-dimensional sectors do not contribute to the commutator subalgebra $\sv_z=[\vz,\vz]$ (see \cref{thm:Thm0}).
For $n=2$ qubits, total angular momentum takes only the values $j=0$ and $j=1$; therefore, the accidental symmetry does not impose any non-trivial constraints on $\sv_z$.
(See \cite{symmetry_paper} for a further explanation of the 2-qubit case.)
On the other hand, for $n\geq3$ qubits, the accidental symmetry pairs sectors with dimensions two and larger, and hence imposes restrictions on $\sv_z$.

\vspace{0.5em}
\noindent \textbf{(2)} All pairs of sectors $\H\qj$ related by the accidental symmetry are supported on the subspace of all sectors with charge bounded by
\begin{align} \label{eq:acc_sym_boundaries}
     q \,\,\leq\,\, n + \left\lceil\frac{n}{2}-1\right\rceil\,.
\end{align}
This can be seen using \cref{eq:q_q'}, which implies that the maximal charge for which an accidental symmetry sector appears is
\begin{align*}
    \max\{q,\,q'\} &\eq \frac{n}{2} + \,\max_{0<j'<j\leq n/2}\big\{2j'-j,\,2j-j'\big\} \\[4pt]
    &\eq \frac{n}{2} + \max_{\jmin\leq j\leq n/2}(2j) - \min_{0<j'<j}(j')\\[4pt]
    &\eq n+\left\lceil\frac{n}{2}-1\right\rceil\,.
\end{align*}

\vspace{0.5em}
\noindent \textbf{(3)} For $\qmax<\lfloor\frac{n}{2}\rfloor+3$, the subspace with bounded maximum charge $\qmax$ contains \textit{no} complete pair of accidental symmetry sectors.
On the other hand, for 
\begin{align}
    \qmax \,\,\geq\,\,\left\lfloor\frac{n}{2}\right\rfloor+3\,,
 \label{eq:acc_sym_minimal}
\end{align}
this subspace always contains at least one complete pair of accidental symmetry sectors.
To see this, first note that $q'>q$ for each pair of sectors in \cref{eq:q_q'}.
Then, the minimal primed charge $q'$ is
\begin{align*}
    \min_{0<j'<j\leq n/2}(q') &\eq \frac{n}{2} + \min_{0<j'<j\leq n/2}(2j-j')\\[4pt]
    &\eq \frac{n}{2} + \min_{1<j\leq n/2}(2j-(j-1))\\[4pt]
    &\eq \left\lfloor\frac{n}{2}\right\rfloor+3\,.
\end{align*}

\section{Characterizing unitaries realizable using the TC interaction and global $z$ field}
\label{sec:full_char}
\subsection{Failure of semi-universality}\label{Sec:failure}
When there are $n\geq3$ qubits, \cref{prop:accidental_symmetry_matrix} imposes non-trivial restrictions on the commutator subgroup $\SV_z$ of unitaries generated by Hamiltonians $\HTC$ and $\overline{\HTC}$.
Since both Hamiltonians respect the accidental symmetry, namely
\begin{align}
    \big[\HTC, S_{j,j'}\big]=\big[\overline{\HTC}, S_{j,j'}\big]=0 \,,
\end{align}
for each pair of angular momenta $0<j'<j\leq n/2$, it follows that every
unitary $V\in\SV_z$ also commutes with $S_{j,j'}$,
\begin{align}
    [V,S_{j,j'}]=0 \,.
\end{align}
Thus, for each pair $j>j'$, $V$ acts identically on the two sectors paired by this accidental symmetry. 
In particular,
\begin{align}
    \pi_{q,j}(V)=\pi_{q',j'}(V)\,,
\end{align}
where (\cref{eq:q_q'})
\begin{align*}
    q = \frac{n}{2}+2j'-j\quad\text{and}\quad q' = \frac{n}{2}-j'+2j\,.
\end{align*}
This, in particular, means the Lie algebra associated to $\mathcal{SV}_z$ is strictly contained in the Lie algebra associated to $\mathcal{S}\Vsym$, i.e.,
\begin{align}
    \sv_z \,\neq\, \big[\g,\g\big] \quad:\quad(n\geq3)\,,
\end{align}
which completes the proof of \cref{part:3} of \cref{thm:Thm0}.

\subsection{Characterizing the commutator subgroup}
As stated in the following proposition, it turns out that, aside from constraints on the relative phases, such as those described in \cref{eq:sym_phase_constraint} for the symmetric subspace, the constraints caused by the accidental symmetry are the only ones.
In particular, in \cref{sec:accidental_sym_restricts}, we fully characterize $\SV_z$, the group of unitaries realized by Hamiltonians $\HTC$ and $\overline{\HTC}$, in the truncated subspace of the Hilbert space with charge less than or equal to an arbitrary $\qmax$, and show that
\begin{proposition}
\label{prop:quasi_semi_universality}
For $n\geq3$, each unitary $V\in\SV_z$   
satisfies
\begin{align}
    \pi\qj(V) = \pi_{q',j'}(V) \quad\begin{cases}
        q &= \frac{n}{2}+2j'-j \\[6pt]
        q' &= \frac{n}{2} - j' + 2j\,,
    \end{cases}
 \label{eq:vqj_pairs_prop}
\end{align}
for $0<j'<j\leq n/2$.
\footnote{Note that for $n=2$, there is no such constraint.}

Conversely, for any unitary $V\in\SVsym$ satisfying the above constraints for $0<j'<j\leq n/2$, there exists 
\begin{align*}
    \widetilde{V}\in  \SV_z=\Big\langle\exp({-it\HTC}),\, \exp({-it\overline{\HTC}}) \,\,:\,\, t\in\mathbb{R}\Big\rangle\,,
\end{align*}
such that $\widetilde{V}\Pi^{q_\text{max}}=V\Pi^{q_\text{max}}$.
\end{proposition}
Equivalently, this means that
\begin{align}
    \mathcal{SV}^{\text{acc}}\,\Pi^{\qmax} \eq \mathcal{SV}_z\,\Pi^{\qmax}\,,
 \label{eq:quasi_semi_universality}
\end{align}
where $\mathcal{SV}^{\text{acc}}$ is the subgroup of
\begin{align*}
    \mathcal{S}\Vsym=\left[\Vsym,\Vsym\right]
\end{align*}
that respects the accidental symmetry, i.e.,
\begin{align}
 \begin{split}
     \mathcal{SV}^{\text{acc}}:=\Big\{V\in \mathcal{S}&\Vsym \,\,:\,\, [V,S_{j,j'}]=0\,;\,\\[4pt]
    &\hspace{16pt}\forall (j,j')\,:\,0< j'<j\leq\frac{n}{2}\Big\} \,.
 \end{split}
\end{align}
In summary, for any finite charge cut-off $\qmax$, $\mathcal{SV}_z$ and $\mathcal{SV}^{\text{acc}}$ coincide on the truncated charge subspace.
Equivalently, in terms of Lie algebras, this means that
\begin{proposition}\label{prop:quasi_semi_universality:algebra}
For any  cutoff $q_\text{max}\ge 0$, the truncation of $\sv_z:=\mathfrak{alg}\{i\HTC,i\overline{\HTC}\}$ to the subspace with charge $q\le q_\text{max}$ is equal to
\begin{align}
 \begin{split}
    \sv_z\Pi^{q_\text{max}} = \Big\{A\Pi^{q_\text{max}}\,\,:\,\,&A\in [\mathfrak{g},\mathfrak{g}]\,\,\text{and}\\
    &[A, S_{j,j'}]=0\,\,\forall(j>j')\Big\}\,.
 \end{split}
\end{align}
\end{proposition}
See \cref{sec:accidental_sym_restricts} for the proof of these equivalent propositions.
An interesting corollary of the above propositions is that, as stated in \cref{part:4} of \cref{thm:Thm0}, when restricted to any eigen-subspace of $Q$, semi-universality holds.
That is, for any integer $q\ge 0$, the commutator subgroup of $\V_z$ is equal to the commutator subgroup of all unitaries that respect $\text{U}(1)\times \mathbb{S}_n$ symmetry.
This follows from the fact that the accidental symmetry only imposes restrictions between pairs of sectors labeled $(q,j)$ and $(q',j')$ for which $q\neq q'$.

\subsection{Characterizing the full group $\mathcal{V}_z$}
Having characterized $\mathcal{SV}_z$, and using the fact that $\mathcal{SV}_z$ is the commutator subgroup of $\mathcal{V}_z$, we are now ready to characterize $\mathcal{V}_z$, the group of unitaries that can be realized by Hamiltonians $\HTC$ and $J_z$.
Recall that, as we have seen in \cref{eq:Dec}, a unitary $V$ is in $\mathcal{V}_z$ if and only if there exist $\widetilde{V}\in \mathcal{SV}_z$ and a phase $\theta_z\in[0,4\pi)$ such that
\begin{align}
    V= \exp(-i\theta_z J_z) \widetilde{V} \,,
\end{align}
which in turn implies that the action of $V$ on $\mathcal{H}_{q,j}$ decomposes as
\begin{align}
    \pi_{q,j}(V)= \exp\!\big(-i\theta_z \pi_{q,j}(J_z)\big) \pi_{q,j}(\widetilde{V})\,.
\end{align}
Based on this observation, and using the characterization of $\mathcal{SV}_z$ in \cref{prop:quasi_semi_universality}, we establish the following result.
\begin{theorem}
\label{thm:full_characterization}
For any cutoff $q_{\text{max}}\ge 0$, consider an arbitrary set of unitary matrices 
\begin{align}
    v_{q,j}\in \UU(\H\qj)\,\,:\,\, q\le q_{\text{max}}\,\,,\,\text{and}\,\,\,    j_{\text{min}}\le j\le \frac{n}{2}\,.
\end{align}
Then, there exists a unitary $V \in \V_z$,   realizable using Hamiltonians $\HTC$ and $J_z$,  that implements all these unitaries within the corresponding sectors such that
\begin{align}
     \pi_{q,j}(V)=v_{q,j}\,\,:\,\,q\le q_{\text{max}}\,\,,\,\text{and}\,\,\,j_{\text{min}}\le j\le \frac{n}{2}\,,
\end{align}
if, and only if, there exists $\theta_z\in [-2\pi,2\pi)$ such that  the following conditions hold:
\begin{enumerate}
\item (Accidental Symmetry) For $n\geq3$ and for each pair $(j,j')$ with $0< j'<j \leq n/2$, and $(q,q')$ given by \cref{eq:q_q'},
\begin{align}
    v_{q,j}= e^{i(j'-j)\theta_z} v_{q',j'}\,.
 \label{eq:thm_AccSym_constraint}
\end{align}

\item (Phase Constraint) Phases $\theta\qj := \arg\det(v\qj)$ satisfy 
\begin{align}
 \begin{split}
    \theta_{q,j} &\eq  \theta_z\Tr\big(\pi\qj(J_z)\big)\\[8pt]
    &\eq\begin{cases}
    \left(q-\frac{n}{2}+j+1\right) & :\,q<\frac{n}{2}+j\\\hspace{8pt}\times
    \left(q-\frac{n}{2}-j\right)
    \frac{\theta_z}{2} \\[6pt]
    0 & :\,q\geq \frac{n}{2}+j\,,
    \end{cases}
 \end{split}
 \label{eq:thm_phase_constraint}
\end{align}
where the equations hold modulo $2\pi$.
\end{enumerate}
\end{theorem}
\noindent On the other hand, $J_z^2$ breaks the accidental symmetry in manner such that the constraint in \cref{eq:thm_AccSym_constraint} is broken.
As we show in \cref{sec:Jz2_proof}, indeed using this Hamiltonian the set of realizable unitaries can be extended to achieve semi-universality.\\

\section{Application: PI Qubit Unitaries with the Oscillator as an Ancillary System\\ (Main Result 2)}
\label{sec:qubit_unis}
Finally, we discuss implications of \cref{thm:full_characterization} for specific classes of \textit{qubit} unitaries realized using the TC Hamiltonian and global $z$ field, which are relevant for quantum computing.
To implement a unitary $U$ acting on $n$ qubits, the oscillator is treated as an ancillary system and returned to its initial state.
If the oscillator is initialized in an eigenstate $\ket{k}_{\text{osc.}}$ of its intrinsic Hamiltonian $a^{\dag}a$, such that $a^{\dag}a\ket{k}_{\text{osc.}}=k\ket{k}_{\text{osc.}}$, then we hope to find a unitary $V\in\V_z$ such that for any $n$-qubit state $\ket{\psi}\in(\C^2)^{\otimes n}$,
\begin{align}\label{eq:VU_ancilla_k}
    V\big(\ket{\psi}\otimes \ket{k}_{\text{osc}}\big)\,=\,{e^{-i\alpha}}\big(U\ket{\psi}\big)\otimes \ket{k}_{\text{osc}}\,,
\end{align}
or equivalently,
\begin{align}\label{eq:VU_ancilla_k2}
    U\,=\,e^{i\al}\ipo{k}{V}{k}_{\text{osc}}\,.
\end{align}
In other words, $V$ implements $U$ up to phase $e^{-i\alpha}$, as in \cref{fig:VU_gate}.

In this case, $U$ conserves not only the charge $Q:=a^{\dag}a+J_z+n/2$, but also $a^{\dag}a$ and $J_z$ individually. 
\footnote{Note that $R_z(\phi)$ and $V_{\text{TC}}(r)$ describe the time evolution of the system in the rotating frame. 
Then, \cref{eq:VU_ancilla_k} means that in the lab frame, we realize the unitary $R^\dag_z(T \nu) U R_z(T \nu)$, where $T$ is the total duration of the circuit. 
Of course, if needed, one can cancel these additional rotations by applying $R^\dag_z(T \nu)$ and $ R_z(T \nu)$ before and after the circuit, respectively.}
Therefore, in this way we can only realize $n$-qubit unitaries $U$ that respect the U(1) symmetry $R_z(\phi)=\exp(-{i\theta J_z}): \theta \in [-2\pi, 2\pi)$, or, equivalently, commute with $J_z$, i.e.
\begin{align}
    [U, J_z]=0 \,.
\end{align}
In addition to this symmetry, $U$ also inherits the permutational symmetry of $V$.

It can be easily shown that the group of all PI, $J_z$-conserving  unitaries on $(\C^2)^{\otimes n}$ consists of all unitaries of the form
\begin{align}\label{eq:charU}
 \begin{split}
    U &\,=\,\sum_{j=\jmin}^{n/2} \sum_{m=-j}^j  e^{i\phi_{j,m}}  P_{j,m} \quad:\quad \phi_{j,m}\in[0,2\pi)\,,
 \end{split}
\end{align}
where $e^{i\phi_{j,m}} $ is an arbitrary phase, and $P_{j,m}$ is the projector to the common eigen-subspace of $J^2$ and $J_z$ within $(\C^2)^{\otimes n}$, corresponding to eigenvalues $j(j+1)$ and $m$, respectively, such that
\begin{align}
 \begin{split}
    J^2 &=\sum_{j=j_\text{min}}^{n/2} j(j+1) P_j=\sum_{j=\jmin}^{n/2} j(j+1)\sum_{m=-j}^j  P_{j,m} \\
    J_z &=\sum_{j=\jmin}^{n/2} \sum_{m=-j}^j m P_{j,m} \,.
 \end{split}
\end{align}
To see this, note that any PI operator $A$ commutes with $J^2$ and is therefore block diagonal with respect to sectors with different total angular momentum $j$, i.e.,
$A=\sum_j P_j A P_j$.
Furthermore, by Schur-Weyl duality, $P_j A P_j$ acts nontrivially only on the $\SU(2)$ irrep factor.
Within each $\SU(2)$ irrep, $J_z$ is nondegenerate. 
Therefore, if $A$ also commutes with $J_z$, then $P_j A P_j=\sum_m a_{j,m}P_{j,m}$, for some $a_{j,m}\in\mathbb{C}$.
If $A$ is unitary, then each $a_{j,m}$ is a phase, $a_{j,m}=e^{i\phi_{j,m}}$, which proves the form in \cref{eq:charU}.

Then, we are interested to determine which unitaries of the form in \cref{eq:charU} are realizable using only $\HTC$ and $J_z$ via \cref{eq:VU_ancilla_k2}?

The following theorem gives a complete characterization of the $n$-qubit unitaries realizable using $\HTC$ and $\overline{\HTC}$, without allowing an arbitrary global phase $e^{i\alpha}$.
\begin{theorem} \label{thm:general_k_ancilla} Let $U$ be an $n$-qubit unitary, and let $\ket{k}_{\text{osc}}$ be an eigenstate of the harmonic oscillator Hamiltonian $a^\dag a$ with eigenvalue $k$. 
Then there exists a unitary $V \in \SV_z$, realizable using Hamiltonians $\HTC$ and $\overline{\HTC}$, such that $U = \ipo{k}{V}{k}_{\text{osc}}$, if and only if the following conditions hold:
\begin{enumerate}
\item $U$ is PI and preserves $J_z$, i.e., it is of the form given in \cref{eq:charU}.
\item If $n$ is even, then $\phi_{j=0,m=0} = 0 \, (\mathrm{mod}\, 2\pi)$, or equivalently $U P_{0,0} = P_{0,0}$.
\item If $k = 0$, i.e., the oscillator is in vaccum, then $\phi_{j,-j} = 0 \,(\mathrm{mod}\,2\pi)$ for all $j$, or equivalently,
\begin{align}
    U P_{j,-j} = P_{j,-j}.
\end{align}
\end{enumerate}
\end{theorem}
In words, if the oscillator is initialized in an excited state, i.e., in any eigenstate $\ket{k}_{\text{osc}}$ of $a^{\dag}a$ with eigenvalue $k\neq 0$, then the only restriction on the realizable unitaries $U$ occurs in the sector with $j=0$, which exists only for even $n$. 
Namely, one can realize only those unitaries $U$ satisfying $\phi_{j=0,m=0}=0$.
This restriction arises because $\HTC$, $\overline{\HTC}$, and $J_z$ all vanish in this sector: states in this sector are eigenstates with eigenvalue zero of all three Hamiltonians. 
In practice, however, this restriction is irrelevant, since any PI unitary $U$ that conserves $J_z$ is still realizable using $\HTC$ and $\overline{\HTC}$, up to a global phase $e^{i\alpha}$.

On the other hand, when the oscillator is initialized in the vacuum, i.e., when $k=0$, additional constraints arise in the sectors with $m=-j$. 
This is because, for $k=0$, all such states are eigenstates of both $\HTC$ and $\overline{\HTC}$ with eigenvalue zero:
\begin{align}
    \left(P_{j,m=-j}\otimes \ket{k=0}\!\bra{k=0}_{\text{osc}}\right)\HTC =0\,,
\end{align}
with an analogous equation for $\overline{\HTC}$. 
Furthermore, they are eigenvectors of $J_z$ with eigenvalue $-j$. 
Together, these observations establish the necessity of the conditions in this theorem for the unitary $U$ to be realizable using $\HTC$ and $\overline{\HTC}$. 
We prove the sufficiency of these conditions in \cref{sec:proof2}.

In the above theorem, we did not include the Hamiltonian $J_z$.
As shown in \cref{sec:commutator_subgroups} and \cref{thm:full_characterization}, adding $iJ_z$ to the generators $i\HTC$ and $i\overline{\HTC}$ of the Lie algebra $\sv_z$ only modifies its center. 
In particular, this enlarges the group of realizable unitaries characterized by the above theorem to include unitaries of the form
\begin{align}
    \exp(-i\beta J_z)\, U=U\exp(-i\beta J_z)\,\,,
\end{align}
where $U$ satisfies the conditions of the above theorem and $\beta \in [-2\pi, 2\pi)$.  
Finally, incorporating the global phase freedom $e^{i\alpha}$ allowed in \cref{eq:VU_ancilla_k2}, we arrive at the condition:

\begin{corollary}\label{lem:Jz_conserving}
A PI, U(1)-invariant $n$-qubit unitary $U$, i.e., a unitary of the form given in \cref{eq:charU}, can be realized using Hamiltonians $\HTC$ and $J_z$, up to a global phase $e^{i\alpha}$, by coupling to an oscillator initialized in and returned to the vacuum (i.e., eigenvector with eigenvalue $k=0$ of $a^\dag a$), if and only if there exists $\beta\in[-2\pi,2\pi)$ such that
\begin{align}\label{eq:const0}
    \phi_{j,-j} \eq \alpha + j \beta \,\,:\,\,\mathrm{mod}\, 2\pi,
\end{align}
for all $j = j_{\text{min}}, \dots, n/2$.
\end{corollary}

\noindent\textbf{The special case of two and three qubits.}  
It is important to note that for $n=2$ and $n=3$ qubits, \cref{lem:Jz_conserving} does not place any restrictions on the allowed phases $\phi_{j,-j}$ -- since in each case  $j$ only takes two values $j=0, 1$ for $n=2$, and $j = \tfrac{1}{2}, \tfrac{3}{2}$ for $n=3$.
Indeed in these cases, all PI, U(1)-invariant unitaries are realizable using just the TC interaction and a global $z$ field.
In the companion paper \cite{circuit_paper}, we build on this observation to construct explicit quantum circuits that implement \textit{all} two-qubit PI, U(1)-invariant unitaries through \cref{eq:VU_ancilla_k}, using the TC interaction and a global $z$ field.

\begin{figure*}[htp]
    \centering
    \includegraphics[width=0.7\textwidth]{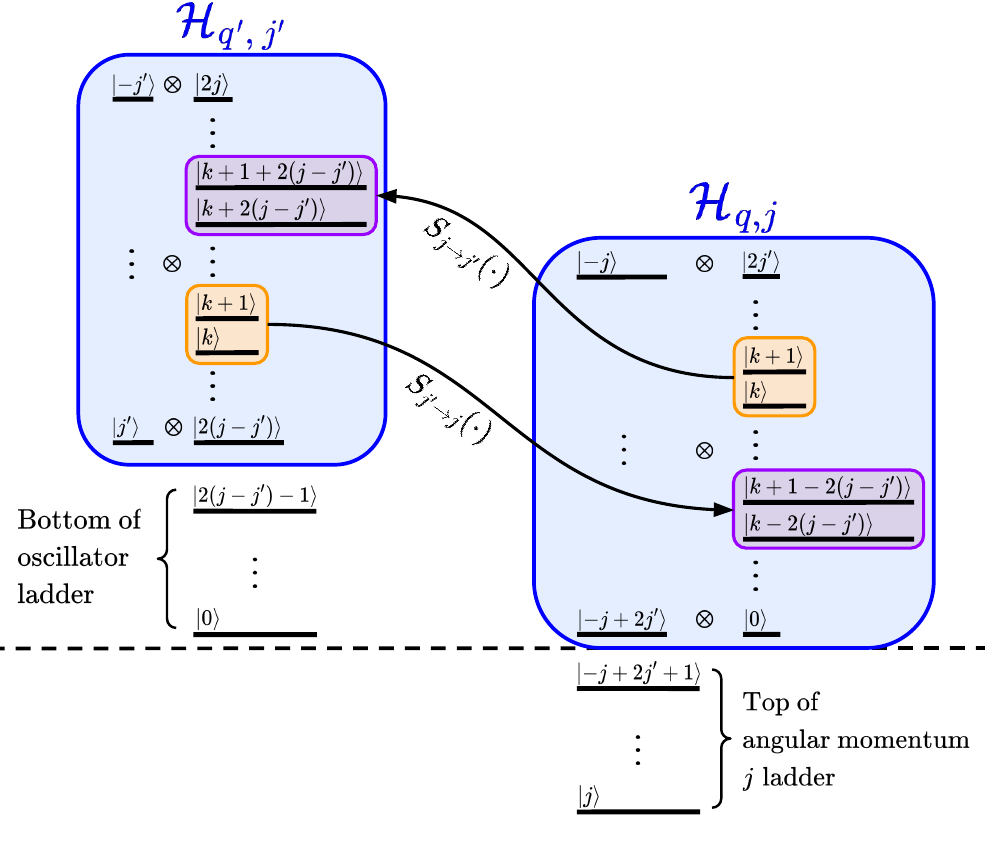}
    \caption{\textbf{Sketch of the proof of \cref{lem:F}}.
     Let $F$ be a PI, $U(1)$-invariant operator supported only on the subspace spanned by $\ket{k}$ and $\ket{k+1}$, shown in the orange box.
     Now consider a pair of sectors $\mathcal{H}_{q,j}$ and $\mathcal{H}_{q',j'}$ that are related by the accidental symmetry.
     This symmetry maps the projection of $F$ onto $\mathcal{H}_{q,j}$ to an operator acting on $\mathcal{H}_{q',j'}$, and vice versa; these symmetry-transformed operators are shown in purple.
     The key point is that these symmetry-transformed operators have no support on the subspace spanned by $\ket{k}$ and $\ket{k+1}$.
     Hence, if we form an augmented operator $\widetilde{F}$ by adding the corresponding symmetry-transformed components to $F$, then $\widetilde{F}$ agrees with $F$ when restricted to the subspace spanned by $\ket{k}$ and $\ket{k+1}$.
     See \cref{app:F_proof} for further details.
     }
    \label{fig:Sjj}
\end{figure*}

\subsection{Proof of Theorem \ref{thm:general_k_ancilla}}\label{sec:proof2}
The necessity of the three conditions in \cref{thm:general_k_ancilla} is already explained above.
We now prove their sufficiency using our characterization of $\mathcal{SV}_z$, the group of unitaries realizable by $\HTC$ and $\overline{\HTC}$, given in \cref{prop:quasi_semi_universality}, or equivalently, using the Lie-algebraic characterization in \cref{prop:quasi_semi_universality:algebra}.

Consider any $n$-qubit Hamiltonian in the form
\begin{align}
    H_1 \eq \sum_{j>0}^{n/2} \sum_{m=-j+1}^j  h_{j,m} P_{j,m}\,,
\end{align}
where $h_{j,m}$ are arbitrary real numbers. 
Note that this Hamiltonian does not have terms of the form $P_{j,-j}$.  
We also define
\begin{align}
    H_2 \eq \sum_{j>0}^{n/2} h_{j,-j} P_{j,-j}\,,
\end{align}
where $h_{j,-j}$ are arbitrary real numbers. 
Note that for $k\ge 1$, any unitary $U$ satisfying the conditions in the theorem can be written as
\begin{align}
    U=\exp(i H_2)\exp(i H_1)=\exp\!\big(i (H_1+H_2)\big)\,,
\end{align}
for properly chosen $H_1$ and $H_2$ in the above form, whereas for $k=0$, any such $U$ can be written as
$U=\exp(i H_1)$. 

\vspace{1mm}
We prove that when $k\ge 1$ there exist Hamiltonians $\widetilde{H}_1$ and $\widetilde{H}_2$ such that (i) 
they are realizable with $\HTC$ and $\overline{\HTC}$, i.e.
\begin{align}
    i\widetilde{H}_{1},\,i\widetilde{H}_{2} \,\in\, \sv_z := \alg\big\{i\HTC, i\overline{\HTC}\big\}\,,
\end{align}
and (ii) they satisfy 
\begin{align}
   r=1,2\,:\quad \widetilde{H}_{r} \big(\mathbb{I}_{\text{qubits}}\otimes |k\rangle_{\text{osc}}\big)=H_{r}\otimes |k\rangle_{\text{osc}}\,,
\end{align}
which means
\begin{align}
    r=1,2\,:\quad\langle k|\exp(i \widetilde{H}_{r})  |k\rangle_{\text{osc}}=\exp(i H_{r}) \,.
\end{align}
For $k=0$, we prove that a Hamiltonian $\widetilde{H}_1$ satisfying the above conditions exists.
This proves \cref{thm:general_k_ancilla} for all $k\ge 0$.

First, for arbitrary $k\ge 0$ define Hamiltonian
\begin{align}\label{eq:Hprime1}
    {H}_1' = \sum_{j>0}^{n/2} \sum_{m=-j+1}^j  h_{j,m}\Big(&P_{j,m} \otimes|k\rangle\langle k|_{\text{osc}} \\\nonumber
    &-P_{j,m-1} \otimes|k+1\rangle\langle k+1|_{\text{osc}}\Big)\,,
\end{align}
and for $k\ge 1$ define Hamiltonian
\begin{align}
 \begin{split}
    {H}_2' = \sum_{j>0}^{n/2}   h_{j,-j}\Big(&P_{j,-j} \otimes|k\rangle\langle k|_{\text{osc}} \\
    &-P_{j,-j+1} \otimes|k-1\rangle\langle k-1|_{\text{osc}}\Big)\,. 
 \end{split}
 \label{eq:Hprime2}
\end{align}
Both Hamiltonians ${H}_1'$ and ${H}_2'$ are PI and respect the  U(1) symmetry, which means $i{H}_{r}'\in \mathfrak{g}$, for $r=1,2$. 
Furthermore,
\begin{align}
    r=1,2\,:\quad\Tr({H}_{r}' \Pi_{q,j})=0\,,
\end{align}
for all $q$ and $j$, which means 
\begin{align}
    r=1,2\,:\quad i{H}_{r}'\in [\mathfrak{g}, \mathfrak{g}]\,.
\end{align}
Finally, we note that while $H_1'$ and $H_2'$ do not necessarily respect the accidental symmetry, one can construct Hamiltonians $\widetilde{H}_1$ and $\widetilde{H}_2$ that do respect this symmetry, and whose restrictions to the supports of $H_1'$ and $H_2'$ coincide with those of $H_1'$ and $H_2'$, respectively.
Roughly speaking, this is possible because the nonzero matrix elements of the operators $S_{j,j'}$,  with which Hamiltonians respecting the accidental symmetry commute,  connect eigenstates of the oscillator that differ by at least two excitations, e.g. $\ket{k}_{\text{osc}}$ and $\ket{k+2}_{\text{osc}}$. 
More precisely,
\begin{lemma} \label{lem:F}
Suppose $iF\in [\mathfrak{g},\mathfrak{g}]$, i.e. $F$ is a PI, U(1)-invariant Hamiltonian on $n$ qubits and the oscillator that satisfies $\Tr(F\Pi_{q,j}) = 0$ for all $q$ and $j$.  
Assume that the support of $F$ is restricted to the bosonic subspaces with $k$ and $k+1$ excitations, such that
\begin{align}
    F = F\Big[\mathbb{I}_{\text{qubits}} \otimes \big(\pure{k}_{\text{osc}} + \pure{k+1}_{\text{osc}}\big)\Big]\,,
\end{align}
for some integer $k \ge 0$.
Then, there exists a Hamiltonian $i\widetilde{F} \in \sv_z := \mathfrak{alg}\{i\HTC, i\overline{\HTC}\}$ such that
\begin{align}
    F=\widetilde{F}\Big[\mathbb{I}_{\text{qubits}} \otimes \big(\pure{k}_{\text{osc}} + \pure{k+1}_{\text{osc}}\big)\Big]\,.
\end{align}
\end{lemma}
The essential ideas behind the proof of this lemma are illustrated in \cref{fig:Sjj}, but we postpone the full proof, which is based on \cref{prop:quasi_semi_universality}, to \cref{app:F_proof}.
Applying this lemma to Hamiltonians $H'_1$ and $H'_2$, we conclude that there exist a Hamiltonian $\widetilde{H}_1$ for $k\ge 0$ and Hamiltonians $\widetilde{H}_1$ and $\widetilde{H}_2$ for $k\ge 1$ with the aforementioned properties, and this completes the proof of \cref{thm:general_k_ancilla}.

\section{Full Characterization of $\V_z$ and $\mathcal{W}_z$\\ (Proofs of the Main Result 1)}
\label{sec:full_char_proof}
In this section, we prove \cref{thm:Thm0} and  \cref{thm:full_characterization}, and provide a complete characterization of $\V_z$ and $\mathcal{W}_z$, 
restricted to sectors with an arbitrarily large, but finite maximum charge $\qmax$.

\subsection{Overview of the proofs}
\begin{step}\label{step:1} In \cref{sec:subsystem_control}, we prove that when restricted to a sector with fixed $q$ and $j$, the two Hamiltonians $\HTC$ and $\overline{\HTC}$ generate the special unitary group on $\mathcal{H}\qj$, i.e.,
\begin{proposition}\label{prop:subspace_control} (Subsystem universality) For any charge $q\geq0$ and angular momentum $j\in \{\jmin,\,\dotsc,\,n/2\}$,
\begin{align}
 \label{eq:subspace_control}
    \mathfrak{alg}_{\mathbb{R}}\Big\{i\pi_{q,j}\big(\HTC\big),\,i\pi_{q,j}\big(\overline{\HTC}\big)\Big\} \eq \su(\mathcal{H}_{q,j})\,.
\end{align}
\end{proposition}
Equivalently, this means
\footnote{Note that for $j<n/2$, $\SV_z\Pi_{q,j}$ contains multiple identical copies of $\SU\big(\mathcal{H}_{q,j}\big)$, since $\H\qj$ appears with multiplicity $M(n,j)$.}
\begin{align}
    \SV_z\Pi\qj \,\simeq\, \1_{M(n,j)}\otimes\SU(\H\qj)\,.
\end{align}
Following the language used in \cite{HLM_2024_SU(d)}, we refer to this property as \emph{subsystem universality}.

As noted before, Keyl et al. \cite{Keyl_2014_control} previously proved this in the special case of $j=n/2$.
\footnote{More precisely, they studied the projection of 
$\wz:=\alg_{\R}\big\{i\HTC,\,iJ_z,\, i \Nhat\big\}
$ to sector $\H_{q,\,j=n/2}$, and showed that for any $q\geq0$, it is the entire unitary algebra $\mathfrak{u}(\H_{q,\,j=n/2})$. 
Since $[\uu(d),\uu(d)]=\su(d)$, this means
\begin{align}\nonumber
    \pi_{q,\,j=n/2}(\mathfrak{sv}_z) \,=\, \pi_{q,\,j=n/2}\big([\wz,\wz]\big)\,=\,\su(d)\,.
\end{align}}
To prove this result, we use a slightly different approach based on   \cref{lem:subspace_control}, which is adapted from Theorem 2 of \cite{Schirmer_etal_2001_controllability}.
(See \cref{app:subsystem_control} for further discussion of this lemma.)
The proof presented here, based on \cref{lem:subspace_control}, is slightly simpler and more general than the one used in \cite{Keyl_2014_control}, in that it applies to arbitrary angular momentum $j$, and can be used to classify other sets of Hamiltonians.
\end{step}
\begin{step}\label{step:2} In \cref{sec:Thm0pt1_proof}, we combine the subsystem universality property together with \cref{lem:comm_subalg} and prove \cref{part:1} of \cref{thm:Thm0}.
That is, we show that $\sv_z$ is centerless and equal to $[\vz,\vz] \eq [\wz,\wz]$.
\end{step}
\begin{step}\label{step:3} Then, using \cref{lem:semi_universality}, in \cref{sec:semi_universality_arbitraryJ} we prove that when restricted to each angular momentum sector $j$, considering any finite cutoff $\qmax\geq0$,
\begin{align}
    (\Pi^{\qmax}P_j)\mathcal{SV}_z =(\Pi^{\qmax}P_j)\mathcal{S}\Vsym\,,
\end{align}
or equivalently, in terms of Lie algebras,

\begin{align}
    (\Pi^{\qmax}P_j)\mathfrak{sv}_z =(\Pi^{\qmax}P_j)[\mathfrak{g},\mathfrak{g}] \,.
\end{align}
This, in particular, proves \cref{part:5} of \cref{thm:Thm0} and the semi-universality property in \cref{lem:sym_semi_universality}, the special case where $j=n/2$.
\end{step}
\begin{step}\label{step:4} Then, in \cref{sec:accidental_sym_restricts} we fully characterize the projection of $\mathcal{SV}_z$ onto the subspace of sectors with charge bounded by a finite but arbitrarily large cutoff $\qmax$.
This establishes \cref{prop:quasi_semi_universality}.
\end{step}
\begin{step}\label{step:5} In \cref{sec:sufficiency_proof_fullchar}, we take into account the center of $\vz$, which corresponds to the relative phases between sectors with different $q$ and $j$, to fully characterize $\mathcal{V}_z$.
Here, we prove the sufficiency of the two conditions in \cref{thm:full_characterization}.
\end{step}
\begin{step}\label{step:6} Finally, in \cref{sec:Jz2_proof}, we show that including Hamiltonian $J_z^2$, in addition to $\HTC$ and $\overline{\HTC}$, suffices to achieve semi-universality, i.e.,
\begin{align*}
    \Pi^{\qmax}\SV_z^+ \eq \Pi^{\qmax}\SVsym\,,
\end{align*}
which proves \Cref{part:3} of \cref{thm:Thm0}.
\end{step}

\subsection{Tools}
\label{sec:tools}
Before proceeding to the various proofs outlined above, we introduce two tools -- \cref{lem:subspace_control} and \cref{lem:semi_universality} -- which are broadly useful for studying controllability of quantum systems in the presence of symmetries.
\subsection*{Tool 1: A simple criterion for universality}
The following lemma is a special case of general control theoretic results from \cite{Schirmer_etal_2001_controllability}:
\begin{lemma}\label{lem:subspace_control}
Consider anti-Hermitian operators $i(A_++A_-)$ and $A_+-A_-$, defined by
\begin{align}
 \begin{split}
    A_+ &= \sum_{y=y_{\text{min}}}^{y_{\text{min}}+d-2} a_y\,\ket{y}\langle y+1| \,\,:\,\, a_y\in\R\,,\\[4pt]
    \,\,\quad A_-&:=(A_+)^{\dag}\,,
 \end{split}
 \label{eq:A+_form}
\end{align}
where $\ket{y}: y=y_{\text{min}},\,\dotsc,\,y_{\text{min}}+d-1$ is an orthonormal basis for $\mathbb{C}^d$. 
Suppose $a_y$ are all non-zero and for all $y>y_{\text{min}}$,
\begin{align}\label{eq:conds}
    \Delta^2 a_y^2 \,\neq\, \Delta^2 a_{\text{y}_{\text{min}}}^2\,,
\end{align}
where
\begin{align}
    \Delta^2 a_y^2 \,:=\, a_{y+1}^2-2 a_y^2+a_{y-1}^2\,,
 \label{eq:finite_diff}
\end{align}
is the second order finite difference of 
\begin{align}
    a_y^2 \,:=\, \langle y|A_+ A_-\ket{y}\,,
\end{align}
and we use the convention that $a_{y_{\text{min}}-1}=a_{y_{\text{min}}+d-1}=0$. 
Then, the Lie algebra generated by operators $i(A_++A_-)$ and $A_+-A_-$ is the set of traceless anti-Hermitian operators,
\begin{align}
    \mathfrak{alg}_{\mathbb{R}}\big\{i(A_++A_-),\,(A_+-A_-)\big\} \,=\, \mathfrak{su}(d)\,.
\end{align}
\end{lemma}
This lemma follows from a more general result (\cref{lem:subspace_control_general}) presented in \cite{Schirmer_etal_2001_controllability}, which we discuss in \cref{app:subsystem_control}.

\subsection*{Tool 2: From subsystem universality to semi-universality}
The following lemma, quoted from \cite{HLM_2024_SU(d)},  can be used to establish whether a subgroup $\U\subseteq\V^G$ of $G$-invariant unitaries (for any finite or compact Lie group $G$) is \textit{semi-universal}, that is whether $\U$ contains the commutator subgroup $[\V^G,\V^G]$.

\begin{lemma}\label{lem:semi_universality}
Let $\V^G$ be the group of all $G$-invariant unitaries; that is, $V\in\V^G$ if and only if $[V,U(g)]=0$ for all $g\in G$, where $U(g)$ is the unitary representation of a finite or compact Lie group $G$.
Suppose this representation has the isotypic decomposition,
\begin{align}
    \H = \bigoplus_{\lambda\in\Lambda} \H_{\lambda} = \bigoplus_{\lambda\in\Lambda} \mathcal{Q}_{\lambda}\otimes\mathcal{M}_{\lambda}\,,
\end{align}
where the sum is over a set of finite $\Lambda$ of inequivalent irreps of $G$, $\mathcal{Q}_{\lambda}$ is a space carrying the irrep $\lambda$, and $\mathcal{M}_{\lambda}$ corresponds to the multiplicity of $\lambda$, i.e., $\dim(\mathcal{M}_{\lambda}) > 0$ is the multiplicity of $\lambda$ in $\H$.
Let $\pi_\lambda (A)=\Tr_{\mathcal{Q}_\lambda}(\Pi_\lambda A\Pi_{\lambda})/\dim(\mathcal{Q}_\lambda)$ be the projection of any operator $A$ to $\mathcal{M}_{\lambda}$, where the partial trace is over $Q_\lambda$.
A subgroup $\U\subseteq\V^G$ of $G$-invariant unitaries contains the commutator subgroup of all $G$-invariant unitaries $\SV^G=[\V^G,\V^G]$, if, and only if, the following two conditions hold: \\

\noindent\textbf{A (Subsystem universality in all sectors):} For any irrep $\lambda\in\Lambda$, the action of $\U$ on the corresponding multiplicty subsystem $\mathcal{M}_{\lambda}$ contains $\SU(\mathcal{M}_{\lambda})$, i.e. $\SU(\mathcal{M}_{\lambda})\subseteq\pi_{\lambda}(\U)=\{\pi_{\lambda}(V):V\in\U\}$. \\

\noindent\textbf{B (Pairwise independence):} For any pair of distinct irreps $\lambda_1,\lambda_2\in\Lambda$, if $\dim(\mathcal{M}_1)=\dim(\mathcal{M}_2)\geq2$, then there exists a unitary $V\in\U$ such that
\begin{align}
    \big|\Tr(\pi_{\lambda_1}(V)\big| \,\neq\, \big|\Tr(\pi_{\lambda_2}(V)\big|\,.
\end{align}
\end{lemma}
This result is proven in \cite{HLM_2024_SU(d)}, using  Goursat's and Serre's lemmas. 
Condition $\textbf{B}$ may be replaced by the following condition, which is more easily applied to the group $\V_z\subseteq\Vsym$:\\

\noindent   \textbf{B' (Pairwise independence, equivalent condition):} 
For any pair of distinct irreps $\lambda_1,\lambda_2\in\Lambda$, if $\dim(\mathcal{M}_1)=\dim(\mathcal{M}_2)\geq2$, then there exists a Hamiltonian $H$ such that $\exp(iHt)\in\U$ for all $t\in\R$, and
\begin{align}
    \Tr(H^2\tau_{\lambda_1}) - \Tr(H\tau_{\lambda_1})^2 \neq \Tr(H^2\tau_{\lambda_2}) - \Tr(H\tau_{\lambda_2})^2\,,
 \label{eq:cond_B'}
\end{align}
where $\tau_{\lambda}=\Pi_{\lambda}/\Tr(\Pi_{\lambda})$ is the maximally-mixed state in $\mathcal{Q}_{\lambda}\otimes\mathcal{M}_{\lambda}$.
Note that the quantity $\Tr(H^2\tau_{\lambda})-\Tr(H\tau_{\lambda})^2$ is indeed the \textit{energy variance} associated to Hamiltonian $H$ in state $\tau_{\lambda}$.

Here, we show that conditions \textbf{A}, \textbf{B} are equivalent to conditions \textbf{A}, \textbf{B'}.
First, to see that \textbf{B'} implies \textbf{B}, note that
\begin{align}
 \nonumber
    \Tr(H^2 &\tau_{\lambda}) - \Tr(H \tau_{\lambda})^2 \\[6pt]
    &\,=\, \frac{\Tr\big(\pi_{\lambda}(H^2)\big)}{\dim(\mathcal{M}_{\lambda})} - \frac{\Tr\big(\pi_{\lambda}(H)\big)^2}{\dim(\mathcal{M}_{\lambda})^2} \label{eq:energy_variance}\\[6pt]
    \nonumber&\,=\, -\frac{1}{2 \dim(\mathcal{M}_{\lambda})^2} \frac{d^2}{dt^2} \left\{\left|\Tr\left[\pi_{\lambda}(e^{-iH t})\right]\right|^2\right\}_{t=0}
    \,.
\end{align}
Hence, the inequality in \cref{eq:cond_B'} implies $\left|\Tr\left[\pi_{\lambda_1}\big(e^{-iHt}\big)\right]\right| \neq \left|\Tr\left[\pi_{\lambda_2}\big(e^{-iHt}\big)\right]\right|$ for some nonzero time $t$ near zero , as these functions -- which are well-behaved, i.e. infinitely differentiable -- having different second derivatives at $t=0$ necessarily means that they also take on different values in a neighborhood of $t=0$. 

Conversely, if \textbf{A} and \textbf{B} hold, then according to \cref{lem:semi_universality}, $\mathcal{T}$ contains $\mathcal{SV}^G$, which in turn implies for any $G$-invariant Hamiltonian $H$, satisfying the conditions $\Tr(\pi_{\lambda}(H))=0$ for all $\lambda\in\Lambda$, we have the unitary $\exp(i H t)\in \mathcal{T} $ for all $t\in\mathbb{R}$.
Now for any distinct pairs $\lambda_1,\lambda_2\in\Lambda$ we can pick $H$ such that $\pi_{\lambda_1}(H)=0$, whereas $\pi_{\lambda_2}(H)$ is an arbitrary traceless Hermitian operator.
If $\dim(\mathcal{M}_1)=\dim(\mathcal{M}_2)\geq2$, then $\pi_{\lambda_2}(H)$ can be non-zero, which in turn implies $\Tr(H^2\tau_{\lambda_2}) - \Tr(H\tau_{\lambda_2})^2=\Tr(H^2\tau_{\lambda_2})>0$.
Hence, condition \textbf{B'} is satisfied.

\subsection{\cref{step:1}: Subsystem Universality (Proof of \cref{prop:subspace_control})}\label{sec:subsystem_control}
Recall the decomposition
\begin{align}
    (\C^{2})^{\otimes n}\otimes\mathcal{L}^2(\R) =
    \bigoplus_{j=\jmin}^{n/2}\bigoplus_{q=n/2-j}^{\infty}\left(\C^{M(n,j)} \otimes\H\qj\right)
\end{align}
of the combined qubit-oscillator Hilbert space into invariant sectors labeled by angular momentum ($j$) and charge ($q$).
Using \cref{lem:subspace_control}, we can easily characterize the projection of the Lie algebra generated by the two Hamiltonians $\HTC$ and $\overline{\HTC}$ into any sector $\mathcal{H}_{q,j}$.
To do this, choose
\begin{align}
    A_+ \,:=\, \pi\qj\big(J_- \otimes a^{\dag}\big)
\end{align}
to be the projection of $J_-\otimes a^{\dag}$ onto $\mathcal{H}_{q,j}$, and $A_-=(A_+)^{\dag}$, so that
\begin{align}
 \begin{split}
    \pi\qj(\HTC) &= \gTC(A_++A_-)\\[4pt]
    \pi\qj(\overline{\HTC}) &= -i\gTC(A_+-A_-)\,.
 \end{split}
\end{align}
Explicitly, this projection can be found in terms of matrix elements of this operator with respect to the basis
\begin{align}
    \ket{y}=\left|j,m=y-\frac{n}{2},\alpha\right\rangle\otimes \ket{k=q-y}_{\text{osc}}\,,
\end{align}
where $\{\ket{j,m,\alpha}\}$ is defined in \cref{eq:basis}.
Here, $y_{\text{min}}=\frac{n}{2}-j$, and $y$ takes the values
\begin{align}
    y=\frac{n}{2}-j,\,\dotsc\,,\frac{n}{2}-j+(d-1)\,,
 \label{eq:y_range}
\end{align}
where $d=d_n(q,j)$ is the dimension of $\H\qj$, given in \cref{eq:Hqjn_dimension}.
Note that with respect to this basis, the operator $J_- \otimes a^{\dag}$ takes the form required in \cref{eq:A+_form}.

Then, $A_+A_-$ is equal to the projection of the operator
\begin{align}
    J_-J_+\otimes a^{\dag}a \eq (J^2-J^2_z-J_z)  \otimes a^{\dag}a 
\end{align}
to the sector $\mathcal{H}_{q,j}$.
This, in particular, implies  
\begin{align}
 \begin{split}
    \langle y|A_+A_-\ket{y}&= \left(j(j+1)-\left(y-\frac{n}{2}\right)\left(y-\frac{n}{2}+1\right)\right)\\[2pt]
    &\hspace{30pt}\times (q-y)\,,
 \end{split}
\end{align}
which is a polynomial of degree 3 in $y$. 
Note that this function vanishes at points $y_{\text{min}}-1=n/2-j-1$ and $y_{\text{min}}+d-1=\min(n/2+j,\,q)$, which is consistent with our convention in \cref{lem:subspace_control}.
One can confirm that \cref{eq:finite_diff} holds for all $y$ in \cref{eq:y_range}. 
In particular, the second-order finite differences are given by
\begin{align}
    \Delta^2 a_y^2 \,=\, 2(n+q-1) - 6y\,,
\end{align}
which is linear in $y$, meaning the condition in \cref{eq:conds} is satisfied for all $y>y_{\text{min}}$.
Therefore, \cref{lem:subspace_control} implies that
\begin{align}
    \alg_{\R}\Big\langle \pi\qj(i\HTC),\,\pi\qj(i\overline{\HTC})\Big\rangle \eq \su(\H\qj)\,.
 \label{eq:suHqj}
\end{align}
Equivalently, when projected to an arbitrary sector $\H\qj$, the family of unitaries realized by $\HTC$ and $\overline{\HTC}$ is the special unitary group in $\mathcal{H}_{q,j}$, which proves \cref{prop:subspace_control}.
\footnote{It is worth noting that if we apply the above lemma to the Hamiltonians $J_x=(J_++J_-)/2$ and $J_y=(J_+-J_-)/(2i)$, i.e., in the absence of interaction with the oscillator, then $a_y^2$ will be a quadratic polynomial in $y$. 
Therefore, its second-order finite difference $\Delta^2 a_y^2$ will be constant, independent of $y$, which means the assumption of the lemma is not satisfied. 
Of course, in this case we know that the group generated by the projections of the unitaries $R_x(\theta) = \exp(-i\theta J_x)$ and $R_y(\theta) = \exp(-i\theta J_y)$ is $\SU(2)$, and not $\SU(\mathcal{H}_{q,j})$, which is consistent with the lemma.}

\subsection{\cref{step:2}: Proof of \cref{part:1} of \cref{thm:Thm0}}
\label{sec:Thm0pt1_proof}
First, $[\vz, \vz]=[\wz, \wz]$ follows immediately from \cref{eq:Lie_vw}.
Furthermore, since $\sv_z$ is contained in $\vz$, its commutator subalgebra is also contained in the commutator subalgebra of $\vz$, i.e., $[\sv_z,\sv_z]\subset [\vz,\vz]$. 
The fact that these two are actually equal follows from the identities
\begin{align}
    \big[J_z,\HTC\big]\,=\,-i\overline{\HTC}\,\,,\quad \left[J_z,\overline{\HTC}\right] \,=\, i\HTC\,,
 \label{eq:Jz_HTC_commutator}
\end{align}
together with the following lemma, which is proven in \cref{app:comm_subalg_proof}.
\begin{lemma}
\label{lem:comm_subalg}
Suppose Lie algebra $\mathfrak{h}=\mathfrak{alg}\{A,B\}$ is generated by elements $A$ and $B$, which satisfy $[A, [A,B]]\propto B$ with a non-zero coefficient. 
Then, the commutator subalgebra of $\mathfrak{h}$ is generated by $B$ and $\overline{B}:=[A,B]$, i.e.,
\begin{align}
    \big[\mathfrak{h},\mathfrak{h}\big] = \mathfrak{alg}\{B,\overline{B} \}\,.
\end{align}
\end{lemma}
In particular, taking $A=\HTC$ and $B=J_z$, we conclude that $[\vz,\vz]=\sv_z$, which completes the proof of \cref{eq:comm_subgroup_relations}.

\vspace{2mm}
To prove that $\sv_z$ is centerless is equivalent to show that for any non-zero $A\in \sv_z$, at least one of $[A,\HTC]$ or $[A,\overline{\HTC}]$ is non-zero. 
Clearly, all elements of $\sv_z$ are PI and commute with $J^2$ and $Q$;
therefore, they are block-diagonal with respect to the projectors $\{\Pi\qj\}$.
Furthermore, it can be easily seen that
\begin{align}
    \Tr\big(\Pi\qj \HTC\big)=\Tr\big(\Pi\qj \overline{\HTC}\big)=0\,,
\end{align}
for all $q$ and $j$.
\footnote{Note that for any eigenstate $\ket{k}_{\text{osc}}$ of $a^\dag a$, $\ipo{k}{\HTC}{k}_{\text{osc}}=0$.}
This, in turn, implies that all elements of Lie algebra $\sv_z$ have a decomposition as $A=\sum\qj \widetilde{A}\qj$, where $\widetilde{A}\qj=\Pi\qj A$ is orthogonal to $\Pi\qj$, that is,
\begin{align}\label{eq:orth}
    \forall A\in \sv_z\,\,:\,\,\Tr\big(\Pi\qj A\big)=\Tr\big(\widetilde{A}\qj\big)=0\,.
\end{align}
Therefore, to prove that  $\sv_z$ is centerless, it suffices to show that any non-zero PI operator $\widetilde{A}\qj=\Pi\qj A$ that commutes with both $\Pi\qj\HTC$ and $\Pi\qj\overline{\HTC}$ should be proportional to $\Pi\qj$.
That is, if $\widetilde{A}\qj$ is PI, then for all $q$ and $j$,
\begin{align} \label{eq:Atilde_comms}
    \big[\widetilde{A}\qj,\HTC\big]=\big[\widetilde{A}\qj,\overline{\HTC}\big]=0 \,\,\, \Longrightarrow \,\,\,  \widetilde{A}\qj\propto \Pi\qj\,.
\end{align}
To establish \cref{eq:Atilde_comms}, first note that any PI operator $\widetilde{A}\qj$ decomposes as 
\begin{align}
    \widetilde{A}\qj = \1_{M(n,j)}\otimes a\qj\quad:\quad a\qj\in \mathcal{L}(\H\qj)\,,
\end{align}
with respect to the decomposition in \cref{eq:full_block_structure2}, which means
\begin{align}
 \begin{split}
    \big[\widetilde{A}\qj,\HTC\big] &=0 \,\, \implies \,\,\big[a\qj,\,\pi\qj(\HTC)\big] =0\\[4pt]
    \big[\widetilde{A}\qj,\overline{\HTC}\big] &= 0\,\, \implies \,\,\big[a\qj,\,\pi\qj(\overline{\HTC})\big] =0\,.
 \end{split}
\end{align}
Since $\pi\qj(\HTC)$ and $\pi\qj(\overline{\HTC})$ generate the Lie algebra $\su(\H\qj)$ (the subsystem universality property in \cref{prop:subspace_control}), $a\qj$ commutes with every element of $\su(\H\qj)$, which mean it can only be a scalar multiple of the identity in $\H\qj$.
Therefore, $\widetilde{A}\qj$ is also proportional to the identity in $\C^{M(n,j)}\otimes\H\qj$, that is $\widetilde{A}\qj\propto\Pi\qj$.
Finally, since by \cref{eq:orth}, $\Tr(\widetilde{A}\qj)=0$, we conclude that $\widetilde{A}\qj=0$, which in turn, implies $\mathfrak{sv}_z$ is centerless.

\subsection{\cref{step:3}: Semi-universality inside each angular momentum sector}
\label{sec:semi_universality_arbitraryJ}
Next, we use \cref{lem:semi_universality} to prove 
\begin{align}
   (\Pi^{\qmax}P_j) \mathcal{SV}_z \eq (\Pi^{\qmax}P_j)\mathcal{S}\Vsym\,,
 \label{eq:semi_universality_j}
\end{align}
or equivalently,
\begin{align}\label{eq:qkr}
 \begin{split}
    (\Pi^{\qmax}P_j) \sv_z &\eq (\Pi^{\qmax}P_j) [\g,\g]\\[4pt]
    &\eq \1_{M(n,j)}\otimes  \bigoplus_{q=0}^{q_{\text{max}}}\mathfrak{su}\big(\dim(\H\qj)\big)\,.
 \end{split}
\end{align}
As a preliminary step toward the general case, treated in \cref{thm:full_characterization}, we first prove this statement for the special case $j=n/2$, corresponding to the symmetric subspace.
In particular, we apply \cref{lem:semi_universality} to 
\begin{align*}
    \SV_z\PiSym\subset\SV^{\Uo\times\Sn}\PiSym\,,
\end{align*}
where
\begin{align}
    \PiSym=\Pi^{\qmax}P_{j=n/2}\,.
\end{align}
We now check that both conditions of \cref{lem:semi_universality} are satisfied.

First, subsystem universality of $\SV_z$ in all sectors, condition A, was already shown in \cref{sec:subsystem_control}.
Second, when restricted to the symmetric subspace, $\HTC$ satisfies condition B' of \cref{lem:semi_universality}. 
Indeed, for $q<n$, the sectors $\H_{q,\,j=n/2}$ have distinct dimensions, namely $q+1$. 
On the other hand, for $q\geq n$, the sectors $\H_{q,\,j=n/2}$ all have the same dimension, namely $n+1$. 
In this regime, we compute in \cref{app:energy_variance} that
\begin{align}
 \begin{split}
    \Tr(\HTC\tau_{q,\,j=n/2}) &= 0\\[2pt]
    \Tr\left(\HTC^2\tau_{q,\,j=n/2}\right) &= \frac{n(n+2)(2q+1-n)}{6}\,.
 \end{split}
 \label{eq:HTC_energy_variances}
\end{align}
Therefore,
\begin{align*}
    \Tr(\HTC^2\tau_{q,\,j=n/2}) - \Tr(\HTC\tau_{q,\,j=n/2})^2
\end{align*}
is strictly increasing with $q$ for $q\geq n$. 
Hence, \cref{eq:cond_B'} is satisfied.
We conclude that both conditions of \cref{lem:semi_universality} hold, which proves \cref{eq:semi_universality_j} in the special case $j=n/2$, i.e.,
\begin{align}\label{eq:CommSub_symmetric}
 \begin{split}
    \sv_z \PiSym &\eq [\g,\g]\PiSym \\[4pt]
    &\eq\bigoplus_{q=0}^{\qmax} \mathfrak{su}\big(\min\{n+1, \,q+1\}\big)\,.
 \end{split}
\end{align}

This argument extends straightforwardly to arbitrary angular momentum $j$ sectors. 
First, note that, upon restriction to a fixed angular momentum sector $j$, the Hilbert space decomposes as
\begin{align}
    (P_j\otimes \mathbb{I}_{\rm osc}) \big((\C^{2})^{\otimes n}\otimes \HHO\big)\cong\mathbb{C}^{M(n,j)} \otimes \mathbb{C}^{2j+1}\otimes \HHO\,,
\end{align}
where $M(n,j)$ is the multiplicity of the spin-$j$ irrep in $(\C^2)^{\otimes n}$.
This subspace is spanned by basis vectors of the form
\begin{align}\label{eq:basis1}
    \ket{j,m,\alpha}\otimes \ket{k}_{\mathrm{osc}}\quad:
    \quad -j\le m\le j\,\,,\quad k\ge 0,
\end{align}
where $\alpha=1,\ldots,M(n,j)$ labels the multiplicity space. 
All operators under consideration are PI, and hence act trivially on this multiplicity space, which carries the corresponding representation of $\mathbb{S}_n$.

Furthermore, using \cref{eq:Hqjn_dimension}, one sees that, within a fixed angular-momentum sector $j$, the dimensions of the charge sectors $\mathcal{H}_{q,j}$ are given by
\begin{align}
    d_n(q,j) \eq \min\left\{
                    2j+1,\,
                    q+1+j-\frac{n}{2}
                \right\}.
\end{align}
Thus, as in the fully symmetric sector $j=n/2$, these dimensions increase monotonically with $q$ until they reach their maximal value $2j+1$.
We refer to the sectors with $d_n(q,j)=2j+1$ as filled sectors.
Equivalently, $\mathcal{H}_{q,j}$ is filled precisely when $q\ge \frac{n}{2}+j$.
Therefore, two distinct charge sectors within the same angular-momentum subspace $j$ can have the same dimension only if both are filled.

In \cref{app:energy_variance} we explicitly calculate the generalization of \cref{eq:HTC_energy_variances} to arbitrary \textit{filled} $j$ sectors, and show that
\begin{align}
 \begin{split}
    \Tr(\HTC\tau_{q,j}) &\eq 0,\\[4pt]
    \Tr\left(\HTC^2\tau_{q,j}\right)
    &\eq \frac{2j(2j+2)(2q+1-n)}{6}\,.
 \end{split}
 \label{eq:HTC_energy_variances_general_j1}
\end{align}
In particular, since $\Tr(\HTC\tau_{q,j})=0$, the variance of $\HTC$ in $\tau_{q,j}$ is equal to the second expression in \cref{eq:HTC_energy_variances_general_j1}.
This variance grows monotonically with $q$.
Hence, among filled sectors, which all have the same dimension, distinct charge sectors have distinct variances.
This verifies the second condition of \cref{lem:semi_universality}.
Together with the dimension condition above, this proves \cref{eq:semi_universality_j} and \cref{eq:qkr} in the subspace with fixed angular momentum $j$.\\

Although the identities in \cref{eq:HTC_energy_variances_general_j1} can be
verified by a direct calculation, there is also a more conceptual way to understand how they can be obtained from \cref{eq:HTC_energy_variances}, and how 
the same semi-universality argument applies for every $j$.
The key observation is that the matrix elements of $\HTC$ and $\overline{\HTC}$ in the basis above are independent of both $n$, the total number of qubits, and the multiplicity label $\alpha$.
Instead, they depend only on $j$, $m$, and $k$; see \cref{eq:HTC_elements}.
This property is inherited from the collective angular momentum operators $J_x,J_y,J_z$.\footnote{More generally, any operator expressible as a function of the collective angular momentum operators, tensored with oscillator operators, has the same property.}
Consequently, the matrix elements of $\HTC$ and $\overline{\HTC}$ in the angular-momentum sector $j$ of an $n$-qubit system coincide with those in the angular-momentum sector $j$ of an auxiliary system of $n'=2j$ qubits.

In the system with $n'=2j$ qubits, the sector with angular momentum $j=n'/2$ is precisely the fully symmetric subspace of $(\mathbb{C}^2)^{\otimes n'}$.
Therefore, after projection to a fixed-$j$ subspace, the Lie algebra
\[
    \mathfrak{sv}_z \eq \alg_{\mathbb{R}}\big\{i\HTC,i\overline{\HTC}\big\}
\]
for the original $n$-qubit system is naturally isomorphic to the corresponding Lie algebra for the $n'=2j$ qubit system restricted to its fully symmetric
subspace.
The only difference is that, for general $n$, these operators are tensored with the identity on the multiplicity space $\mathbb{C}^{M(n,j)}$.
The same argument applies to the projection of $\mathfrak{g}$.

In conclusion, the validity of \cref{eq:qkr} for an arbitrary number of qubits $n'$ in the symmetric subspace $j=n'/2$, established in \cref{eq:CommSub_symmetric}, implies the validity of \cref{eq:qkr} for an arbitrary angular momentum sector $j\le n/2$ of the original $n$-qubit system, after the cutoff is matched according to the relation above.
This completes the proof of \cref{eq:qkr}.
\footnote{As a remark, any circuit synthesis method used to implement a particular unitary acting on the symmetric subspace of $n'$ qubits can be used to implement an equivalent unitary acting on the $j=n'/2$ angular momentum sector of an $n$-qubit system ($n>n'$), when projected to that sector.}

\subsection{\Cref{step:4}: Full characterization of $\mathcal{SV}_z$\\ (Proof of \cref{prop:quasi_semi_universality})}
\label{sec:accidental_sym_restricts}
As we showed in \cref{Sec:failure}, the accidental symmetry of $\HTC$ implies that, when $n\ge 3$, $\mathcal{SV}_z$ is a strict subgroup of $\SVsym$, or equivalently, $\mathfrak{sv}_z$ is a strict subalgebra of $[\g,\g]$.
In particular, in \cref{Sec:failure} we showed that each $V\in\SV_z$ satisfies
\begin{align}
    \pi_{q,j}(V)=\pi_{q',j'}\big(\widetilde{V}\big)\,,
\end{align}
where
\begin{align*}
    q = \frac{n}{2}+2j'-j \quad\text{and}\quad
    q' = \frac{n}{2}-j'+2j \,.
\end{align*}
In this section, we show that these are the only restrictions.
I.e., we prove the converse direction of \cref{prop:quasi_semi_universality}: for any unitary $V\in \SVsym$ that satisfies the above constraints, there exists $\widetilde{V}\in \mathcal{SV}_z$ such that
\begin{align}
    \pi_{q,j}(\widetilde{V})=\pi_{q,j}(V)\,\quad:\quad q\le \qmax\,.
\end{align}
First, note that this statement is equivalent to
\begin{align}\label{Eq:subset}
    \SV_z \sum_{(q,j)\in\Lambda} \Pi_{q,j} \eq \SV^{\Uo\times\Sn} \sum_{(q,j)\in\Lambda} \Pi_{q,j}\,,
\end{align}
where $\Lambda$ is an arbitrary finite subset of pairs $(q,j)$ such that, for each pair related by the accidental symmetry, at most one element of the pair appears in $\Lambda$. 
In other words, the claim is that, when sectors paired by the accidental symmetry are not considered simultaneously, the two groups $\SV_z$ and $\SV^{\Uo\times\Sn}$ are identical, in any finite set of sectors. 

To establish this claim, our key tool is again \cref{lem:semi_universality}. Indeed, in \cref{sec:subsystem_control}, we have already shown condition A in \cref{lem:semi_universality}, namely subsystem universality: \cref{prop:subspace_control} is satisfied when projected to any single sector $\H\qj$.
Furthermore, in \cref{prop:accidental_symmetry_matrix}, we showed that $\HTC$ has different energy variances in any set of sectors with the same dimension, \textit{except} for the accidental-symmetry pairs in \cref{eq:vqj_pairs_prop}.
Therefore, other than these pairs, condition B' of \cref{lem:semi_universality} is satisfied for all other sectors. 
Hence, \cref{lem:semi_universality} implies \cref{Eq:subset}, which proves \cref{prop:quasi_semi_universality}.

\subsection{\cref{step:5}: Constraints on relative phases
(Completing the proof of \cref{thm:full_characterization})}
\label{sec:sufficiency_proof_fullchar}
So far, in this section, we have focused on $\mathcal{SV}_z$, the commutator subgroup of $\mathcal{V}_z$, or equivalently, on $\mathfrak{sv}_z$, the commutator subalgebra of $\vz$.
As we already showed in \cref{sec:commutator_subgroups}, at the Lie-algebraic level this distinction is captured by
\begin{align}
    \vz
    \eq
    \mathfrak{sv}_z
    \oplus
    \left\{
    i\sum_{q,j}
    \Tr(\Pi_{q,j}J_z)
    \frac{\Pi_{q,j}}{\Tr(\Pi_{q,j})}
    \right\} \,.
\end{align}
Next, we explain how one can determine the exact form of unitaries $V\in\mathcal{V}_z$, for example relative to the basis
\begin{align}
    \ket{j,m,\alpha}\otimes \ket{k}_{\mathrm{osc}}\quad:\quad -j\le m\le j,\quad k\ge 0\,.
\end{align}
Equivalently, we determine components
\begin{align*}
    \pi_{q,j}(V)\quad :\quad V\in \mathcal{V}_z,
\end{align*}
where $\pi_{q,j}$ is defined in \cref{eq:vqj_component}.
As we have seen in \cref{eq:Dec}, $V$ is in $\mathcal{V}_z$ if and only if there exists a unitary $W\in\mathcal{SV}_z$ such that
\begin{align}
    V=\exp(i\theta_z J_z)W
\end{align}
for some $\theta_z\in[0,4\pi)$. 
This decomposition implies
\begin{align}
    \pi_{q,j}(V) \eq \pi_{q,j}\big(\exp(i\theta_z J_z)\big)\,\pi_{q,j}(W) \,.
\end{align}
Similarly, the projection of $V$ to a sector $\H_{q',j'}$ paired with $\H\qj$ by the accidental symmetry is given by
\begin{align}
    \pi_{q',j'}(V) \eq \pi_{q',j'}\big(\exp(i\theta_z J_z)\big)\,\pi_{q',j'}(W) \,.
\end{align}
Since $W\in\mathcal{SV}_z$, its actions on sectors paired by the accidental
symmetry are identified:
\begin{align}
    \pi_{q,j}(W)=\pi_{q',j'}(W) \, .
\end{align}
Furthermore, as we have seen in \cref{eq:acc_sym_Jz_matrix},
\[
    \pi\qj(J_z)=\pi_{q',j'}(J_z)+(j'-j)\1 \,,
\]
which implies
\begin{align}
    \pi\qj\big(\exp(i\theta_z J_z)\big)
    = e^{i\theta_z(j'-j)} \pi_{q',j'}\big(\exp(i\theta_z J_z)\big) \, .
\end{align}
Comparing these relations, we conclude that
\begin{align}\label{eq:cond0}
    e^{i\theta_z j}\pi_{q,j}(V) \eq e^{i\theta_z j'}\pi_{q',j'}(V) \,.
\end{align}
Furthermore,
\begin{align}
 \begin{split}
    \det(\pi_{q,j}(V))
    &= \det\big(\pi_{q,j}\big(\exp\big(i\theta_z J_z)\big)\big)\,
    \det\big(\pi_{q,j}(W)\big) \\[4pt]
    &= \det\big(\pi_{q,j}\big(\exp(i\theta_z J_z)\big)\big) \\[4pt]
    &= \exp\!\left(i\theta_z\,\Tr\big(\pi_{q,j}(J_z)\big)\right) \, ,
 \end{split}
\end{align}
where we used $\det(e^B)=e^{\Tr(B)}$ and $\det(\pi_{q,j}(W))=1$ for $W\in\mathcal{SV}_z$.

To determine $\Tr(\pi_{q,j}(J_z))$, recall that $\pi_{q,j}$ denotes the restriction to the subspace spanned by $\ket{j,m,\alpha}\otimes\ket{k}_{\mathrm{osc}}$ with fixed $q$, $j$, and $\alpha$. 
Since
\[
    Q=a^\dag a+J_z+\frac{n}{2}\,,
\]
fixing the charge $q$ implies
\[
    k=q-\frac{n}{2}-m\,.
\]
Thus, the allowed values of $m$ are
\begin{align}
    -j \,\le\, m \,\le\, \min\left\{j,\,q-\frac{n}{2}\right\}
    \,=:\,m_{\max}(q,j) \,.
\end{align}
It follows that
\begin{align}
    \Tr\big(\pi_{q,j}(J_z)\big)
    &= \sum_{m=-j}^{m_{\max}(q,j)} m  \\[6pt]\nonumber
    &= \begin{cases}
        \left(q-\frac{n}{2}+j+1\right) & :\,q<\frac{n}{2}+j\\\hspace{8pt}\times
        \frac{1}{2}\left(q-\frac{n}{2}-j\right)
        \\[6pt]
        0 & :\,q\geq \frac{n}{2}+j\,.
    \end{cases}
\end{align}
Therefore, for all $V\in\mathcal{V}_z$, there is a $\theta_z\in[0,4\pi)$, such that
\begin{align}
    \det(\pi_{q,j}(V))
    &\eq
    \exp\!\big(i\theta_z\,\Tr\big(\pi_{q,j}(J_z)\big)\big)\\[8pt]\nonumber
    &\hspace{-24pt}\eq
    \begin{cases}
        \exp\Big[\frac{i\theta_z}{2}\left(q-\frac{n}{2}+j+1\right) & :\,q<\frac{n}{2}+j\\\hspace{42pt}\times
        \left(q-\frac{n}{2}-j\right)\Big]
        \\[6pt]
        1 & :\,q\geq \frac{n}{2}+j\,,
    \end{cases}
\end{align}
Equivalently, defining
\begin{align}
    \theta_{q,j} \,:=\, \arg\det(v_{q,j}) \, ,
\end{align}
we have
\begin{align}\label{eq:theta}
    \theta_{q,j}
    =
    \theta_z\,\Tr(\pi_{q,j}(J_z))
    \quad \mod 2\pi \,,
\end{align}
or \cref{eq:thm_phase_constraint}.

\vspace{1em}
Next, we prove the converse direction.
Consider a set of unitaries $\{v\qj:q\le \qmax\}$ and an angle $\theta_z\in[0,4\pi)$ satisfying the two conditions in \cref{thm:full_characterization}, namely \cref{eq:cond0} and \cref{eq:theta}.
We show that there exists a unitary $V\in\mathcal{V}_z$ such that
\[
    \pi_{q,j}(V)=v\qj\quad:\quad
    q\le \qmax .
\]
To prove this, first note that there exists a unitary $W\in\Vsym$ such that
\[
    \pi_{q,j}(W)=v\qj \quad:\quad q\le \qmax .
\]
By assumption, the determinant conditions in \cref{eq:theta} hold for the chosen $\theta_z$. 
Now define
\begin{align}
    \widetilde{W}=e^{-i\theta_z J_z}W \,.
\end{align}
Then, for $q\le \qmax$,
\begin{align}
 \begin{split}
    \widetilde{w}\qj := \pi_{q,j}(\widetilde{W})
    &= \exp\big(-i\theta_z\pi_{q,j}(J_z)\big)\,\pi_{q,j}(W) \\[4pt]
    &= \exp\big(-i\theta_z\pi_{q,j}(J_z)\big)\,v\qj \,.
 \end{split}
\end{align}
It follows that
\begin{align}
    \det(\widetilde{w}\qj)
    &= \exp\!\left(-i\theta_z\Tr(\pi_{q,j}(J_z))\right)\det(v\qj)
    =1 \,,
\end{align}
where in the last equality we used \cref{eq:theta}.

Furthermore, the assumption that \cref{eq:cond0} holds for any pair of sectors related by the accidental symmetry implies that
\begin{align}
 \begin{split}
    \widetilde{w}\qj
    &\eq \exp\!\big(-i\theta_z\pi_{q,j}(J_z)\big)\,v\qj\\[4pt]
    &\eq \exp\!\big(-i\theta_z\pi_{q,j}(J_z)\big)\,
    e^{-i\theta_z(j-j')}\,v_{q',j'} \\[4pt]
    &\eq \exp\!\big(-i\theta_z\pi_{q',j'}(J_z)\big)\,v_{q',j'}\\[4pt]
    &\eq \widetilde{w}_{q',j'} \, .
 \end{split}
\end{align}
Here, in the second line we used the relation imposed by the accidental symmetry in \cref{eq:cond0}, and in the third line we used the corresponding relation between the representations of $J_z$ on the two paired sectors.

Therefore, the components $\widetilde{w}\qj$ satisfy the defining conditions for membership in $\SV_z$. 
Applying the second part of the characterization of $\SV_z$ in \cref{prop:quasi_semi_universality}, there exists a unitary
\begin{align*}
    \widetilde{V}\in \SV_z
    =
    \Big\langle
    \exp({-it\HTC}),\,
    \exp({-it\overline{\HTC}})
    \,:\,
    t\in\mathbb{R}
    \Big\rangle
\end{align*}
such that
\[
    \pi_{q,j}(\widetilde{V})=\widetilde{w}\qj\quad:
    \quad q\le \qmax\,.
\]
Equivalently,
\[
    \widetilde{V}\Pi^{q_\text{max}}
    =
    \widetilde{W}\Pi^{q_\text{max}}\,.
\]
Finally, define
\begin{align}
    V = e^{i\theta_z J_z}\widetilde{V} \,.
\end{align}
Since $\widetilde{V}\in\SV_z$ and $e^{i\theta_z J_z}\in\mathcal{V}_z$, we have $V\in\mathcal{V}_z$. 
Moreover, for all $q\le \qmax$,
\begin{align}
    \pi_{q,j}(V)
    &\eq \exp\big(i\theta_z\pi_{q,j}(J_z)\big)\,\pi_{q,j}(\widetilde{V}) \\[4pt]
    &\eq \exp\big(i\theta_z\pi_{q,j}(J_z)\big)\,\widetilde{w}\qj \\[4pt]
    &\eq v\qj \,.
\end{align}
This proves the desired converse statement.

\subsection{\Cref{step:6}: Semi-Universality via $J_z^2$ \\ \hspace{16pt}(Proof of part 2 of \cref{thm:Thm0})}
\label{sec:Jz2_proof}
Finally, we consider the group
\begin{align*}
    \V_z^+ := \left\langle V,\,e^{-itJ_z^2}\,:\,t\in\mathbb{R},\,V\in\mathcal{V}_z\right\rangle\,,
\end{align*}
i.e., the group of unitaries realized using $\HTC$, $J_z$, and $J_z^2$.
In particular, we prove that the commutator subgroup $\SV_z^+:=[\V_z^+,\V_z^+]$ of the group $\V_z^+$ satisfies
\begin{align}
   \SV_z^+ \Pi^{\qmax} \eq \SVsym\Pi^{\qmax} \,,
\end{align}
for any cutoff $\qmax\geq0$.

This follows from \cref{lem:semi_universality}.
First, $\V_z^+$ inherits the subsystem universality condition, namely condition A, which is already satisfied by $\V_z$; see \cref{prop:subspace_control}.
Second, recall that $\V_z$ satisfies the pairwise independence condition, condition B', for all pairs of sectors except those related by the accidental symmetry of $\HTC$.
We now show that, for such exceptional pairs, the Hamiltonian $J_z^2$ satisfies \cref{eq:cond_B'}. 
Therefore, $\V_z^+$ satisfies condition B' for all pairs of sectors.

To prove this, recall from \cref{sec:accidental_sym} that the projections of $J_z$ into two sectors $\H\qj$ and $\H_{q',j'}$ paired by the accidental symmetry differ by a scalar multiple of the identity:
\begin{align} \label{eq:rel}
    \pi\qj(J_z) \eq \pi_{q',j'}(J_z)+(j-j')\1 \,.
\end{align}
Consequently, the  variance of $J_z^2$ with respect to the maximally
mixed state
\[
    \tau\qj\,:=\,\frac{\Pi\qj}{\Tr(\Pi\qj)}
\]
is different in the two sectors.
More explicitly, defining
\[
    \Var\qj(J_z^2)
    \,:=\,
    \Tr\big((J_z^2)^2\tau\qj\big)
    -
    \Tr\big(J_z^2\tau\qj\big)^2
\]
to be the variance of $\pi_{j,q}(J_z^2)$ with respect to the maximally mixed stated $\tau\qj$, we show in \cref{app:Jz2_variance} that
\begin{align}\label{eq:var}
    \Var\qj(J_z^2)-\Var_{{q',j'}}(J_z^2)
    \eq
    \frac{4j'(j'+1)(j'-j)^2}{3} \,.
\end{align}
Since for any pair of sectors $\mathcal{H}\qj$ and $\mathcal{H}_{q',j'}$  connected by the accidental symmetry $j\neq j'$, this implies that
\begin{align}
    \Var_{\tau_{q',j'}}(J_z^2)
    \neq 
    \Var_{\tau\qj}(J_z^2)\,.
\end{align}
Therefore, $J_z^2$ satisfies \cref{eq:cond_B'} in \cref{lem:semi_universality}
for any pair of sectors related by the accidental symmetry. 
Hence, for every pair of sectors, condition B' is satisfied by either $\HTC$ or $J_z^2$.
Thus, $\SV_z^+$ satisfies both conditions in \cref{lem:semi_universality},
which implies that $\SV_z^+$ is semi-universal when projected onto any finite
set of sectors.

The phase constraint for $\V_z^+$, analogous to \cref{eq:thm_phase_constraint} in \cref{thm:full_characterization} for $\V_z$, is given in \cref{app:Jz2_phases}.

\section{Summary \& Discussion}
Recent years have seen rapidly growing interest in quantum control schemes based
on global or collective interactions. 
These include the use of global entangling gates for efficient circuit constructions \cite{Maslov_2018,van_de_Wetering_2021}, globally driven architectures for quantum processors \cite{menta_etal_ZZ_2026,Cioni_etal_2026,Katz_etal_2022,Katz_etal_2023} and recent universality results for globally controlled analog quantum simulators
\cite{hu_global_2026}.

The Tavis-Cummings model is one of the central models for describing the collective interaction of many two-level systems with a single bosonic mode, and it plays an important role in a variety of physical platforms, including cavity and circuit QED, trapped ions, and related light-matter interfaces.
Despite its apparent simplicity and its long history, the structure of the unitaries generated by this interaction is surprisingly subtle. 
In particular, while its familiar conserved quantities explain the decomposition into U(1) charge and angular-momentum sectors, they do not by themselves fully determine the relations between the induced dynamics in different sectors. 
Our results show that the model possesses additional structure, including constraints arising from the accidental symmetry discussed above, which must be taken into account in any complete controllability analysis.

In this paper, we completely characterized the group of unitary time-evolution
operators realized by the Tavis-Cummings interaction $\HTC$ (or, equivalently,  anti Tavis-Cummings $H_{\text{anti-TC}}$) and the global $z$ field $J_z$, when restricted to finite, but arbitrarily large, invariant subspaces of the combined Hilbert space
$\H_{\mathrm{qubits}}\otimes\H_{\mathrm{osc}}$ of $n$ qubits and a single bosonic mode.
Due to the symmetries of these interactions, this Hilbert space decomposes into a direct sum of invariant \textit{sectors}, each corresponding to the mutual eigenspaces of two conserved observables: the total angular momentum $J^2$ and the charge $Q=a^{\dag}a+J_z+n/2$.
Previous work by Keyl et al. \cite{Keyl_2014_control} showed that, when
projected to a single sector in the totally symmetric subspace $j=n/2$ and
with fixed charge $q$, $\HTC$ and $J_z$ realize \textit{all} unitaries within
the sector $\mathcal{H}_{q,j=n/2}$.
We extended that result by characterizing the correlations that appear between
the components of unitaries realized by $\HTC$ and $J_z$ in different sectors.

Building on this sector decomposition, our main result gives a complete finite-cutoff characterization of the group generated by $\HTC$ and $J_z$.
In particular, we showed that the dynamics is fully controllable within each sector, but that the unitaries realized in different sectors are not independent.
Instead, their determinants and, in certain cases, their relative actions are constrained by an accidental symmetry of the Tavis-Cummings interaction.
These constraints completely determine the image of the generated group on any finite collection of sectors.
We also showed that adding the Hamiltonian $J_z^2$ removes the constraints imposed by this accidental symmetry, hence restoring full independence between the sectors.

We further discuss this unexpected symmetry in greater detail in our companion paper \cite{symmetry_paper} and explain how it can be understood via the Schwinger representation of angular momentum.

These control theoretic results are useful for understanding how the Tavis-Cummings model can be used to implement useful operations in different contexts. 
For example, as we showed in \cref{sec:Vswap}, there exists a fixed unitary, realized by $\HTC$ and $J_z$, that can be used to transfer arbitrary pure states between the symmetric subspace of $n$ qubits and the lowest $n+1$ energy levels of a bosonic mode.
This can be used as a state preparation protocol for preparing desired states with finite support in an oscillator, and vice-versa, for preparing symmetric qubit states.
This method is exact and requires only collective control of the qubits, which  means it is \textit{potentially} easier to implement and/or more efficient, compared to existing methods for preparing bosonic states.
Such methods are of interest for quantum information processing using bosonic modes -- and more generally hybrid continuous-variable discrete-variable (CV-DV) systems -- a framework which  has several advantages over qubit-only (discrete) systems \cite{liu_etal_review_2024}.

Furthermore, building on our characterization of the group of realizable unitaries generated by $J_z$ and $\HTC$ in \cref{thm:full_characterization}, we addressed in \cref{lem:Jz_conserving} the following important question: \emph{assuming that the bosonic mode is initialized in the vacuum state and returned to the vacuum state at the end of the protocol, what is the set of $n$-qubit unitaries that can be realized using only $\HTC$ and $J_z$?}
Our result provides a complete characterization of this set.

In our companion paper \cite{circuit_paper}, we build on these results to investigate applications to quantum computing. 
In particular, we show that adding a global $x$ field makes it possible to realize all permutation-invariant (PI) $n$-qubit unitaries. 
Equivalently, arbitrary PI $n$-qubit unitaries can be realized by coupling the qubits to a bosonic mode initialized in its vacuum state, using the Tavis-Cummings interaction together with global $z$ and $x$ fields.

These results imply the existence of novel methods for exactly implementing $n$-qubit PI entangling gates using the Tavis-Cummings \textit{or} anti-Tavis-Cummings interaction and collective control fields. 
In \cite{circuit_paper}, we also present exact pulse sequences, or equivalently
quantum circuits, for realizing arbitrary two-qubit PI unitaries in this way,
including several useful gates such as controlled-$Z$ (CZ) and SWAP.\\

\section*{Acknowledgments}
This work was supported in part by a collaboration between the U.S. Department of Energy and other agencies.
This material is based upon work supported by the U.S. Department of Energy, Office of Science, National Quantum Information Science Research Centers, Quantum Systems Accelerator (Award No. DE-SCL0000121).
Additional support is acknowledged from DOE Office of Advanced Scientific Computing Research, under Award No. DE-SC0026321, NSF PHY-2046195, NSF QLCI grant OMA-2120757, and ARL-ARO QCISS Grant number W911NF-21-1-0005.
We thank David Jakab for helpful discussions during the early stages of this project.

\bibliography{main_bib}

\begin{thebibliography}{53}%
\makeatletter
\providecommand \@ifxundefined [1]{%
 \@ifx{#1\undefined}
}%
\providecommand \@ifnum [1]{%
 \ifnum #1\expandafter \@firstoftwo
 \else \expandafter \@secondoftwo
 \fi
}%
\providecommand \@ifx [1]{%
 \ifx #1\expandafter \@firstoftwo
 \else \expandafter \@secondoftwo
 \fi
}%
\providecommand \natexlab [1]{#1}%
\providecommand \enquote  [1]{``#1''}%
\providecommand \bibnamefont  [1]{#1}%
\providecommand \bibfnamefont [1]{#1}%
\providecommand \citenamefont [1]{#1}%
\providecommand \href@noop [0]{\@secondoftwo}%
\providecommand \href [0]{\begingroup \@sanitize@url \@href}%
\providecommand \@href[1]{\@@startlink{#1}\@@href}%
\providecommand \@@href[1]{\endgroup#1\@@endlink}%
\providecommand \@sanitize@url [0]{\catcode `\\12\catcode `\$12\catcode
  `\&12\catcode `\#12\catcode `\^12\catcode `\_12\catcode `\%12\relax}%
\providecommand \@@startlink[1]{}%
\providecommand \@@endlink[0]{}%
\providecommand \url  [0]{\begingroup\@sanitize@url \@url }%
\providecommand \@url [1]{\endgroup\@href {#1}{\urlprefix }}%
\providecommand \urlprefix  [0]{URL }%
\providecommand \Eprint [0]{\href }%
\providecommand \doibase [0]{https://doi.org/}%
\providecommand \selectlanguage [0]{\@gobble}%
\providecommand \bibinfo  [0]{\@secondoftwo}%
\providecommand \bibfield  [0]{\@secondoftwo}%
\providecommand \translation [1]{[#1]}%
\providecommand \BibitemOpen [0]{}%
\providecommand \bibitemStop [0]{}%
\providecommand \bibitemNoStop [0]{.\EOS\space}%
\providecommand \EOS [0]{\spacefactor3000\relax}%
\providecommand \BibitemShut  [1]{\csname bibitem#1\endcsname}%
\let\auto@bib@innerbib\@empty
\bibitem [{\citenamefont {Lloyd}(1993)}]{Lloyd_1993}%
  \BibitemOpen
  \bibfield  {author} {\bibinfo {author} {\bibfnamefont {S.}~\bibnamefont
  {Lloyd}},\ }\bibfield  {title} {\bibinfo {title} {A potentially realizable
  quantum computer},\ }\href {https://doi.org/10.1126/science.261.5128.1569}
  {\bibfield  {journal} {\bibinfo  {journal} {Science}\ }\textbf {\bibinfo
  {volume} {261}},\ \bibinfo {pages} {1569} (\bibinfo {year} {1993})},\ \Eprint
  {https://arxiv.org/abs/https://www.science.org/doi/pdf/10.1126/science.261.5128.1569}
  {https://www.science.org/doi/pdf/10.1126/science.261.5128.1569} \BibitemShut
  {NoStop}%
\bibitem [{\citenamefont {Benjamin}(2001)}]{Benjamin_2001}%
  \BibitemOpen
  \bibfield  {author} {\bibinfo {author} {\bibfnamefont {S.~C.}\ \bibnamefont
  {Benjamin}},\ }\bibfield  {title} {\bibinfo {title} {Quantum computing
  without local control of qubit-qubit interactions},\ }\href
  {https://doi.org/10.1103/PhysRevLett.88.017904} {\bibfield  {journal}
  {\bibinfo  {journal} {Phys. Rev. Lett.}\ }\textbf {\bibinfo {volume} {88}},\
  \bibinfo {pages} {017904} (\bibinfo {year} {2001})}\BibitemShut {NoStop}%
\bibitem [{\citenamefont {Benjamin}(2000)}]{Benjamin_2000}%
  \BibitemOpen
  \bibfield  {author} {\bibinfo {author} {\bibfnamefont {S.~C.}\ \bibnamefont
  {Benjamin}},\ }\bibfield  {title} {\bibinfo {title} {Schemes for parallel
  quantum computation without local control of qubits},\ }\href
  {https://doi.org/10.1103/PhysRevA.61.020301} {\bibfield  {journal} {\bibinfo
  {journal} {Phys. Rev. A}\ }\textbf {\bibinfo {volume} {61}},\ \bibinfo
  {pages} {020301} (\bibinfo {year} {2000})}\BibitemShut {NoStop}%
\bibitem [{\citenamefont {Fitzsimons}\ \emph {et~al.}(2007)\citenamefont
  {Fitzsimons}, \citenamefont {Xiao}, \citenamefont {Benjamin},\ and\
  \citenamefont {Jones}}]{Fitzsimons_etal_2007}%
  \BibitemOpen
  \bibfield  {author} {\bibinfo {author} {\bibfnamefont {J.}~\bibnamefont
  {Fitzsimons}}, \bibinfo {author} {\bibfnamefont {L.}~\bibnamefont {Xiao}},
  \bibinfo {author} {\bibfnamefont {S.~C.}\ \bibnamefont {Benjamin}},\ and\
  \bibinfo {author} {\bibfnamefont {J.~A.}\ \bibnamefont {Jones}},\ }\bibfield
  {title} {\bibinfo {title} {Quantum information processing with delocalized
  qubits under global control},\ }\href
  {https://doi.org/10.1103/PhysRevLett.99.030501} {\bibfield  {journal}
  {\bibinfo  {journal} {Phys. Rev. Lett.}\ }\textbf {\bibinfo {volume} {99}},\
  \bibinfo {pages} {030501} (\bibinfo {year} {2007})}\BibitemShut {NoStop}%
\bibitem [{\citenamefont {Raussendorf}(2005)}]{Raussendorf_2005}%
  \BibitemOpen
  \bibfield  {author} {\bibinfo {author} {\bibfnamefont {R.}~\bibnamefont
  {Raussendorf}},\ }\bibfield  {title} {\bibinfo {title} {Quantum computation
  via translation-invariant operations on a chain of qubits},\ }\href
  {https://doi.org/10.1103/PhysRevA.72.052301} {\bibfield  {journal} {\bibinfo
  {journal} {Phys. Rev. A}\ }\textbf {\bibinfo {volume} {72}},\ \bibinfo
  {pages} {052301} (\bibinfo {year} {2005})}\BibitemShut {NoStop}%
\bibitem [{\citenamefont {Kay}\ and\ \citenamefont
  {Pachos}(2004)}]{Kay_Pachos_2004}%
  \BibitemOpen
  \bibfield  {author} {\bibinfo {author} {\bibfnamefont {A.}~\bibnamefont
  {Kay}}\ and\ \bibinfo {author} {\bibfnamefont {J.~K.}\ \bibnamefont
  {Pachos}},\ }\bibfield  {title} {\bibinfo {title} {Quantum computation in
  optical lattices via global laser addressing},\ }\href
  {https://doi.org/10.1088/1367-2630/6/1/126} {\bibfield  {journal} {\bibinfo
  {journal} {New Journal of Physics}\ }\textbf {\bibinfo {volume} {6}},\
  \bibinfo {pages} {126} (\bibinfo {year} {2004})}\BibitemShut {NoStop}%
\bibitem [{\citenamefont {Cesa}\ and\ \citenamefont
  {Pichler}(2023)}]{cesa2023universal}%
  \BibitemOpen
  \bibfield  {author} {\bibinfo {author} {\bibfnamefont {F.}~\bibnamefont
  {Cesa}}\ and\ \bibinfo {author} {\bibfnamefont {H.}~\bibnamefont {Pichler}},\
  }\bibfield  {title} {\bibinfo {title} {Universal quantum computation in
  globally driven rydberg atom arrays},\ }\href@noop {} {\bibfield  {journal}
  {\bibinfo  {journal} {Physical Review Letters}\ }\textbf {\bibinfo {volume}
  {131}},\ \bibinfo {pages} {170601} (\bibinfo {year} {2023})}\BibitemShut
  {NoStop}%
\bibitem [{\citenamefont {Menta}\ \emph
  {et~al.}(2026{\natexlab{a}})\citenamefont {Menta}, \citenamefont {Cioni},
  \citenamefont {Aiudi}, \citenamefont {Polini},\ and\ \citenamefont
  {Giovannetti}}]{menta_etal_actuators_2026}%
  \BibitemOpen
  \bibfield  {author} {\bibinfo {author} {\bibfnamefont {R.}~\bibnamefont
  {Menta}}, \bibinfo {author} {\bibfnamefont {F.}~\bibnamefont {Cioni}},
  \bibinfo {author} {\bibfnamefont {R.}~\bibnamefont {Aiudi}}, \bibinfo
  {author} {\bibfnamefont {M.}~\bibnamefont {Polini}},\ and\ \bibinfo {author}
  {\bibfnamefont {V.}~\bibnamefont {Giovannetti}},\ }\href
  {https://arxiv.org/abs/2603.23362} {\bibinfo {title} {Global control via
  quantum actuators}} (\bibinfo {year} {2026}{\natexlab{a}}),\ \Eprint
  {https://arxiv.org/abs/2603.23362} {arXiv:2603.23362 [quant-ph]} \BibitemShut
  {NoStop}%
\bibitem [{\citenamefont {Leibfried}\ \emph {et~al.}(2003)\citenamefont
  {Leibfried}, \citenamefont {Blatt}, \citenamefont {Monroe},\ and\
  \citenamefont {Wineland}}]{Leibfried_2003_trappedion}%
  \BibitemOpen
  \bibfield  {author} {\bibinfo {author} {\bibfnamefont {D.}~\bibnamefont
  {Leibfried}}, \bibinfo {author} {\bibfnamefont {R.}~\bibnamefont {Blatt}},
  \bibinfo {author} {\bibfnamefont {C.}~\bibnamefont {Monroe}},\ and\ \bibinfo
  {author} {\bibfnamefont {D.}~\bibnamefont {Wineland}},\ }\bibfield  {title}
  {\bibinfo {title} {Quantum dynamics of single trapped ions},\ }\href
  {https://doi.org/10.1103/RevModPhys.75.281} {\bibfield  {journal} {\bibinfo
  {journal} {Rev. Mod. Phys.}\ }\textbf {\bibinfo {volume} {75}},\ \bibinfo
  {pages} {281} (\bibinfo {year} {2003})}\BibitemShut {NoStop}%
\bibitem [{\citenamefont {Häffner}\ \emph {et~al.}(2008)\citenamefont
  {Häffner}, \citenamefont {Roos},\ and\ \citenamefont
  {Blatt}}]{Haffner_2008_trappedion}%
  \BibitemOpen
  \bibfield  {author} {\bibinfo {author} {\bibfnamefont {H.}~\bibnamefont
  {Häffner}}, \bibinfo {author} {\bibfnamefont {C.}~\bibnamefont {Roos}},\
  and\ \bibinfo {author} {\bibfnamefont {R.}~\bibnamefont {Blatt}},\ }\bibfield
   {title} {\bibinfo {title} {Quantum computing with trapped ions},\ }\href
  {https://doi.org/https://doi.org/10.1016/j.physrep.2008.09.003} {\bibfield
  {journal} {\bibinfo  {journal} {Physics Reports}\ }\textbf {\bibinfo {volume}
  {469}},\ \bibinfo {pages} {155} (\bibinfo {year} {2008})}\BibitemShut
  {NoStop}%
\bibitem [{\citenamefont {Blais}\ \emph {et~al.}(2004)\citenamefont {Blais},
  \citenamefont {Huang}, \citenamefont {Wallraff}, \citenamefont {Girvin},\
  and\ \citenamefont {Schoelkopf}}]{Blais_2004_cQED}%
  \BibitemOpen
  \bibfield  {author} {\bibinfo {author} {\bibfnamefont {A.}~\bibnamefont
  {Blais}}, \bibinfo {author} {\bibfnamefont {R.-S.}\ \bibnamefont {Huang}},
  \bibinfo {author} {\bibfnamefont {A.}~\bibnamefont {Wallraff}}, \bibinfo
  {author} {\bibfnamefont {S.~M.}\ \bibnamefont {Girvin}},\ and\ \bibinfo
  {author} {\bibfnamefont {R.~J.}\ \bibnamefont {Schoelkopf}},\ }\bibfield
  {title} {\bibinfo {title} {Cavity quantum electrodynamics for superconducting
  electrical circuits: An architecture for quantum computation},\ }\href
  {https://doi.org/10.1103/PhysRevA.69.062320} {\bibfield  {journal} {\bibinfo
  {journal} {Phys. Rev. A}\ }\textbf {\bibinfo {volume} {69}},\ \bibinfo
  {pages} {062320} (\bibinfo {year} {2004})}\BibitemShut {NoStop}%
\bibitem [{\citenamefont {Wallraff}\ \emph {et~al.}(2004)\citenamefont
  {Wallraff}, \citenamefont {Schuster}, \citenamefont {Blais}, \citenamefont
  {Frunzio}, \citenamefont {Huang}, \citenamefont {Majer}, \citenamefont
  {Kumar}, \citenamefont {Girvin},\ and\ \citenamefont
  {Schoelkopf}}]{Wallraff_2004_cQED}%
  \BibitemOpen
  \bibfield  {author} {\bibinfo {author} {\bibfnamefont {A.}~\bibnamefont
  {Wallraff}}, \bibinfo {author} {\bibfnamefont {D.~I.}\ \bibnamefont
  {Schuster}}, \bibinfo {author} {\bibfnamefont {A.}~\bibnamefont {Blais}},
  \bibinfo {author} {\bibfnamefont {L.}~\bibnamefont {Frunzio}}, \bibinfo
  {author} {\bibfnamefont {R.-S.}\ \bibnamefont {Huang}}, \bibinfo {author}
  {\bibfnamefont {J.}~\bibnamefont {Majer}}, \bibinfo {author} {\bibfnamefont
  {S.}~\bibnamefont {Kumar}}, \bibinfo {author} {\bibfnamefont {S.~M.}\
  \bibnamefont {Girvin}},\ and\ \bibinfo {author} {\bibfnamefont {R.~J.}\
  \bibnamefont {Schoelkopf}},\ }\bibfield  {title} {\bibinfo {title} {Strong
  coupling of a single photon to a superconducting qubit using circuit quantum
  electrodynamics},\ }\href {https://doi.org/10.1038/nature02851} {\bibfield
  {journal} {\bibinfo  {journal} {Nature}\ }\textbf {\bibinfo {volume} {431}},\
  \bibinfo {pages} {162} (\bibinfo {year} {2004})}\BibitemShut {NoStop}%
\bibitem [{\citenamefont {Raimond}\ \emph {et~al.}(2001)\citenamefont
  {Raimond}, \citenamefont {Brune},\ and\ \citenamefont
  {Haroche}}]{Raimond_2001_cavity}%
  \BibitemOpen
  \bibfield  {author} {\bibinfo {author} {\bibfnamefont {J.~M.}\ \bibnamefont
  {Raimond}}, \bibinfo {author} {\bibfnamefont {M.}~\bibnamefont {Brune}},\
  and\ \bibinfo {author} {\bibfnamefont {S.}~\bibnamefont {Haroche}},\
  }\bibfield  {title} {\bibinfo {title} {Manipulating quantum entanglement with
  atoms and photons in a cavity},\ }\href
  {https://doi.org/10.1103/RevModPhys.73.565} {\bibfield  {journal} {\bibinfo
  {journal} {Rev. Mod. Phys.}\ }\textbf {\bibinfo {volume} {73}},\ \bibinfo
  {pages} {565} (\bibinfo {year} {2001})}\BibitemShut {NoStop}%
\bibitem [{\citenamefont {Walther}\ \emph {et~al.}(2006)\citenamefont
  {Walther}, \citenamefont {Varcoe}, \citenamefont {Englert},\ and\
  \citenamefont {Becker}}]{Walther_2006_cavity}%
  \BibitemOpen
  \bibfield  {author} {\bibinfo {author} {\bibfnamefont {H.}~\bibnamefont
  {Walther}}, \bibinfo {author} {\bibfnamefont {B.~T.~H.}\ \bibnamefont
  {Varcoe}}, \bibinfo {author} {\bibfnamefont {B.-G.}\ \bibnamefont
  {Englert}},\ and\ \bibinfo {author} {\bibfnamefont {T.}~\bibnamefont
  {Becker}},\ }\bibfield  {title} {\bibinfo {title} {Cavity quantum
  electrodynamics},\ }\href {https://doi.org/10.1088/0034-4885/69/5/R02}
  {\bibfield  {journal} {\bibinfo  {journal} {Reports on Progress in Physics}\
  }\textbf {\bibinfo {volume} {69}},\ \bibinfo {pages} {1325} (\bibinfo {year}
  {2006})}\BibitemShut {NoStop}%
\bibitem [{\citenamefont {Varcoe}\ \emph {et~al.}(2000)\citenamefont {Varcoe},
  \citenamefont {Brattke}, \citenamefont {Weidinger},\ and\ \citenamefont
  {Walther}}]{Varcoe_2000_cavity}%
  \BibitemOpen
  \bibfield  {author} {\bibinfo {author} {\bibfnamefont {B.~T.~H.}\
  \bibnamefont {Varcoe}}, \bibinfo {author} {\bibfnamefont {S.}~\bibnamefont
  {Brattke}}, \bibinfo {author} {\bibfnamefont {M.}~\bibnamefont {Weidinger}},\
  and\ \bibinfo {author} {\bibfnamefont {H.}~\bibnamefont {Walther}},\
  }\bibfield  {title} {\bibinfo {title} {Preparing pure photon number states of
  the radiation field},\ }\href {https://doi.org/10.1038/35001526} {\bibfield
  {journal} {\bibinfo  {journal} {Nature}\ }\textbf {\bibinfo {volume} {403}},\
  \bibinfo {pages} {743} (\bibinfo {year} {2000})}\BibitemShut {NoStop}%
\bibitem [{\citenamefont {Jaynes}\ and\ \citenamefont
  {Cummings}(1963)}]{Jaynes_Cummings_1963}%
  \BibitemOpen
  \bibfield  {author} {\bibinfo {author} {\bibfnamefont {E.}~\bibnamefont
  {Jaynes}}\ and\ \bibinfo {author} {\bibfnamefont {F.}~\bibnamefont
  {Cummings}},\ }\bibfield  {title} {\bibinfo {title} {Comparison of quantum
  and semiclassical radiation theories with application to the beam maser},\
  }\href {https://doi.org/10.1109/PROC.1963.1664} {\bibfield  {journal}
  {\bibinfo  {journal} {Proceedings of the IEEE}\ }\textbf {\bibinfo {volume}
  {51}},\ \bibinfo {pages} {89} (\bibinfo {year} {1963})}\BibitemShut {NoStop}%
\bibitem [{\citenamefont {Larson}\ and\ \citenamefont
  {Mavrogordatos}(2021)}]{JC_history}%
  \BibitemOpen
  \bibfield  {author} {\bibinfo {author} {\bibfnamefont {J.}~\bibnamefont
  {Larson}}\ and\ \bibinfo {author} {\bibfnamefont {T.}~\bibnamefont
  {Mavrogordatos}},\ }\href {https://doi.org/10.1088/978-0-7503-3447-1} {\emph
  {\bibinfo {title} {The Jaynes–Cummings Model and Its Descendants}}},\
  2053-2563\ (\bibinfo  {publisher} {IOP Publishing},\ \bibinfo {year}
  {2021})\BibitemShut {NoStop}%
\bibitem [{\citenamefont {Childs}\ and\ \citenamefont
  {Chuang}(2000)}]{Childs_Chuang_2000}%
  \BibitemOpen
  \bibfield  {author} {\bibinfo {author} {\bibfnamefont {A.~M.}\ \bibnamefont
  {Childs}}\ and\ \bibinfo {author} {\bibfnamefont {I.~L.}\ \bibnamefont
  {Chuang}},\ }\bibfield  {title} {\bibinfo {title} {Universal quantum
  computation with two-level trapped ions},\ }\href
  {https://doi.org/10.1103/PhysRevA.63.012306} {\bibfield  {journal} {\bibinfo
  {journal} {Phys. Rev. A}\ }\textbf {\bibinfo {volume} {63}},\ \bibinfo
  {pages} {012306} (\bibinfo {year} {2000})}\BibitemShut {NoStop}%
\bibitem [{\citenamefont {Yuan}\ and\ \citenamefont
  {Lloyd}(2007)}]{Yuan_Lloyd_2007}%
  \BibitemOpen
  \bibfield  {author} {\bibinfo {author} {\bibfnamefont {H.}~\bibnamefont
  {Yuan}}\ and\ \bibinfo {author} {\bibfnamefont {S.}~\bibnamefont {Lloyd}},\
  }\bibfield  {title} {\bibinfo {title} {Controllability of the coupled
  spin-$\frac{1}{2}$ harmonic oscillator system},\ }\href
  {https://doi.org/10.1103/PhysRevA.75.052331} {\bibfield  {journal} {\bibinfo
  {journal} {Phys. Rev. A}\ }\textbf {\bibinfo {volume} {75}},\ \bibinfo
  {pages} {052331} (\bibinfo {year} {2007})}\BibitemShut {NoStop}%
\bibitem [{\citenamefont {Cirac}\ and\ \citenamefont
  {Zoller}(1995)}]{Cirac_Zoller_1995}%
  \BibitemOpen
  \bibfield  {author} {\bibinfo {author} {\bibfnamefont {J.~I.}\ \bibnamefont
  {Cirac}}\ and\ \bibinfo {author} {\bibfnamefont {P.}~\bibnamefont {Zoller}},\
  }\bibfield  {title} {\bibinfo {title} {Quantum computations with cold trapped
  ions},\ }\href {https://doi.org/10.1103/PhysRevLett.74.4091} {\bibfield
  {journal} {\bibinfo  {journal} {Phys. Rev. Lett.}\ }\textbf {\bibinfo
  {volume} {74}},\ \bibinfo {pages} {4091} (\bibinfo {year}
  {1995})}\BibitemShut {NoStop}%
\bibitem [{\citenamefont {Tavis}\ and\ \citenamefont
  {Cummings}(1968)}]{tavis_1968_exact_solut}%
  \BibitemOpen
  \bibfield  {author} {\bibinfo {author} {\bibfnamefont {M.}~\bibnamefont
  {Tavis}}\ and\ \bibinfo {author} {\bibfnamefont {F.~W.}\ \bibnamefont
  {Cummings}},\ }\bibfield  {title} {\bibinfo {title} {Exact solution for an
  n-molecule—radiation-field hamiltonian},\ }\href
  {https://doi.org/10.1103/physrev.170.379} {\bibfield  {journal} {\bibinfo
  {journal} {Physical Review}\ }\textbf {\bibinfo {volume} {170}},\ \bibinfo
  {pages} {379–384} (\bibinfo {year} {1968})}\BibitemShut {NoStop}%
\bibitem [{\citenamefont {Tavis}\ and\ \citenamefont
  {Cummings}(1969)}]{TC2_1969}%
  \BibitemOpen
  \bibfield  {author} {\bibinfo {author} {\bibfnamefont {M.}~\bibnamefont
  {Tavis}}\ and\ \bibinfo {author} {\bibfnamefont {F.~W.}\ \bibnamefont
  {Cummings}},\ }\bibfield  {title} {\bibinfo {title} {Approximate solutions
  for an $n$-molecule-radiation-field hamiltonian},\ }\href
  {https://doi.org/10.1103/PhysRev.188.692} {\bibfield  {journal} {\bibinfo
  {journal} {Phys. Rev.}\ }\textbf {\bibinfo {volume} {188}},\ \bibinfo {pages}
  {692} (\bibinfo {year} {1969})}\BibitemShut {NoStop}%
\bibitem [{\citenamefont {Bashir}\ and\ \citenamefont {{Sebawe
  Abdalla}}(1995)}]{Bashir_Abdalla_1995}%
  \BibitemOpen
  \bibfield  {author} {\bibinfo {author} {\bibfnamefont {M.}~\bibnamefont
  {Bashir}}\ and\ \bibinfo {author} {\bibfnamefont {M.}~\bibnamefont {{Sebawe
  Abdalla}}},\ }\bibfield  {title} {\bibinfo {title} {The most general solution
  for the wave equation of the transformed tavis-cummings model},\ }\href
  {https://doi.org/https://doi.org/10.1016/0375-9601(95)00469-J} {\bibfield
  {journal} {\bibinfo  {journal} {Physics Letters A}\ }\textbf {\bibinfo
  {volume} {204}},\ \bibinfo {pages} {21} (\bibinfo {year} {1995})}\BibitemShut
  {NoStop}%
\bibitem [{\citenamefont {Bogoliubov}\ \emph {et~al.}(1996)\citenamefont
  {Bogoliubov}, \citenamefont {Bullough},\ and\ \citenamefont
  {Timonen}}]{Bogoliubov_etal_1996}%
  \BibitemOpen
  \bibfield  {author} {\bibinfo {author} {\bibfnamefont {N.~M.}\ \bibnamefont
  {Bogoliubov}}, \bibinfo {author} {\bibfnamefont {R.~K.}\ \bibnamefont
  {Bullough}},\ and\ \bibinfo {author} {\bibfnamefont {J.}~\bibnamefont
  {Timonen}},\ }\bibfield  {title} {\bibinfo {title} {Exact solution of
  generalized tavis - cummings models in quantum optics},\ }\href
  {https://doi.org/10.1088/0305-4470/29/19/015} {\bibfield  {journal} {\bibinfo
   {journal} {Journal of Physics A: Mathematical and General}\ }\textbf
  {\bibinfo {volume} {29}},\ \bibinfo {pages} {6305} (\bibinfo {year}
  {1996})}\BibitemShut {NoStop}%
\bibitem [{\citenamefont {Rybin}\ \emph {et~al.}(1998)\citenamefont {Rybin},
  \citenamefont {Kastelewicz}, \citenamefont {Timonen},\ and\ \citenamefont
  {Bogoliubov}}]{Rybin_etal_1998}%
  \BibitemOpen
  \bibfield  {author} {\bibinfo {author} {\bibfnamefont {A.}~\bibnamefont
  {Rybin}}, \bibinfo {author} {\bibfnamefont {G.}~\bibnamefont {Kastelewicz}},
  \bibinfo {author} {\bibfnamefont {J.}~\bibnamefont {Timonen}},\ and\ \bibinfo
  {author} {\bibfnamefont {N.}~\bibnamefont {Bogoliubov}},\ }\bibfield  {title}
  {\bibinfo {title} {The su(1,1) tavis-cummings model},\ }\href
  {https://doi.org/10.1088/0305-4470/31/20/009} {\bibfield  {journal} {\bibinfo
   {journal} {Journal of Physics A: Mathematical and General}\ }\textbf
  {\bibinfo {volume} {31}},\ \bibinfo {pages} {4705} (\bibinfo {year}
  {1998})}\BibitemShut {NoStop}%
\bibitem [{\citenamefont {Tessier}\ \emph {et~al.}(2003)\citenamefont
  {Tessier}, \citenamefont {Deutsch}, \citenamefont {Delgado},\ and\
  \citenamefont {Fuentes-Guridi}}]{Tessier_etal_2003}%
  \BibitemOpen
  \bibfield  {author} {\bibinfo {author} {\bibfnamefont {T.~E.}\ \bibnamefont
  {Tessier}}, \bibinfo {author} {\bibfnamefont {I.~H.}\ \bibnamefont
  {Deutsch}}, \bibinfo {author} {\bibfnamefont {A.}~\bibnamefont {Delgado}},\
  and\ \bibinfo {author} {\bibfnamefont {I.}~\bibnamefont {Fuentes-Guridi}},\
  }\bibfield  {title} {\bibinfo {title} {Entanglement sharing in the two-atom
  tavis-cummings model},\ }\href {https://doi.org/10.1103/PhysRevA.68.062316}
  {\bibfield  {journal} {\bibinfo  {journal} {Phys. Rev. A}\ }\textbf {\bibinfo
  {volume} {68}},\ \bibinfo {pages} {062316} (\bibinfo {year}
  {2003})}\BibitemShut {NoStop}%
\bibitem [{\citenamefont {Vadeiko}\ \emph {et~al.}(2003)\citenamefont
  {Vadeiko}, \citenamefont {Miroshnichenko}, \citenamefont {Rybin},\ and\
  \citenamefont {Timonen}}]{Vadeiko_etal_2003}%
  \BibitemOpen
  \bibfield  {author} {\bibinfo {author} {\bibfnamefont {I.~P.}\ \bibnamefont
  {Vadeiko}}, \bibinfo {author} {\bibfnamefont {G.~P.}\ \bibnamefont
  {Miroshnichenko}}, \bibinfo {author} {\bibfnamefont {A.~V.}\ \bibnamefont
  {Rybin}},\ and\ \bibinfo {author} {\bibfnamefont {J.}~\bibnamefont
  {Timonen}},\ }\bibfield  {title} {\bibinfo {title} {Algebraic approach to the
  tavis-cummings problem},\ }\href {https://doi.org/10.1103/PhysRevA.67.053808}
  {\bibfield  {journal} {\bibinfo  {journal} {Phys. Rev. A}\ }\textbf {\bibinfo
  {volume} {67}},\ \bibinfo {pages} {053808} (\bibinfo {year}
  {2003})}\BibitemShut {NoStop}%
\bibitem [{\citenamefont {Genes}\ \emph {et~al.}(2003)\citenamefont {Genes},
  \citenamefont {Berman},\ and\ \citenamefont {Rojo}}]{Genes_etal_2003}%
  \BibitemOpen
  \bibfield  {author} {\bibinfo {author} {\bibfnamefont {C.}~\bibnamefont
  {Genes}}, \bibinfo {author} {\bibfnamefont {P.~R.}\ \bibnamefont {Berman}},\
  and\ \bibinfo {author} {\bibfnamefont {A.~G.}\ \bibnamefont {Rojo}},\
  }\bibfield  {title} {\bibinfo {title} {Spin squeezing via atom-cavity field
  coupling},\ }\href {https://doi.org/10.1103/PhysRevA.68.043809} {\bibfield
  {journal} {\bibinfo  {journal} {Phys. Rev. A}\ }\textbf {\bibinfo {volume}
  {68}},\ \bibinfo {pages} {043809} (\bibinfo {year} {2003})}\BibitemShut
  {NoStop}%
\bibitem [{\citenamefont {Fink}\ \emph {et~al.}(2009)\citenamefont {Fink},
  \citenamefont {Bianchetti}, \citenamefont {Baur}, \citenamefont {G\"oppl},
  \citenamefont {Steffen}, \citenamefont {Filipp}, \citenamefont {Leek},
  \citenamefont {Blais},\ and\ \citenamefont {Wallraff}}]{Fink_etal_2009}%
  \BibitemOpen
  \bibfield  {author} {\bibinfo {author} {\bibfnamefont {J.~M.}\ \bibnamefont
  {Fink}}, \bibinfo {author} {\bibfnamefont {R.}~\bibnamefont {Bianchetti}},
  \bibinfo {author} {\bibfnamefont {M.}~\bibnamefont {Baur}}, \bibinfo {author}
  {\bibfnamefont {M.}~\bibnamefont {G\"oppl}}, \bibinfo {author} {\bibfnamefont
  {L.}~\bibnamefont {Steffen}}, \bibinfo {author} {\bibfnamefont
  {S.}~\bibnamefont {Filipp}}, \bibinfo {author} {\bibfnamefont {P.~J.}\
  \bibnamefont {Leek}}, \bibinfo {author} {\bibfnamefont {A.}~\bibnamefont
  {Blais}},\ and\ \bibinfo {author} {\bibfnamefont {A.}~\bibnamefont
  {Wallraff}},\ }\bibfield  {title} {\bibinfo {title} {Dressed collective qubit
  states and the tavis-cummings model in circuit qed},\ }\href
  {https://doi.org/10.1103/PhysRevLett.103.083601} {\bibfield  {journal}
  {\bibinfo  {journal} {Phys. Rev. Lett.}\ }\textbf {\bibinfo {volume} {103}},\
  \bibinfo {pages} {083601} (\bibinfo {year} {2009})}\BibitemShut {NoStop}%
\bibitem [{\citenamefont {Agarwal}\ \emph {et~al.}(2012)\citenamefont
  {Agarwal}, \citenamefont {Rafsanjani},\ and\ \citenamefont
  {Eberly}}]{Agarwal_etal_2012}%
  \BibitemOpen
  \bibfield  {author} {\bibinfo {author} {\bibfnamefont {S.}~\bibnamefont
  {Agarwal}}, \bibinfo {author} {\bibfnamefont {S.~M.~H.}\ \bibnamefont
  {Rafsanjani}},\ and\ \bibinfo {author} {\bibfnamefont {J.~H.}\ \bibnamefont
  {Eberly}},\ }\bibfield  {title} {\bibinfo {title} {Tavis-cummings model
  beyond the rotating wave approximation: Quasidegenerate qubits},\ }\href
  {https://doi.org/10.1103/PhysRevA.85.043815} {\bibfield  {journal} {\bibinfo
  {journal} {Phys. Rev. A}\ }\textbf {\bibinfo {volume} {85}},\ \bibinfo
  {pages} {043815} (\bibinfo {year} {2012})}\BibitemShut {NoStop}%
\bibitem [{\citenamefont {Zou}\ \emph {et~al.}(2013)\citenamefont {Zou},
  \citenamefont {Liu}, \citenamefont {Feng}, \citenamefont {Yang},
  \citenamefont {Chen},\ and\ \citenamefont {Twamley}}]{Zou_etal_2013}%
  \BibitemOpen
  \bibfield  {author} {\bibinfo {author} {\bibfnamefont {J.~H.}\ \bibnamefont
  {Zou}}, \bibinfo {author} {\bibfnamefont {T.}~\bibnamefont {Liu}}, \bibinfo
  {author} {\bibfnamefont {M.}~\bibnamefont {Feng}}, \bibinfo {author}
  {\bibfnamefont {W.~L.}\ \bibnamefont {Yang}}, \bibinfo {author}
  {\bibfnamefont {C.~Y.}\ \bibnamefont {Chen}},\ and\ \bibinfo {author}
  {\bibfnamefont {J.}~\bibnamefont {Twamley}},\ }\bibfield  {title} {\bibinfo
  {title} {Quantum phase transition in a driven tavis–cummings model},\
  }\href {https://doi.org/10.1088/1367-2630/15/12/123032} {\bibfield  {journal}
  {\bibinfo  {journal} {New Journal of Physics}\ }\textbf {\bibinfo {volume}
  {15}},\ \bibinfo {pages} {123032} (\bibinfo {year} {2013})}\BibitemShut
  {NoStop}%
\bibitem [{\citenamefont {Keyl}\ \emph {et~al.}(2014)\citenamefont {Keyl},
  \citenamefont {Zeier},\ and\ \citenamefont
  {Schulte-Herbrüggen}}]{Keyl_2014_control}%
  \BibitemOpen
  \bibfield  {author} {\bibinfo {author} {\bibfnamefont {M.}~\bibnamefont
  {Keyl}}, \bibinfo {author} {\bibfnamefont {R.}~\bibnamefont {Zeier}},\ and\
  \bibinfo {author} {\bibfnamefont {T.}~\bibnamefont {Schulte-Herbrüggen}},\
  }\bibfield  {title} {\bibinfo {title} {Controlling several atoms in a
  cavity},\ }\href {https://doi.org/10.1088/1367-2630/16/6/065010} {\bibfield
  {journal} {\bibinfo  {journal} {New Journal of Physics}\ }\textbf {\bibinfo
  {volume} {16}},\ \bibinfo {pages} {065010} (\bibinfo {year}
  {2014})}\BibitemShut {NoStop}%
\bibitem [{\citenamefont {Retzker}\ \emph {et~al.}(2007)\citenamefont
  {Retzker}, \citenamefont {Solano},\ and\ \citenamefont
  {Reznik}}]{Retzker_2007}%
  \BibitemOpen
  \bibfield  {author} {\bibinfo {author} {\bibfnamefont {A.}~\bibnamefont
  {Retzker}}, \bibinfo {author} {\bibfnamefont {E.}~\bibnamefont {Solano}},\
  and\ \bibinfo {author} {\bibfnamefont {B.}~\bibnamefont {Reznik}},\
  }\bibfield  {title} {\bibinfo {title} {Tavis-cummings model and collective
  multiqubit entanglement in trapped ions},\ }\bibfield  {journal} {\bibinfo
  {journal} {Physical Review A}\ }\textbf {\bibinfo {volume} {75}},\ \href
  {https://doi.org/10.1103/physreva.75.022312} {10.1103/physreva.75.022312}
  (\bibinfo {year} {2007})\BibitemShut {NoStop}%
\bibitem [{\citenamefont {S\o{}rensen}\ and\ \citenamefont
  {M\o{}lmer}(1999)}]{Molmer_Sorensen}%
  \BibitemOpen
  \bibfield  {author} {\bibinfo {author} {\bibfnamefont {A.}~\bibnamefont
  {S\o{}rensen}}\ and\ \bibinfo {author} {\bibfnamefont {K.}~\bibnamefont
  {M\o{}lmer}},\ }\bibfield  {title} {\bibinfo {title} {Quantum computation
  with ions in thermal motion},\ }\href
  {https://doi.org/10.1103/PhysRevLett.82.1971} {\bibfield  {journal} {\bibinfo
   {journal} {Phys. Rev. Lett.}\ }\textbf {\bibinfo {volume} {82}},\ \bibinfo
  {pages} {1971} (\bibinfo {year} {1999})}\BibitemShut {NoStop}%
\bibitem [{\citenamefont {Deliyannis}\ and\ \citenamefont
  {Marvian}(2026{\natexlab{a}})}]{symmetry_paper}%
  \BibitemOpen
  \bibfield  {author} {\bibinfo {author} {\bibfnamefont {P.}~\bibnamefont
  {Deliyannis}}\ and\ \bibinfo {author} {\bibfnamefont {I.}~\bibnamefont
  {Marvian}},\ }\bibfield  {title} {\bibinfo {title} {Accidental symmetry in
  the tavis-cummings model via the schwinger boson representation},\
  }\href@noop {} {\bibfield  {journal} {\bibinfo  {journal} {In preparation}\ }
  (\bibinfo {year} {2026}{\natexlab{a}})}\BibitemShut {NoStop}%
\bibitem [{\citenamefont {Marvian}(2022)}]{marvian_sym_loc_2022}%
  \BibitemOpen
  \bibfield  {author} {\bibinfo {author} {\bibfnamefont {I.}~\bibnamefont
  {Marvian}},\ }\bibfield  {title} {\bibinfo {title} {Restrictions on
  realizable unitary operations imposed by symmetry and locality},\ }\href
  {https://www.nature.com/articles/s41567-021-01464-0} {\bibfield  {journal}
  {\bibinfo  {journal} {Nature Physics}\ }\textbf {\bibinfo {volume} {18}},\
  \bibinfo {pages} {283} (\bibinfo {year} {2022})}\BibitemShut {NoStop}%
\bibitem [{\citenamefont {Marvian}\ \emph
  {et~al.}(2024{\natexlab{a}})\citenamefont {Marvian}, \citenamefont {Liu},\
  and\ \citenamefont {Hulse}}]{marvian_SU(2)}%
  \BibitemOpen
  \bibfield  {author} {\bibinfo {author} {\bibfnamefont {I.}~\bibnamefont
  {Marvian}}, \bibinfo {author} {\bibfnamefont {H.}~\bibnamefont {Liu}},\ and\
  \bibinfo {author} {\bibfnamefont {A.}~\bibnamefont {Hulse}},\ }\bibfield
  {title} {\bibinfo {title} {Rotationally invariant circuits: Universality with
  the exchange interaction and two ancilla qubits},\ }\href
  {https://doi.org/10.1103/PhysRevLett.132.130201} {\bibfield  {journal}
  {\bibinfo  {journal} {Phys. Rev. Lett.}\ }\textbf {\bibinfo {volume} {132}},\
  \bibinfo {pages} {130201} (\bibinfo {year} {2024}{\natexlab{a}})}\BibitemShut
  {NoStop}%
\bibitem [{\citenamefont {Hulse}\ \emph {et~al.}(2024)\citenamefont {Hulse},
  \citenamefont {Liu},\ and\ \citenamefont {Marvian}}]{HLM_2024_SU(d)}%
  \BibitemOpen
  \bibfield  {author} {\bibinfo {author} {\bibfnamefont {A.}~\bibnamefont
  {Hulse}}, \bibinfo {author} {\bibfnamefont {H.}~\bibnamefont {Liu}},\ and\
  \bibinfo {author} {\bibfnamefont {I.}~\bibnamefont {Marvian}},\ }\href
  {https://arxiv.org/abs/2407.21249} {\bibinfo {title} {A framework for
  semi-universality: Semi-universality of 3-qudit su(d)-invariant gates}}
  (\bibinfo {year} {2024}),\ \Eprint {https://arxiv.org/abs/2407.21249}
  {arXiv:2407.21249 [quant-ph]} \BibitemShut {NoStop}%
\bibitem [{\citenamefont {Marvian}(2024)}]{marvian_abelian_2024}%
  \BibitemOpen
  \bibfield  {author} {\bibinfo {author} {\bibfnamefont {I.}~\bibnamefont
  {Marvian}},\ }\bibfield  {title} {\bibinfo {title} {Theory of quantum
  circuits with abelian symmetries},\ }\href
  {https://doi.org/10.1103/PhysRevResearch.6.043292} {\bibfield  {journal}
  {\bibinfo  {journal} {Phys. Rev. Res.}\ }\textbf {\bibinfo {volume} {6}},\
  \bibinfo {pages} {043292} (\bibinfo {year} {2024})}\BibitemShut {NoStop}%
\bibitem [{\citenamefont {Deliyannis}\ and\ \citenamefont
  {Marvian}(2026{\natexlab{b}})}]{circuit_paper}%
  \BibitemOpen
  \bibfield  {author} {\bibinfo {author} {\bibfnamefont {P.}~\bibnamefont
  {Deliyannis}}\ and\ \bibinfo {author} {\bibfnamefont {I.}~\bibnamefont
  {Marvian}},\ }\bibfield  {title} {\bibinfo {title} {Permutation-invariant
  $n$-body gates via the tavis-cummings interaction},\ }\href@noop {}
  {\bibfield  {journal} {\bibinfo  {journal} {In preparation}\ } (\bibinfo
  {year} {2026}{\natexlab{b}})}\BibitemShut {NoStop}%
\bibitem [{\citenamefont {Zimbor\'as}\ \emph {et~al.}(2015)\citenamefont
  {Zimbor\'as}, \citenamefont {Zeier}, \citenamefont {Schulte-Herbr\"uggen},\
  and\ \citenamefont {Burgarth}}]{PhysRevA.92.042309}%
  \BibitemOpen
  \bibfield  {author} {\bibinfo {author} {\bibfnamefont {Z.}~\bibnamefont
  {Zimbor\'as}}, \bibinfo {author} {\bibfnamefont {R.}~\bibnamefont {Zeier}},
  \bibinfo {author} {\bibfnamefont {T.}~\bibnamefont {Schulte-Herbr\"uggen}},\
  and\ \bibinfo {author} {\bibfnamefont {D.}~\bibnamefont {Burgarth}},\
  }\bibfield  {title} {\bibinfo {title} {Symmetry criteria for quantum
  simulability of effective interactions},\ }\href
  {https://doi.org/10.1103/PhysRevA.92.042309} {\bibfield  {journal} {\bibinfo
  {journal} {Phys. Rev. A}\ }\textbf {\bibinfo {volume} {92}},\ \bibinfo
  {pages} {042309} (\bibinfo {year} {2015})}\BibitemShut {NoStop}%
\bibitem [{\citenamefont {Marvian}\ \emph
  {et~al.}(2024{\natexlab{b}})\citenamefont {Marvian}, \citenamefont {Liu},\
  and\ \citenamefont {Hulse}}]{MLH_2024}%
  \BibitemOpen
  \bibfield  {author} {\bibinfo {author} {\bibfnamefont {I.}~\bibnamefont
  {Marvian}}, \bibinfo {author} {\bibfnamefont {H.}~\bibnamefont {Liu}},\ and\
  \bibinfo {author} {\bibfnamefont {A.}~\bibnamefont {Hulse}},\ }\bibfield
  {title} {\bibinfo {title} {Rotationally invariant circuits: Universality with
  the exchange interaction and two ancilla qubits},\ }\href
  {https://doi.org/10.1103/PhysRevLett.132.130201} {\bibfield  {journal}
  {\bibinfo  {journal} {Phys. Rev. Lett.}\ }\textbf {\bibinfo {volume} {132}},\
  \bibinfo {pages} {130201} (\bibinfo {year} {2024}{\natexlab{b}})}\BibitemShut
  {NoStop}%
\bibitem [{\citenamefont {Bartlett}\ \emph {et~al.}(2007)\citenamefont
  {Bartlett}, \citenamefont {Rudolph},\ and\ \citenamefont
  {Spekkens}}]{Bartlett_etal_2007}%
  \BibitemOpen
  \bibfield  {author} {\bibinfo {author} {\bibfnamefont {S.~D.}\ \bibnamefont
  {Bartlett}}, \bibinfo {author} {\bibfnamefont {T.}~\bibnamefont {Rudolph}},\
  and\ \bibinfo {author} {\bibfnamefont {R.~W.}\ \bibnamefont {Spekkens}},\
  }\bibfield  {title} {\bibinfo {title} {Reference frames, superselection
  rules, and quantum information},\ }\href
  {https://doi.org/10.1103/RevModPhys.79.555} {\bibfield  {journal} {\bibinfo
  {journal} {Rev. Mod. Phys.}\ }\textbf {\bibinfo {volume} {79}},\ \bibinfo
  {pages} {555} (\bibinfo {year} {2007})}\BibitemShut {NoStop}%
\bibitem [{\citenamefont {Schirmer}\ \emph {et~al.}(2001)\citenamefont
  {Schirmer}, \citenamefont {Fu},\ and\ \citenamefont
  {Solomon}}]{Schirmer_etal_2001_controllability}%
  \BibitemOpen
  \bibfield  {author} {\bibinfo {author} {\bibfnamefont {S.~G.}\ \bibnamefont
  {Schirmer}}, \bibinfo {author} {\bibfnamefont {H.}~\bibnamefont {Fu}},\ and\
  \bibinfo {author} {\bibfnamefont {A.~I.}\ \bibnamefont {Solomon}},\
  }\bibfield  {title} {\bibinfo {title} {Complete controllability of quantum
  systems},\ }\href {https://doi.org/10.1103/PhysRevA.63.063410} {\bibfield
  {journal} {\bibinfo  {journal} {Phys. Rev. A}\ }\textbf {\bibinfo {volume}
  {63}},\ \bibinfo {pages} {063410} (\bibinfo {year} {2001})}\BibitemShut
  {NoStop}%
\bibitem [{\citenamefont {Maslov}\ and\ \citenamefont
  {Nam}(2018)}]{Maslov_2018}%
  \BibitemOpen
  \bibfield  {author} {\bibinfo {author} {\bibfnamefont {D.}~\bibnamefont
  {Maslov}}\ and\ \bibinfo {author} {\bibfnamefont {Y.}~\bibnamefont {Nam}},\
  }\bibfield  {title} {\bibinfo {title} {Use of global interactions in
  efficient quantum circuit constructions},\ }\href
  {https://doi.org/10.1088/1367-2630/aaa398} {\bibfield  {journal} {\bibinfo
  {journal} {New Journal of Physics}\ }\textbf {\bibinfo {volume} {20}},\
  \bibinfo {pages} {033018} (\bibinfo {year} {2018})}\BibitemShut {NoStop}%
\bibitem [{\citenamefont {van~de Wetering}(2021)}]{van_de_Wetering_2021}%
  \BibitemOpen
  \bibfield  {author} {\bibinfo {author} {\bibfnamefont {J.}~\bibnamefont
  {van~de Wetering}},\ }\bibfield  {title} {\bibinfo {title} {Constructing
  quantum circuits with global gates},\ }\href
  {https://doi.org/10.1088/1367-2630/abf1b3} {\bibfield  {journal} {\bibinfo
  {journal} {New Journal of Physics}\ }\textbf {\bibinfo {volume} {23}},\
  \bibinfo {pages} {043015} (\bibinfo {year} {2021})}\BibitemShut {NoStop}%
\bibitem [{\citenamefont {Menta}\ \emph
  {et~al.}(2026{\natexlab{b}})\citenamefont {Menta}, \citenamefont {Cioni},
  \citenamefont {Aiudi}, \citenamefont {Caravelli}, \citenamefont {Polini},\
  and\ \citenamefont {Giovannetti}}]{menta_etal_ZZ_2026}%
  \BibitemOpen
  \bibfield  {author} {\bibinfo {author} {\bibfnamefont {R.}~\bibnamefont
  {Menta}}, \bibinfo {author} {\bibfnamefont {F.}~\bibnamefont {Cioni}},
  \bibinfo {author} {\bibfnamefont {R.}~\bibnamefont {Aiudi}}, \bibinfo
  {author} {\bibfnamefont {F.}~\bibnamefont {Caravelli}}, \bibinfo {author}
  {\bibfnamefont {M.}~\bibnamefont {Polini}},\ and\ \bibinfo {author}
  {\bibfnamefont {V.}~\bibnamefont {Giovannetti}},\ }\bibfield  {title}
  {\bibinfo {title} {Building globally controlled quantum processors with $zz$
  interactions},\ }\href {https://doi.org/10.1103/lz5d-lnz2} {\bibfield
  {journal} {\bibinfo  {journal} {Phys. Rev. A}\ }\textbf {\bibinfo {volume}
  {113}},\ \bibinfo {pages} {012614} (\bibinfo {year}
  {2026}{\natexlab{b}})}\BibitemShut {NoStop}%
\bibitem [{\citenamefont {Cioni}\ \emph {et~al.}(2026)\citenamefont {Cioni},
  \citenamefont {Menta}, \citenamefont {Aiudi}, \citenamefont {Polini},\ and\
  \citenamefont {Giovannetti}}]{Cioni_etal_2026}%
  \BibitemOpen
  \bibfield  {author} {\bibinfo {author} {\bibfnamefont {F.}~\bibnamefont
  {Cioni}}, \bibinfo {author} {\bibfnamefont {R.}~\bibnamefont {Menta}},
  \bibinfo {author} {\bibfnamefont {R.}~\bibnamefont {Aiudi}}, \bibinfo
  {author} {\bibfnamefont {M.}~\bibnamefont {Polini}},\ and\ \bibinfo {author}
  {\bibfnamefont {V.}~\bibnamefont {Giovannetti}},\ }\bibfield  {title}
  {\bibinfo {title} {Conveyor-belt superconducting quantum computer},\ }\href
  {https://doi.org/10.1103/6zzp-ctyx} {\bibfield  {journal} {\bibinfo
  {journal} {Phys. Rev. A}\ }\textbf {\bibinfo {volume} {113}},\ \bibinfo
  {pages} {012439} (\bibinfo {year} {2026})}\BibitemShut {NoStop}%
\bibitem [{\citenamefont {Katz}\ \emph {et~al.}(2022)\citenamefont {Katz},
  \citenamefont {Cetina},\ and\ \citenamefont {Monroe}}]{Katz_etal_2022}%
  \BibitemOpen
  \bibfield  {author} {\bibinfo {author} {\bibfnamefont {O.}~\bibnamefont
  {Katz}}, \bibinfo {author} {\bibfnamefont {M.}~\bibnamefont {Cetina}},\ and\
  \bibinfo {author} {\bibfnamefont {C.}~\bibnamefont {Monroe}},\ }\bibfield
  {title} {\bibinfo {title} {$n$-body interactions between trapped ion qubits
  via spin-dependent squeezing},\ }\href
  {https://doi.org/10.1103/PhysRevLett.129.063603} {\bibfield  {journal}
  {\bibinfo  {journal} {Phys. Rev. Lett.}\ }\textbf {\bibinfo {volume} {129}},\
  \bibinfo {pages} {063603} (\bibinfo {year} {2022})}\BibitemShut {NoStop}%
\bibitem [{\citenamefont {Katz}\ \emph {et~al.}(2023)\citenamefont {Katz},
  \citenamefont {Cetina},\ and\ \citenamefont {Monroe}}]{Katz_etal_2023}%
  \BibitemOpen
  \bibfield  {author} {\bibinfo {author} {\bibfnamefont {O.}~\bibnamefont
  {Katz}}, \bibinfo {author} {\bibfnamefont {M.}~\bibnamefont {Cetina}},\ and\
  \bibinfo {author} {\bibfnamefont {C.}~\bibnamefont {Monroe}},\ }\bibfield
  {title} {\bibinfo {title} {Programmable $n$-body interactions with trapped
  ions},\ }\href {https://doi.org/10.1103/PRXQuantum.4.030311} {\bibfield
  {journal} {\bibinfo  {journal} {PRX Quantum}\ }\textbf {\bibinfo {volume}
  {4}},\ \bibinfo {pages} {030311} (\bibinfo {year} {2023})}\BibitemShut
  {NoStop}%
\bibitem [{\citenamefont {Hu}\ \emph {et~al.}(2026)\citenamefont {Hu},
  \citenamefont {Gomez}, \citenamefont {Chen}, \citenamefont {Trowbridge},
  \citenamefont {Goldschmidt}, \citenamefont {Manchester}, \citenamefont
  {Chong}, \citenamefont {Jaffe},\ and\ \citenamefont
  {Yelin}}]{hu_global_2026}%
  \BibitemOpen
  \bibfield  {author} {\bibinfo {author} {\bibfnamefont {H.-Y.}\ \bibnamefont
  {Hu}}, \bibinfo {author} {\bibfnamefont {A.~M.}\ \bibnamefont {Gomez}},
  \bibinfo {author} {\bibfnamefont {L.}~\bibnamefont {Chen}}, \bibinfo {author}
  {\bibfnamefont {A.}~\bibnamefont {Trowbridge}}, \bibinfo {author}
  {\bibfnamefont {A.~J.}\ \bibnamefont {Goldschmidt}}, \bibinfo {author}
  {\bibfnamefont {Z.}~\bibnamefont {Manchester}}, \bibinfo {author}
  {\bibfnamefont {F.~T.}\ \bibnamefont {Chong}}, \bibinfo {author}
  {\bibfnamefont {A.}~\bibnamefont {Jaffe}},\ and\ \bibinfo {author}
  {\bibfnamefont {S.~F.}\ \bibnamefont {Yelin}},\ }\href
  {https://arxiv.org/abs/2508.19075} {\bibinfo {title} {Universal dynamics with
  globally controlled analog quantum simulators}} (\bibinfo {year} {2026}),\
  \Eprint {https://arxiv.org/abs/2508.19075} {arXiv:2508.19075 [quant-ph]}
  \BibitemShut {NoStop}%
\bibitem [{\citenamefont {Liu}\ \emph {et~al.}(2024)\citenamefont {Liu},
  \citenamefont {Singh}, \citenamefont {Smith}, \citenamefont {Crane},
  \citenamefont {Martyn}, \citenamefont {Eickbusch}, \citenamefont {Schuckert},
  \citenamefont {Li}, \citenamefont {Sinanan-Singh}, \citenamefont {Soley},
  \citenamefont {Tsunoda}, \citenamefont {Chuang}, \citenamefont {Wiebe},\ and\
  \citenamefont {Girvin}}]{liu_etal_review_2024}%
  \BibitemOpen
  \bibfield  {author} {\bibinfo {author} {\bibfnamefont {Y.}~\bibnamefont
  {Liu}}, \bibinfo {author} {\bibfnamefont {S.}~\bibnamefont {Singh}}, \bibinfo
  {author} {\bibfnamefont {K.~C.}\ \bibnamefont {Smith}}, \bibinfo {author}
  {\bibfnamefont {E.}~\bibnamefont {Crane}}, \bibinfo {author} {\bibfnamefont
  {J.~M.}\ \bibnamefont {Martyn}}, \bibinfo {author} {\bibfnamefont
  {A.}~\bibnamefont {Eickbusch}}, \bibinfo {author} {\bibfnamefont
  {A.}~\bibnamefont {Schuckert}}, \bibinfo {author} {\bibfnamefont {R.~D.}\
  \bibnamefont {Li}}, \bibinfo {author} {\bibfnamefont {J.}~\bibnamefont
  {Sinanan-Singh}}, \bibinfo {author} {\bibfnamefont {M.~B.}\ \bibnamefont
  {Soley}}, \bibinfo {author} {\bibfnamefont {T.}~\bibnamefont {Tsunoda}},
  \bibinfo {author} {\bibfnamefont {I.~L.}\ \bibnamefont {Chuang}}, \bibinfo
  {author} {\bibfnamefont {N.}~\bibnamefont {Wiebe}},\ and\ \bibinfo {author}
  {\bibfnamefont {S.~M.}\ \bibnamefont {Girvin}},\ }\href
  {https://arxiv.org/abs/2407.10381} {\bibinfo {title} {Hybrid oscillator-qubit
  quantum processors: Instruction set architectures, abstract machine models,
  and applications}} (\bibinfo {year} {2024}),\ \Eprint
  {https://arxiv.org/abs/2407.10381} {arXiv:2407.10381 [quant-ph]} \BibitemShut
  {NoStop}%
\bibitem [{\citenamefont {Ballentine}(2015)}]{ballentine2015quantum}%
  \BibitemOpen
  \bibfield  {author} {\bibinfo {author} {\bibfnamefont {L.}~\bibnamefont
  {Ballentine}},\ }\href {https://books.google.com/books?id=2JShngEACAAJ}
  {\emph {\bibinfo {title} {Quantum Mechanics: A Modern Development}}},\ G -
  Reference,Information and Interdisciplinary Subjects Series\ (\bibinfo
  {publisher} {World Scientific},\ \bibinfo {year} {2015})\BibitemShut
  {NoStop}%
\end{thebibliography}%

\newpage 
\onecolumngrid
\appendix

\newpage
\newcommand\appitemtwo[2]{
{\textbf{\cref{#1} \nameref*{#1}}} \dotfill \pageref{#1}\\ \begin{minipage}[t]{0.9\textwidth} #2\end{minipage}}
\newcommand\appitem[1]{\hyperref[{#1}]
{\textbf{\cref{#1}}} \textbf{\nameref*{#1}}
\dotfill \pageref{#1}\vspace{5pt}}
\newcommand\subappitem[1]{
{\textbf{\ref{#1}.} \nameref*{#1}} \hspace*{\fill} \pageref{#1}}
\makeatletter
\newcommand{\appsec}[2]{
  \section{#1}
  \def\@currentlabelname{#1}
  \def\@currentlabel{\thesection}
  \label{#2}
  \addcontentsline{toc}{section}{#1}
}
\newcommand{\appsubsec}[2]{
  \subsection{#1}
  \def\@currentlabelname{#1}
  \def\@currentlabel{\thesubsection}
  \label{#2}
}
\newcommand{\appsubsubsec}[2]{
  \refstepcounter{subsubsection}
  \subsubsection{#1}%
  \addtocounter{subsubsection}{-1}
  \def\@currentlabelname{#1}
  \def\@currentlabel{\thesubsubsection}
  \label{#2}
}
\makeatother

\section*{Appendix: Table of Contents}
\begin{itemize}[label={}]
\item \appitem{app:notation}
\item \appitem{app:computations}
\subitem \subappitem{app:matrix_elements}
\subitem \subappitem{app:matrix_elements_sym}
\subitem \subappitem{app:energy_variance}
\subitem \subappitem{app:Jz2_variance}
\subitem \subappitem{app:charge_vectors}
\item \appitem{app:AccSym_Prop}
\item \appitem{app:simple_lie_decomp}
\item \appitem{app:comm_subalg_proof}
\item \appitem{app:F_proof}
\item \appitem{app:subsystem_control}
\item \appitem{app:Jz2_phases}
\end{itemize}

\newpage
\appsec{Notation Summary}{app:notation}
\Cref{tab:notation} summarizes commonly used notation in this paper.
\begin{table}[htp]
\centering
\renewcommand{\arraystretch}{1.8}
\setlength{\tabcolsep}{6pt}
\begin{tabular}{l|l}
     $n$ & Number of qubits\\
     $j$ & Label for total angular momentum, or equivalently irreps of $\Sn$\\
     $\jmin$ & Minimum angular momentum for an $n$-qubit system: $\jmin=0$ for even $n$ and $\jmin=1/2$ for odd $n$\\
     $m$ & Eigenvalues of $J_z$ \\
     $q$ & Eigenvalues of charge operator $Q=J_z+\Nhat+\frac{n}{2}$ \\
     $k$ & Eigenvalues of oscillator number operator $\Nhat$ \\
     $\V_z$ & Group of unitaries realized using $\HTC$ and $\Zhat$\\
     $\mathcal{W}_z$ & Group of unitaries realized using $\HTC$, $\Zhat$, and $\Nhat$\\
     $\SV$ & Commutator subgroup $[\V,\V]=\{V_1V_2V_1^{\dag}V_2^{\dag}:V_1,V_2\in\V\}$ of a group $\V$\\
     $\Vsym$ & Group of permutation-invariant and U(1)-invariant unitaries\\
     $\UTC(r) := \exp(-ir\HTC/\gTC)$ & Time evolution operator of $\HTC$\\
     $\UZ(\theta) := \exp(-i \theta\Zhat)$ & Time evolution operator of $\Zhat$\\
     $\bigoplus_{j=\jmin}^{n/2}\left(\C^{M(n,j)}\otimes\bigoplus_{q=n/2-j}^{\infty}\H\qj\right)$
 & Decomposition of $(\C^2)^{\otimes n}\otimes\mathcal{L}^2(\R)$ into irreps of $\Uo\times\Sn$ \\
     $\H\qj\simeq\C^{d_n(q,j)}$ & Irrep of $\Uo\times\Sn$, labeled by charge $q$ and angular momentum $j$\\
     $\Pi\qj$ & Projector to $\C^{M(n,j)}\otimes\H\qj$, the subspace with charge $q$ and angular momentum $j$\\
     $\pi\qj(A)$ & Projection of a linear operator $A$ to $\H\qj$ \\
     $d_n(q,j)\sim d(q,j)$ & Dimension of $\H\qj$ \\
     $M(n,j)$ & Multiplicity of angular momentum $j$ irrep. for $n$-qubit system\\
     $\theta\qj:=\arg\det(\pi\qj (V))$ & Phase of the component of any unitary $V\in\V_z$ in $\H\qj$ \\
     $\Hsym\simeq\C^{n+1}$ & Restriction to the $n$-qubit symmetric subspace \\
     $\PiSym=\Pi_{j=n/2}^{\qmax}$ & Projector to $\bigoplus_{q=0}^{\qmax}\H_{q,j=n/2}$, the subspace of $\Hsym\otimes\H_{\text{osc}}$ with charge at most $\qmax$\\
     $\Pi^{\qmax}$ & Projector to $\bigoplus_{j=\jmin}^{n/2}\left(\C^{M(n,j)}\otimes\bigoplus_{q=n/2-j}^{\qmax}\H\qj\right)$, the subspace with charge at most $\qmax$
\end{tabular}
 \caption{Notation reference.}
    \label{tab:notation}
\end{table}
\renewcommand{\arraystretch}{1.0}

\newpage
\appsec{Useful Properties of TC Hamiltonian: Matrix elements, Energy Variance}{app:computations}
Recall the basis, introduced in \cref{sec:CC_PI_unitaries}, of the combined qubit-oscillator Hilbert space:
\begin{align}
    \ket{j,m,\al}\otimes\kosc\begin{cases}j&=\jmin,\dotsc,\frac{n}{2}\\[3pt]
    m&=-j,\dotsc,j\\[3pt]
    \al&=1,\dotsc,M(n,j)\\[3pt]
    k&=0,\dotsc,\infty
    \end{cases}\,,
\end{align}
where
\bes
\begin{align}
    J^2 \ket{j,m,\alpha}&=j(j+1) \ket{j,m,\alpha}\\[4pt] 
    J_z \ket{j,m,\alpha}&=m \ket{j,m,\alpha}\\[2pt]
    a^{\dagger}a\kosc &= k\kosc\,,
\end{align}
\ees
and $\al$ is a multiplicity index over $M(n,j)$, the dimension of the subspace with total angular momentum $j$ and $z$-component $m$ appearing in $(\C^2)^{\otimes n}$.
Since the multiplicity label $\al$ is irrelevant for computations, it is omitted from the notation in this appendix.
We use the simplified notation:
\begin{align}
    \ket{j,m}\otimes\kosc \,&:=\,\ket{j,m,k} \qquad\text{for general angular momentum}\,j\\[6pt]
    \left|j=\frac{n}{2},m\right\rangle\otimes\kosc \,&:=\,\ket{m,k} \qquad\hspace{8pt}\text{for the symmetric subspace}\,j=\frac{n}{2}\,.
\end{align}
These basis states are also eigenstates of $Q$, with
\begin{align}
    Q\,\ket{j,m,k} := q\,\ket{j,m,k} &\eq \left(m+k+\frac{n}{2}\right)\ket{j,m,k}\,.
\end{align}
Therefore, $(j,\al)$ together with any two of $(q,m,k)$ are sufficient to label the basis.
It is useful to write certain quantities in terms of the charge $q$, so for this we use the relabeled basis $\ket{q,j,k}$, which is defined by
\begin{equation}
 \begin{alignedat}{2}
    J^2\ket{q,j,k} &\eq j(j+1)\,\ket{q,j,k}\\[-14pt]
    Q\ket{q,j,k} &\eq q\,\ket{q,j,k} \hspace{80pt} \begin{cases}
        j\eq \jmin,\dotsc,\frac{n}{2} \\[4pt]
        q\eq 0,1,\dotsc,\infty \\[4pt]
        k\eq \max(0,q-n),\dotsc,q
    \end{cases}\\[-14pt]
    a^{\dagger}a\ket{q,j,k} &\eq k\,\ket{q,j,k} \\[4pt]
    J_z\ket{q,j,k} &\eq \left(q-k-\frac{n}{2}\right)\ket{q,j,k}\,.
 \label{eq:jqk_basis}
 \end{alignedat}
\end{equation}

\appsubsec{Matrix Elements of $\HTC$, $\Zhat$, $a^{\dagger}a$}{app:matrix_elements}
Recall the form of the Tavis-Cummings Hamiltonian,
\begin{align}
 \begin{split}
    \HTC= \frac{\gTC}{2}\sum_{i=1}^{n} \Big(\sigma_+^{(i)}{a} \p \sigma_-^{(i)}{a}^{\dag}\Big)= \gTC\,({J}_+{a} + {J}_-{a}^{\dag})\, ,
 \end{split}
\end{align}
For the remainder of this appendix, we set $g_{\text{TC}}=1$ to avoid unnecessary notational clutter.
Write $a$ and $J_+$ matrix elements \cite{ballentine2015quantum},
\bes
\begin{align}
    a^\dag\ket{k}&=\sqrt{k+1} \ket{k+1}\\[4pt] 
    J_+\ket{j,m}&=\sqrt{(j+m+1)(j-m)}\,\ket{j,m+1} \\[2pt]
    J_-\ket{j,m}&=\sqrt{(j-m+1)(j+m)}\,\ket{j,m-1} \,.
\end{align}
\ees
Therefore, in the $\ket{j,m,k}$ basis, $\HTC$ is a tri-diagonal, symmetric matrix with zeros on the diagonal.
For general angular momenta $\jmin \leq j \leq n/2$, the only non-zero matrix elements are
\begin{align}
    \ipo{j,m,k}{\,\HTC\,}{j,m-1,k+1} &\eq \ipo{j,m-1,k+1}{\,\HTC\,}{j,m,k} \qeq \sqrt{(j+m)(j-m+1)(k+1)}\,,
 \label{eq:HTC_general_elements}
\end{align}
where $-j+1\leq m \leq j\,\,$ and $\,\,0\leq k<\infty$.\\ \\
Or, in the $\ket{q,j,k}$ basis,
\begin{align}
    \ipo{q,j,k}{\,\HTC\,}{q,j,k+1} &\eq \ipo{q,j,k+1}{\,\HTC\,}{q,j,k} \qeq \sqrt{(q+j-n/2-k)(k+1)(j-q+n/2+k+1)}\,,
 \label{eq:HTC_matrix_general}
\end{align}
where $\frac{n}{2}-j \,\leq\, q<\,\infty\,\,$ and $\,\,\max(0,q-j-n/2)\,\leq\,k\,\leq\,q+j-n/2-1$.

\appsubsec{Matrix Elements in the Symmetric Subspace}{app:matrix_elements_sym}
For the $j=n/2$ symmetric subspace, the nonzero matrix elements of $\HTC$ are
\begin{align*}
    \ipo{m,k}{\,\HTC\,}{m-1,k+1} &\eq \ipo{m-1,k+1}{\,\HTC\,}{m,k} \qeq \sqrt{(n/2+m)(n/2-m+1)(k+1)}\,,
\end{align*}
with $-n/2+1\leq m \leq n/2\,\,$ and $\,\,0\leq k\leq \infty$.\\ \\
In the $\ket{q,k}$ basis,
\begin{align}
    \ipo{q,k}{\,\HTC\,}{q,k+1} &\eq \ipo{q,k+1}{\,\HTC\,}{q,k} \qeq \sqrt{(q-k)(k+1)(n-q+k+1)}\,,
 \label{eq:HTC_matrix}
\end{align}
where $0\leq q\leq\infty\,\,$ and $\,\,\max(0,q-n)\,\leq\,k\,\leq\,q-1$.

\appsubsec{Energy variance of $\HTC$}{app:energy_variance}
Consider the maximally-mixed state within $\C^{M(n,j)}\otimes\H_{q,j}$:
\begin{align}
    \tau_{q,j}\,:=\,\frac{\Pi_{q,j}}{\Tr(\Pi_{q,j})} \,\eq\,\frac{\Pi\qj}{d_n(q,j)M(n,j)}\,. 
\end{align}
First, consider the symmetric subspace, which appears with multiplicity one, i.e. $M(n,j=n/2)=1$, so $\Tr(\pi_{q,j=n/2}(\HTC))=\Tr(\HTC\Pi_{q,j=n/2})$.
Then, using \cref{eq:HTC_matrix}, compute for any charge $q$,
\begin{align}
 \begin{split}
    \Tr\left(\HTC^2\tau_{q,\,j=n/2}\right) &\,:=\, \frac{1}{\Tr(\Pi_{q,\,j=n/2})}\Tr\left(\HTC^2\Pi_{q,\,j=n/2}\right)\\
    &\eq \frac{1}{\dim(\H_{q,\,j=n/2})}\sum_{k,k'=\max\{q-n,0\}}^{q} |\langle q,k'|\HTC|q,k\rangle|^2 \\[4pt]
    &\eq \frac{2}{\min\{n+1,q+1\}}\sum_{k=\max\{q-n,0\}}^{q-1} |\langle q,k|\HTC|q,k+1\rangle|^2 \\
    &\eq \frac{2}{\min\{n+1,q+1\}}\sum_{k=\max\{q-n,0\}}^{q-1}(q-k)(k+1)(n-q+k+1) \\[6pt]
    &\eq \begin{cases}
        {n(n+2)(2q+1-n)}/{6} & :\,\,q\geq n\qquad\text{(filled)}\\[6pt]
        {q(q+2)(2n-q+1)}/{6}& :\,\,q<n\qquad\text{(unfilled)}\,.
    \end{cases}
 \end{split}
\end{align}
Now consider arbitrary angular momentum $\jmin\leq j\leq n/2$, which in general appears with multiplicity $M(n,j)$ within $(\C^2)^{\otimes n}$.
Therefore, each sector $\H\qj$ appears with multiplicity $M(n,j)$, so $\Tr(\Pi\qj)=d_n(q,j)M(n,j)$, and
\begin{align}
    \Tr\left(\HTC^2\tau\qj\right) = \frac{\Tr\big(\HTC^2\Pi\qj\big)}{d_n(q,j)M(n,j)} = \frac{\Tr\big(M(n,j)\pi\qj(\HTC^2)\big)}{d_n(q,j)M(n,j)} = \frac{\Tr\big(\pi\qj(\HTC^2)\big)}{d_n(q,j)}\,.
 \label{eq:TrH^2_multiplicity}
\end{align}
Then, the equivalences explained in \cref{sec:semi_universality_arbitraryJ}, imply that the matrices $\pi\qj^{(n)}(\HTC^2)$ for an $n$-qubit system and $\pi_{q-n/2+j,\,j}^{(2j)}(\HTC^2)$ for a $2j$-qubit system are identical.
Using this fact, and \cref{eq:TrH^2_multiplicity}, compute
\begin{align}
 \begin{split}
    \Tr\left(\HTC^2\tau\qjn\right) &\eq \frac{\Tr\left(\pi\qj^{(n)}\big(\HTC^2\big)\right)}{d_n(q,j)} \\[10pt]
    &\eq \frac{\Tr\left(\pi_{q-n/2+j,\,j}^{(2j)}\big(\HTC^2\big)\right)}{d_{2j}(q-n/2+j,j)}
    \\[10pt]
    &\eq \Tr\left(\HTC^2\tau_{q-n/2+j,\,j}^{(2j)}\right) \\[10pt]
    &\eq\begin{cases}
        {2j(j+1)(2q-n+1)}/{3} &:\,\,q\geq n/2+j\qquad\text{(filled)} \\[6pt]
        {(q-\frac{n}{2}+j)(q-\frac{n}{2}+j+2)(3j-q+\frac{n}{2}+1)}/{6}&:\,\,q<n/2+j
    \qquad\text{(unfilled)}\,.\end{cases}
 \end{split}
 \label{eq:TC_energyVar_general}
\end{align}

\appsubsec{Variance of $J_z^2$ in Accidental Symmetry Sectors}{app:Jz2_variance}
To see \cref{eq:var}, first note that \cref{eq:rel} implies
\begin{align}
 \begin{split}
    \pi_{q,j}(J_z^2)
    = \pi_{q',j'}(J_z^2)
    + (j-j')^2\1 
    + 2(j-j')\pi_{q',j'}(J_z)\,.
 \end{split}
\end{align}
The constant term does not affect the variance. 
Hence, using
\[
    \Var(X+Y)
    \eq \Var(X)+\Var(Y)+2\operatorname{Cov}(X,Y)\,,
\]
we find
\begin{align}
    \Var_{q,j}(J_z^2)
    \eq \Var_{q',j'}(J_z^2)
    + 4(j-j')^2\Var_{q',j'}(J_z)
    + 4(j-j')\operatorname{Cov}_{q',j'}(J_z^2,J_z)\,.
\end{align}
Here the covariance is evaluated in the maximally mixed state on the sector $\mathcal{H}_{q',j'}$.
Now recall that, for the paired sectors $\mathcal{H}_{q,j}$ and $\mathcal{H}_{q',j'}$, with $j'<j$, the operator $\pi_{q',j'}(J_z)$ has eigenvalues
\[ -j',-j'+1,\,\ldots,\,j'\,. \]
Therefore, all odd moments vanish:
\[ \Tr(J_z^{2r+1}\tau_{q',j'}) \eq 0\,. \]
In particular,
\[ 
    \operatorname{Cov}_{q',j'}(J_z^2,J_z)
    =
    \left\langle J_z^3\right\rangle_{q',j'}
    -
    \left\langle J_z^2\right\rangle_{q',j'}\left\langle J_z\right\rangle_{q',j'}
    =0\,.
\]
It follows that
\begin{align}
 \begin{split}
    \Var_{q,j}(J_z^2)-\Var_{q',j'}(J_z^2)&\eq 4(j-j')^2\Var_{q',j'}(J_z)\\[4pt]
    &\eq 4(j-j')^2\frac{1}{2j'+1}\sum_{m=-j'}^{j'}m^2\\[4pt]
    &\eq \frac{4j'(j'+1)(j'-j)^2}{3}\,.
 \end{split}
\end{align}
This proves \cref{eq:var}.

\appsubsec{Projections of the Hamiltonians to the center of PI, U(1)-invariant Hamiltonians (Charge Vectors)}{app:charge_vectors}
Recall that the center of the Lie algebra $\g$ of all PI, U(1)-invariant Hamiltonians is spanned by $\{i\Pi\qj\}$. 
In the following, we determine the projection of the Lie algebra generated by $i\HTC$, $iJ_z$, and $i\Nhat$ onto this center. 
(This projection, is also known as the charge vectors associated with the Hamiltonians, label the central elements of $\g$.
See \cite{marvian_sym_loc_2022} for additional details concerning charge vectors.)

Computing charge vectors directly, using \cref{eq:jqk_basis,eq:HTC_matrix} yields for the $j=n/2$ symmetric space $\Hsym$,
\begin{align}
 \begin{split}
    \Tr(\Pi_{q,j=n/2}\,\HTC) &\eq 0\\
    \Tr(\Pi_{q,j=n/2}\,\Nhat) &\eq \Big(\min \{q,n\}+1\Big)\left(q-\frac{\min \{q,n\}}{2}\right)\\
    \Tr(\Pi_{q,j=n/2}\,\Zhat) &\eq \frac{1}{2}\Big(\min \{q,n\}+1\Big)\Big(\min \{q,n\}-n\Big)\,,
 \end{split}
\end{align}
or for $q\leq n$,
\begin{align}
 \begin{split}
    \mathrm{Tr}(\Pi_{q,j=n/2}\,\Nhat) &\eq (q+1)\left(\frac{q}{2}\right)\\
    \mathrm{Tr}(\Pi_{q,j=n/2}\,\Zhat) &\eq \frac{1}{2}(q+1)(q-n)\,,
 \end{split}
 \label{eq:charge_vec1_sym}
\end{align}
and for $q>n$,
\begin{align}
 \begin{split}
    \mathrm{Tr}(\Pi_{q,j=n/2}\,\Nhat) &\eq (n+1)\left(q-\frac{n}{2}\right),\\
    \mathrm{Tr}(\Pi_{q,j=n/2}\,\Zhat) &\eq 0\,.
 \end{split}
 \label{eq:charge_vec2_sym}
\end{align}
Using the equivalences explained in \cref{sec:semi_universality_arbitraryJ} -- matrix elements of $J_z$ and $a^{\dag}a$ in an arbitrary sector $\H\qj$ of $n$ qubits are equivalent to matrix elements in sector $\H_{q',j}$ of $2j$ qubits, with $q'=q+j-n/2$ -- the charge vectors for arbitrary angular momentum $j$ subspaces of $n$ qubits are:
\begin{align}
 \begin{split}
    \mathrm{Tr}(\Pi\qj\,\Nhat) &= M(n,j)\times\mathrm{Tr}(\pi\qj(\Nhat))\\[4pt]
    \mathrm{Tr}(\Pi\qj \,\Zhat) &= M(n,j)\times\mathrm{Tr}(\pi\qj(\Zhat))\,,
 \end{split}
\end{align}
where
\begin{align}
 \begin{split}
    \mathrm{Tr}\left(\pi\qj(\Nhat)\right) &\eq \left(q+j-\frac{n}{2}+1\right)\left(\frac{q+j-n/2}{2}\right)\\[4pt]
    \mathrm{Tr}\left(\pi\qj(\Zhat)\right) &\eq \frac{1}{2}\left(q+j-\frac{n}{2}+1\right)\left(q-j-\frac{n}{2}\right)\,,
 \end{split}
 \label{eq:charge_vec1_general}
\end{align}
for $n/2-j\leq q\leq n/2+j$, and
\begin{align}
 \begin{split}
    \mathrm{Tr}\left(\pi\qj(\Nhat)\right) &\eq (2j+1)\left(q-\frac{n}{2}\right)\,,\\[4pt]
    \mathrm{Tr}\left(\pi\qj(\Zhat)\right) &\eq 0\,.
 \end{split}
 \label{eq:charge_vec2_general}
\end{align}
for $q>n/2+j$.

\appsubsec{Charge Vectors of $J_z^2$}{app:Jz2_charge_vector}
The charge vectors corresponding to $J_z^2$ are
\begin{align}
    \Tr\left(\pi\qj\big(J_z^2\big)\right) \,=\, \sum_{m=-j}^{q-n/2}m^2 \quad&=\quad\frac{1}{6}\left(j(j+1)(2j+1) + \left(q-\frac{n}{2}\right)\left(q-\frac{n}{2}+1\right)(2q-n+1)\right)\\[6pt]\nonumber
    \quad&=\quad \frac{1}{6}\left(q+j+1-\frac{n}{2}\right)\left[j(2j+1)-2j\left(q-\frac{n}{2}\right)+\left(q-\frac{n}{2}\right)(2q-n+1)\right]
\end{align}
for $q< n/2+j$, and
\begin{align}
    \Tr\left(\pi\qj\big(J_z^2\big)\right) \,=\, \sum_{m=-j}^{j}m^2 \quad=\quad \frac{1}{3}\Big(j(j+1)(2j+1)\Big)
\end{align}
for $q\geq n/2+j$.

\newpage
\appsec{An Accidental Symmetry of $\HTC$: Proof of \cref{prop:accidental_symmetry_matrix}}{app:AccSym_Prop}
\begin{proof} (\cref{prop:accidental_symmetry_matrix})
Recall from \cref{eq:Hqjn_dimension} the expression for the dimension of $\H\qj$,
\begin{align*}
 \begin{split}
  \text{dim}(\H\qj) \eq \min\left\{2j+1\,,\,\,q+1+j-\frac{n}{2}\right\} \eq \begin{cases}
    2j+1 & :\,\,q\geq n/2+j\qquad\text{(filled)}\\[4pt]
    q+1+j-\frac{n}{2} & :\,\,q< n/2+j\qquad\text{(unfilled)}\,.
  \end{cases}
 \end{split}
\end{align*}
In \cref{app:energy_variance}, we found that the variance of $\HTC$ with respect to the maximally mixed state $\tau\qj$ in sector $\H\qj$ is
\begin{align}
 \begin{split}
     \frac{\Tr\big(\pi\qj(\HTC)^2\big)}{d_n(q,j)} &\eq \Tr\left(\HTC^2\tau\qj\right)\\[6pt] &\eq \begin{cases}{2j(j+1)(2q-n+1)}/{3} & :\,\,q\geq n/2+j\quad\text{(filled)} \\[4pt]{(q-\frac{n}{2}+j)(q-\frac{n}{2}+j+2)(3j-q+\frac{n}{2}+1)}/{6}& :\,\,q< n/2+j\quad\text{(unfilled)}\,.\end{cases}
 \end{split}
 \label{eq:recall_TC_energyVar_general}
\end{align}
Recall that when $ \text{dim}(\H\qj)=2j+1$ we call the sector \emph{filled}, and otherwise if $\text{dim}(\H\qj)<2j+1$ we call it \emph{unfilled}. 
Therefore, depending on whether $\mathcal{H}_{q,j}$ and $\mathcal{H}_{q',j'}$ are filled or unfilled, we need to consider 3 different cases.\\

\noindent \textbf{(1)} If they are both filled sectors, then the equality of dimensions requires $j=j'$. 
Then, the equality of second moment implies $q=q'$, which means there are no distinct filled sectors that satisfy both constraints.\\

\noindent \textbf{(2)} If they are both unfilled then the equality of dimensions implies $q+j=q'+j'$.
This equation, together with equality of second moment (\cref{eq:recall_TC_energyVar_general}) in turn implies
\begin{align}
    3j-q = 3j'-q' \quad\implies\quad 
    j=j'\,\,\text{and}\,\,q=q'\,.
\end{align}
Again, there are no distinct sectors that satisfy both constraints.\\

\noindent \textbf{(3)} These arguments prove that for each sector $\mathcal{H}_{q,j}$ there exists, at most, one other sector $\mathcal{H}_{q',j'}$ for which both conditions in \cref{eq:variance_equality} hold, and exactly one of them should be filled and the other should be unfilled.
(Because otherwise, if there is third sector $\mathcal{H}_{q'',j''}$, then at least two of these three sectors are filled or unfilled, which has been ruled out by the above argument.) 
Therefore, suppose $\H_{q',j'}$ is filled, i.e. has dimension $2j'+1$, while $\H\qj$ is unfilled.
Then, equality of dimensions implies
\begin{align}
    q = \frac{n}{2} + 2j' - j \quad\iff\quad2j'=\left(q-\frac{n}{2}+j\right)\,.
 \label{eq:equal_dim_clause}
\end{align}
Are there any $q'\geq j'+n/2$ such that equality of second moment also holds?
Given \cref{eq:equal_dim_clause}, one finds that $q'=\frac{n}{2}-j'+2j$, or equivalently, $2j=q'+j'-\frac{n}{2}$ is the unique $q'$ that works.
Thus, \cref{eq:dimension_equality,eq:variance_equality} are both satisfied precisely for pairs of sectors related via
\begin{align}
    q=\frac{n}{2}+2j'-j \qquad\text{and}\qquad
    q'=\frac{n}{2}-j'+2j\,.
 \label{eq:q_and_q'}
\end{align}

Furthermore, for such pairs of sectors $\H\qj$ and $\H_{q',j'}$, it straightforward to see that not only are the variances are the same, but the $m$-ordered matrices of $\HTC$ in the basis $\{\ket{j,m,\al}\ket{k}_{\text{osc}}\}$ are indeed equivalent.
To see this, note that the non-zero matrix elements of $\HTC$ in an unfilled sector $\H\qj$ are, using the definition in \cref{eq:proj} and \cref{eq:HTC_general_elements},
\begin{align}
    \nonumber[\pi\qj(\HTC)]_{r,r-1} \,&:=\, \bra{j,m,\al}\otimes\bra{k}_{\text{osc}}(\HTC)\ket{j,m-1,\al}\otimes\ket{k+1}_{\text{osc}} &&\text{for}\,\, -j+1 \,\leq\, m \,\leq\, -j+2j'\\[6pt]\nonumber 
    &\eq\sqrt{(j+m)(j-m+1)(k+1)}\\[6pt]
    &\eq\sqrt{r(2j-r+1)(q+j+1-n/2-r)} &&\text{for}\,\, 1 \,\leq\, r \,\leq\, 2j'\,,
 \label{eq:matrix_jmk}
\end{align}
and in filled sector $\H_{q',j'}$ are, using \cref{eq:q_and_q'},
\begin{align}
    \nonumber
    [\pi_{q',j'}(\HTC)]_{r,r-1} &\,:=\,\bra{j',m',\al'}\otimes\bra{k'}_{\text{osc}}(\HTC)\ket{j',m'-1,\al'}\otimes\ket{k'+1}_{\text{osc}} &&\text{for}\,\, -j'+1 \,\leq\, m' \,\leq\, j'\\[4pt]\nonumber
    &\eq \sqrt{(j'+m')(j'-m'+1)(k'+1)}\\[4pt]\nonumber
    &\eq\sqrt{r(2j'-r+1)(q'+j'+1-n/2-r)} &&\text{for}\,\, 1 \,\leq\, r \,\leq\, 2j'\,,\\[4pt]
    &\eq \sqrt{r(q+j-n/2-r+1)(2j+1-r)}
    \label{eq:matrix_j'm'k'}\\[6pt]\nonumber
    &\eq [\pi\qj(\HTC)]_{r,r-1}\,.
\end{align}
and their equal Hermitian conjugate terms.
\end{proof}

\newpage
\appsec{Decomposing $\g$ and $\sv_z$ into simple Lie algebras}{app:simple_lie_decomp}
Using the characterizations of $\g$ and $\sv_z$, given \cref{sec:beyond_sym} and \cref{sec:full_char} respectively, their Lie algebras can be decomposed as follows:
\begin{corollary}\label{cor:lie_algebras}
For $\qmax\geq n$, the restriction of $\g$ to sectors with $q\le \qmax$ is,
\begin{align}
    \g\Pi^{\qmax} &\,\cong\, \left(\bigoplus_{d=2}^{n+1}\su(d)^{\oplus \mu_{d}(\g)}\right) \oplus \uu(1)^{\oplus \nu(n)}\,,
 \label{eq:g_onQmax1}
\end{align}
where
\begin{align}
    \mu_d(\g) &\eq \begin{cases}
                      \qmax-(d-2) &\,\, n-d:\text{odd}\,\\[8pt]
                      (n-d)/2+1&\,\, n-d:\text{even}\,,
                  \end{cases}
 \label{eq:g_onQmax2}
\end{align}
and
\begin{align}
 \begin{split}
    \nu(n) \,\,:=\,\,\sum_{j=\jmin}^{n/2}\left(\qmax+1-\left(\frac{n}{2}-j\right)\right) \quad=\quad \left(\left\lfloor\frac{n}{2}\right\rfloor+1\right)\left(\qmax+1-\frac{1}{2}\left\lfloor\frac{n}{2}\right\rfloor\right)\,.
 \end{split}
 \label{eq:g_onQmax3}
\end{align}
Also, for $\qmax\geq n+\lceil\frac{n}{2}-1\rceil$, (which ensures that the all restrictions due to the accidental symmetry of $\HTC$ are realized: see \cref{eq:acc_sym_boundaries}) the restrictions of $\vz$ and $\wz$ to sectors with $q\le \qmax$ are
\begin{align}
 \label{eq:vz_onQmax1}
    \vz\Pi^{\qmax} &\,\cong\, \left(\bigoplus_{d=2}^{n+1}\su(d)^{\oplus \mu_{d}(\vz)}\right) \oplus \uu(1)^{\oplus1} \\[4pt]
    \wz\Pi^{\qmax} &\,\cong\, \left(\bigoplus_{d=2}^{n+1}\su(d)^{\oplus \mu_{d}(\vz)}\right) \oplus \uu(1)^{\oplus2}\,,
 \label{eq:wz_onQmax}
\end{align}
where
\begin{align}
    \mu_d(\vz) &\eq \begin{cases}
                      \qmax-\big(n+d-3\big)/2 &\,\, n-d: \text{odd}\\[8pt]
                      ({n-d})/{2}+1&\,\, n-d: \text{even}\,.
                  \end{cases}
 \label{eq:vz_onQmax2}
\end{align}
\end{corollary}
\Cref{eq:g_onQmax1,eq:g_onQmax2,eq:g_onQmax3} follow directly from the decomposition of $\g$ in \cref{eq:g-char}, together with the dimension formula in \cref{eq:Hqjn_dimension}:
\[ \dim(\H\qj) \eq \min\left\{2j+1\,,\,\,q+1+j-\frac{n}{2}\right\} \,.\]
In particular, the number of $\mathfrak{u}(1)$ factors is precisely the total number of distinct sectors in $\g\Pi^{\qmax}$; for each $j\in[\jmin,\dotsc,n/2]$, there is one sector for each charge $q\geq n/2-j$, giving the sum in \cref{eq:g_onQmax3}.
The number $\mu_d(\g)$ of $\su(d)$ factors depends on the parity of $n-d$.
\begin{itemize}
    \item[$\bullet$] If $n-d$ is odd, then there are filled as well as unfilled sectors with dimension $d$.
    The number of unfilled sectors is $(n-d+1)/2$, the number of angular momenta $j$ satisfying $d<2j+1$.
    The number of filled sectors, i.e. sectors with dimension $d=2j+1$ and charge $q\geq n/2+j$, is $\qmax+1-(n+d-1)/2$.
    Adding these two quantities gives the first expression in \cref{eq:g_onQmax2}.
    \item[$\bullet$] If $n-d$ is even, then there are only unfilled sectors with dimension $d$; one for each distinct angular momenta $j$ satisfying $2j+1\geq d$.
    This gives the second expression in \cref{eq:g_onQmax2}.
\end{itemize}
Likewise, \Cref{eq:vz_onQmax1,eq:vz_onQmax2} follow from the characterization of $\V_z$ in \cref{thm:full_characterization} and  \cref{eq:Lie_vw}.
In this case, $\vz$ and $\wz$ have respectively just one and two $\uu(1)$ factor.
The number of $\su(d)$ factors in $\vz\Pi^{\qmax}$ and $\wz\Pi^{\qmax}$ differs from $\g\Pi^{\qmax}$ due to the accidental symmetry.
In particular, when $n-d$ is odd, the projection of $\vz$ (and $\wz$) on each unfilled sector with dimension $d$ is equal to its projection on a corresponding filled sector with the same dimension.
Therefore, the decompositions of $\vz\Pi^{\qmax}$ and $\wz\Pi^{\qmax}$ contain a single $\su(d)$ factor for each such pair of sectors related by the accidental symmetry.
To avoid double counting these sectors in \cref{eq:g_onQmax3}, we count only filled sectors for $n-d$ odd; there are $\qmax-(n+d-3)/2$ filled sectors with dimension $d$.

\vspace{1mm}
See \cref{tab:commutator_algebras} for examples of these decompositions, in the cases of $n=5$ and $n=6$ qubits.
Also, \cref{fig:6_qubit}, imported from \cite{symmetry_paper}, illustrates the sector decomposition of $\H_{\text{qubits}}\otimes\H_{\text{osc}}$ for the 6-qubit example, from which one can visually derive \cref{eq:vz_onQmax2} and the decomposition in \cref{tab:commutator_algebras}.
\begin{figure*}[htp]
    \centering
    \includegraphics[width=0.68\linewidth]{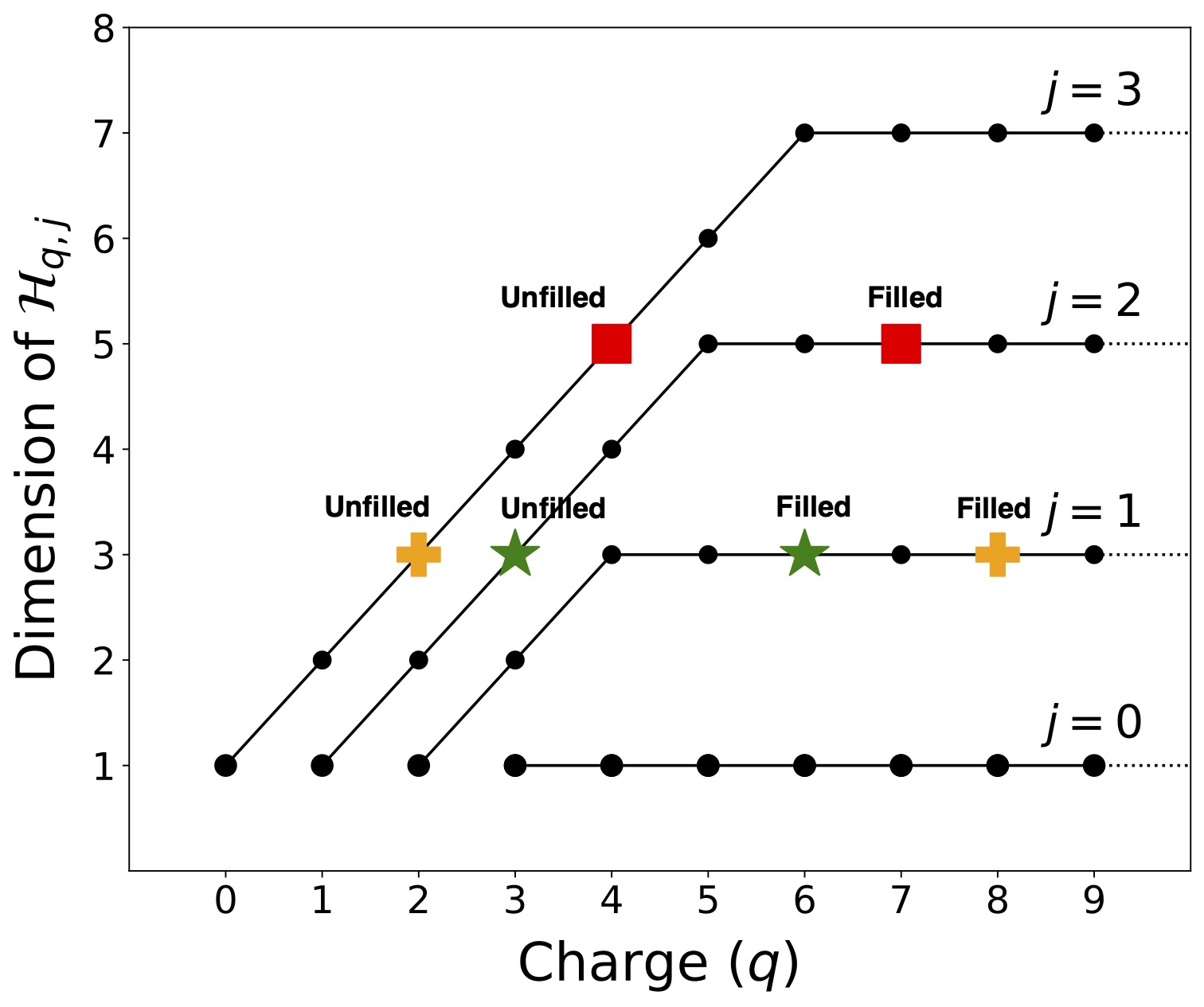}
    \caption{\textbf{Imported from \cite{symmetry_paper} -- Sector Decomposition and Accidental Symmetry of the Tavis-Cummings Hamiltonian:} (6-qubit example).}
 \label{fig:6_qubit}
\end{figure*}

\vspace{2mm}
This corollary directly gives the dimensions of the commutator subgroups $[\g,\g]\Pi^{\qmax}$ and $\sv_z\Pi^{\qmax}$.
First, assuming $\qmax\geq n$,
\begin{align}
 \begin{split}
    \dim\big([\g,\g]\Pi^{\qmax}\big) &\eq \sum_{d=2}^{n+1}\mu_d(\g)(d^2-1) \\[6pt]
    &\eq \qmax\left(\frac{1}{6}n^3+n^2+\frac{4}{3}n\right)\,-\,n\left(\frac{5}{48}n^3+\frac{1}{2}n^2+\frac{5}{24}n\right)\,+\,\begin{cases}
        \frac{\qmax+n}{2} + \frac{5}{16} & :\,n\,\,\text{odd} \\[4pt]
        \frac{3}{4}n & :\,n\,\,\text{even}\,.
    \end{cases}
 \end{split}
\end{align}
Second, assuming $\qmax\geq n + \lceil\frac{n}{2}-1\rceil$,
\begin{align}
 \begin{split}
     \dim\big(\sv_z\Pi^{\qmax}\big) &\eq \sum_{d=2}^{n+1}\mu_d(\sv_z)(d^2-1) \\[6pt]
     &\eq \qmax\left(\frac{1}{6}n^3+n^2+\frac{4}{3}n\right) \,-\,n\left(\frac{1}{8}n^3+\frac{7}{12}n^2+\frac{1}{8}n\right)\,+\,\begin{cases}
        \frac{\qmax}{2} + \frac{7n+3}{12} & :\,n\,\,\text{odd} \\[4pt]
        \frac{13}{12}n & :\,n\,\,\text{even}\,.
    \end{cases}
 \end{split}
\end{align}
Using these formulas, one can confirm that the difference between the dimensions of the commutator subalgebras of $\g$ and $\vz$ is given by \cref{eq:dimension_diff}.
See the examples of $n=5$ and $n=6$ qubits in \cref{tab:commutator_algebras}.

\newpage
\appsec{Proof of \cref{lem:comm_subalg}}{app:comm_subalg_proof}
\begin{proof}
First, note that since $\overline{B}:=[A,B]$ and $[A,[A,B]]\propto B$ with a non-zero coefficient, $\alg\{B,\overline{B}\}$ is indeed a subalgebra of $[\g,\g]$.
Next, we show that $[\g,\g]\subset\alg\{B,\overline{B}\}$ as well, and the two algebras are therefore equal.

We first use induction to prove the following property:
\begin{align}
    X\in\alg\big\{B,\overline{B}\big\} \,\,\implies\,\, [A,X]\in\alg\big\{B,\overline{B}\big\}\,.
 \label{eq:AX_property}
\end{align}
Let $\mathcal{T}_k$ be the subspace of $\alg\{B,\overline{B}\}$ that can be expressed as linear combinations of $l$-th order commutators of generators $B$ and $\overline{B}$, for all $0\leq l\leq k$. 
For the base case, we set $\mathcal{T}_0=\s\{B,\overline{B}\}$. 
This in particular, means
\begin{align}
    \mathcal{T}_{k+1}:=\s\left\{X,\,[B,X],\,[\overline{B},X]\,:\,X\in\mathcal{T}_k\right\}\,.
\end{align}
Then, using $[A,B]=\overline{B}$ and $[A,\overline{B}]=[A,[A,B]]\propto B$, we find that for $X\in \mathcal{T}_0$,
\begin{align}
    [A,X] \in \,\,&\s\left\{[A,B], [A,\overline{B}]\right\}\,=\,\s\big\{\overline{B}, B\big\}=\mathcal{T}_0\,.
\end{align}
The hypothesis assumption is that for arbitrary integer $k\ge 0$,
\begin{align}\label{eq:hyp}
    \forall X\in\mathcal{T}_{k}\,\,:\,\,[A,X] \in \mathcal{T}_{k}\,.
\end{align}
Now consider
\begin{align}
    X' \in \mathcal{T}_{k+1}=\s\left\{X,[B,X],[\overline{B},X]\,:\,X\in\mathcal{T}_k\right\}\,.
\end{align}
Therefore, 
\begin{align}
 \begin{split}
    [A,X'] \in \s\Big\{[A,X],&\big[A,[B,X]\big],\\
    &\big[A,[\overline{B},X]\big]\,:\,X\in\mathcal{T}_k\Big\}\,.
 \end{split}
\end{align}
Applying the Jacobi identity to the second order commutators, this is contained in
\begin{align}
    \s\Big\{[A,X],&\big[\overline{B},X\big],\big[B,[A,X]\big],\\\nonumber
    &\big[B,X\big],\big[\overline{B},[A,X]\big]\,:\,X\in\mathcal{T}_k\Big\}\,.
\end{align}
By definition, $[B,X],\,[\overline{B},X]\in\mathcal{T}_{k+1}$, and by the inductive hypothesis in \cref{eq:hyp}, the remaining three commutators are also contained in $\mathcal{T}_{k+1}$, so $[A,X']\in\mathcal{T}_{k+1}$.
Recall that by definition, any element of $Y\in \alg\{B,\overline{B}\}$ is a finite linear combination of finite nested commutators of $B$ and $\overline{B}$, which means it is contained in $\mathcal{T}_k$ for a sufficiently large 
$k$. 
Therefore, by induction, we conclude that for any $Y\in\alg\{B,\overline{B}\}$, $[A,Y]\in\mathcal{T}_{k+1} \subset\alg\{B,\overline{B}\}$, which proves \cref{eq:AX_property}.

Now, using this property (\cref{eq:AX_property}), we again use induction to show that $[\g,\g]\subset\alg\{B,\overline{B}\}$.
Let $\mathcal{S}_k$ be the subspace of $[\g,\g]$ that can be expressed as linear combinations of $l$-th order commutators of generators $A$ and $B$, for all $1\leq l\leq k$.
The base case is
\begin{align}
    \mathcal{S}_1 = \s\big\{[A,B]:=\overline{B}\big\} \subset\alg\big\{B,\overline{B}\big\}\,.
\end{align}
Now, as the inductive hypothesis, suppose $\mathcal{S}_k\subset\alg\{B,\overline{B}\}$ and let $X'\in\mathcal{S}_{k+1}$.
By definition,
\begin{align}
    X'\in\s\big\{X,[A,X],[B,X]\,:\,X\in\mathcal{S}_k\subset\alg\{B,\overline{B}\}\big\}\,.
\end{align}
Clearly, via the inductive hypothesis and \cref{eq:AX_property}, each of $X$, $[A,X]$, and $[B,X]$ are in $\alg\{B,\overline{B}\}$, so $X'\in\alg\{B,\overline{B}\}$.
Therefore, $\mathcal{S}_k\subset\alg\{B,\overline{B}\}$ for all $k\geq1$ and thus, $[\g,\g]\subset\alg\{B,\overline{B}\}$.
\end{proof}

\newpage
\appsec{Proof of \cref{lem:F}}{app:F_proof}
First, it is useful to recall \cref{prop:quasi_semi_universality:algebra}.
According to this result, an operator $A$ belongs to $\mathfrak{sv}_z$ if and only if (i) $A\in[\mathfrak{g},\mathfrak{g}]$, and (ii) $A$ respects the accidental symmetry, namely
\begin{align}
    \big[A,S_{j,j';\alpha,\alpha'}\big] = 0
\end{align}
for all $j>j'$, where $\alpha$ and $\alpha'$ are arbitrary multiplicity indices, and operator $S_{j,j';\alpha,\alpha'}$ is explicitly given by
\begin{align}
    S_{j,j';\alpha,\alpha'}
    := \sum_{m'=-j'}^{j'}
    \ket{j',m',\alpha'}\bra{j,m'-j+j',\alpha}
    \otimes
    \ket{2j-j'-m'}\bra{j'-m'}_{\rm osc}\,.
 \label{eq:S_def}
\end{align}
Since $A$ is PI, if this condition holds for one choice of $\alpha$ and $\alpha'$, then it holds for all choices.
Hence, we often assume $\alpha$ and $\alpha'$ are arbitrary but fixed, and abbreviate
\begin{align}
    S_{j,j'} \,:=\, S_{j,j';\,\al,\al'}\,.
\end{align}
For $j>j'$, it is useful to define
\begin{align}
    S_{j',j;\,\al',\al} \,:=\, S^\dagger_{j,j';\alpha,\alpha'}\,.
\end{align}
Now, if $A\in[\mathfrak{g},\mathfrak{g}]$, then $A$ respects both the permutational and U(1) symmetries.
Hence it is block-diagonal with respect to the projectors $\Pi_{q,j}$:
\begin{align}
    A
    \eq A\sum_{q,j}\Pi_{q,j}
    \eq \sum_{q,j}\Pi_{q,j}A\Pi_{q,j}
    \eq \sum_{q,j} A_{q,j}\,,
\end{align}
where
\begin{align}
    A_{q,j}:=\Pi_{q,j}A\Pi_{q,j}\,.
\end{align}
Therefore, the requirement that $A$ respect the accidental symmetry is equivalent to the requirement that, for every pair of sectors related by the accidental symmetry,
\begin{align}
    \big[A_{q,j}+A_{q',j'},\,S_{j,j';\alpha,\alpha'}\big]=0\,,
\end{align}
i.e., for every pair of angular momenta $j>j'$ and 
\begin{align}
    q=\frac{n}{2}+2j'-j\,\,,
    \qquad
    q'=\frac{n}{2}-j'+2j\,.
 \label{eq:q_q'_app}
\end{align}

Using this background, we are now ready to prove \cref{lem:F}.
Recall the assumption of the lemma: we are given $iF\in[\mathfrak{g},\mathfrak{g}]$.
Equivalently, $F$ is a PI, U(1)-invariant Hermitian operator satisfying
\begin{align}
    \Tr(F\Pi_{q,j}) \eq 0\,,
\end{align}
for all $q$ and $j$.
However, $F$ does not necessarily respect the accidental symmetry.
The lemma shows that, under the additional assumption
\begin{align} \label{eq:F_support_kk1}
    F \eq F\Big[\mathbb{I}_{\rm qubits}\otimes\big(\pure{k}_{\rm osc}+\pure{k+1}_{\rm osc}\big)\Big]\,,
\end{align}
there exists an operator $\widetilde F$ such that
\begin{enumerate}
    \item $i\widetilde F\in\mathfrak{sv}_z
    :=\mathfrak{alg}\{i\HTC,i\overline{\HTC}\}$;
    \item the restriction of $\widetilde F$ to the support of $F$ agrees with
    $F$, namely
    \begin{align}
        F \eq \widetilde F\Big[\mathbb{I}_{\rm qubits}\otimes\big(\pure{k}_{\rm osc}+\pure{k+1}_{\rm osc}\big)\Big].
    \end{align}
\end{enumerate}

The idea is to construct $\widetilde F$ by symmetrizing $F$ with respect to the accidental symmetry.
Since $F$ is PI and U(1)-invariant, it is block-diagonal with respect to the projectors $\Pi_{q,j}$:
\begin{align}
    F
    \eq F\sum_{q,j}\Pi_{q,j}
    \eq \sum_{q,j}\Pi_{q,j}F\Pi_{q,j}
    \eq \sum_{q,j}F_{q,j}\,,
\end{align}
where
\begin{align}
    F_{q,j} \,:=\, \Pi_{q,j}F\Pi_{q,j}\,.
\end{align}
Similarly, we define
\begin{align}
    \widetilde F
    \eq \sum_{q,j}\Pi_{q,j}\widetilde F\Pi_{q,j}
    \eq \sum_{q,j}\widetilde F_{q,j}\,.
\end{align}

For sectors $\mathcal{H}_{q,j}$ that are not paired with any other sector by the accidental symmetry, we simply define
\begin{align}
    \widetilde F_{q,j} \eq F_{q,j}\,.
\end{align}
For two sectors $\mathcal{H}_{q,j}$ and $\mathcal{H}_{q',j'}$ related by the accidental symmetry, we define
\begin{align}
    \widetilde F_{q,j}
    &\eq
    F_{q,j} +S_{j'\rightarrow j}(F_{q',j'})\,,
\end{align}
where the superoperator $S_{j\rightarrow j'}$ is defined by
\begin{align}
    S_{j\rightarrow j'}(F_{q,j})
    \,:=\,
    \sum_{\alpha'}
    S_{j,j';\alpha_0,\alpha'}\,
    F_{q,j}\,
    S_{j,j';\alpha_0,\alpha'}^\dagger\,,
    \qquad
    \alpha_0 \,\text{arbitrary}\,.
\label{eq:S_superoperator}
\end{align}
The choice of $\alpha_0$ is irrelevant because $F_{q,j}$ is PI.
The sum over $\alpha'$ ensures that $S_{j\rightarrow j'}(F_{q,j})$ is also PI.
Moreover, since the image lies entirely inside the fixed-charge sector $\mathcal{H}_{q',j'}$, the operator $S_{j\rightarrow j'}(F_{q,j})$ is U(1)-invariant.
Thus, $S_{j\rightarrow j'}(F_{q,j})$ is again a PI, U(1)-invariant operator.
Therefore, if $F_{q,j}$ is anti-Hermitian, then
\begin{align}
    S_{j\rightarrow j'}(F_{q,j})\in \mathfrak{g}\,.
\end{align}
Furthermore, because each $S_{j,j';\alpha_0,\alpha'}$ is a partial isometry from the corresponding copy of $\mathcal{H}_{q,j}$ into $\mathcal{H}_{q',j'}$, we have
\begin{align}
    S_{j,j';\alpha_0,\alpha'}^\dagger
    S_{j,j';\alpha_0,\alpha'}
    \eq
    \Pi_{q,j,\alpha_0}\,,
\end{align}
where $\Pi_{q,j,\alpha_0}$ denotes the projector onto the multiplicity copy labeled by $\alpha_0$.
Therefore,
\begin{align}
 \begin{split}
    \Tr\!\big(S_{j\rightarrow j'}(F_{q,j})\big)
    &\eq\sum_{\alpha'}\Tr\!\left(S_{j,j';\alpha_0,\alpha'}\,F_{q,j}\,S_{j,j';\alpha_0,\alpha'}^\dagger\right)\nonumber\\[4pt]
    &\eq \sum_{\alpha'}\Tr\!\left(F_{q,j}\,S_{j,j';\alpha_0,\alpha'}^\dagger S_{j,j';\alpha_0,\alpha'}\right)\\[4pt]
    &\eq \sum_{\alpha'} \Tr\!\left(F_{q,j}\,\Pi_{q,j,\alpha_0}\right)\,.
 \end{split}
\end{align}
Since $F_{q,j}$ is PI, its trace on each multiplicity copy is the same. 
Hence the trace of the image is proportional to the trace of $F_{q,j}$.
In particular, if
\begin{align}
    \Tr(F_{q,j})=0\,,
\end{align}
then
\begin{align}
    \Tr\!\big(S_{j\rightarrow j'}(F_{q,j})\big)=0\,.
\end{align} 
It follows that each block of $\widetilde F$ remains traceless:
\begin{align}
    \Tr\big(\Pi_{q,j}\widetilde F\big)=0\,,
\end{align}
for all $q$ and $j$. 
Therefore,
\begin{align}
    i\widetilde F\in[\mathfrak{g},\mathfrak{g}].
\end{align}
It remains to check that $\widetilde F$ respects the accidental symmetry.
By construction, for every paired pair of sectors $\mathcal{H}_{q,j}$ and $\mathcal{H}_{q',j'}$, the two blocks $\widetilde F_{q,j}$ and $\widetilde F_{q',j'}$ are related by the accidental-symmetry map.
Equivalently,
\begin{align}
    \big[\widetilde F_{q,j}+\widetilde F_{q',j'},\,S_{j,j';\alpha,\alpha'}\big]=0\,.
\end{align}
For unpaired sectors this condition is trivial, and hence, $\widetilde F$ respects the accidental symmetry.
Combining this with $i\widetilde F\in[\mathfrak{g},\mathfrak{g}]$, we conclude from \cref{prop:quasi_semi_universality:algebra} that
\begin{align}
    i\widetilde F\in\mathfrak{sv}_z\,.
\end{align}

Finally, we show that $\widetilde F$ agrees with $F$ after projection onto the subspace in which the oscillator has either $k$ or $k+1$ excitations.
Let
\begin{align}
    \Xi_{k,k+1}\,:=\,\pure{k}_{\rm osc}+\pure{k+1}_{\rm osc}\,,
\end{align}
where we suppress the tensor product with the identity operator on the qubits.
We first note that the accidental-symmetry operators have no matrix elements within this two-level oscillator subspace. 
Recall from \cref{eq:S_comms_Jz_adaga} that
\begin{align}
    \big[a^\dag a,\,S_{j,j'}\big] \eq 2(j-j')S_{j,j'}\,.
 \label{eq:S_number_comm}
\end{align}
Sandwiching this equation between $\pure{k}_{\rm osc}$ (and $\pure{k+1}_{\rm osc}$) on both sides gives
\begin{align}
    \pure{k}_{\rm osc}S_{j,j'}\pure{k}_{\rm osc}
    \eq \pure{k+1}_{\rm osc}S_{j,j'}\pure{k+1}_{\rm osc}
    \eq 0\,.
\end{align}
whenever $j\neq j'$.
Similarly, sandwiching \cref{eq:S_number_comm} between
$\pure{k+1}_{\rm osc}$ on the left and $\pure{k}_{\rm osc}$ on the right gives
\begin{align}
    \pure{k+1}_{\rm osc}\big[a^\dag a,S_{j,j'}\big]\pure{k}_{\rm osc}
    &\,\eq\,\pure{k+1}_{\rm osc}S_{j,j'}\pure{k}_{\rm osc}\,.
\end{align}
On the other hand, using the right-hand side of \cref{eq:S_number_comm}, the
same quantity is
\begin{align}
    \pure{k+1}_{\rm osc}\big[a^\dag a,S_{j,j'}\big]\pure{k}_{\rm osc}
    &\,\eq\, 2(j-j')\,\pure{k+1}_{\rm osc}S_{j,j'}\pure{k}_{\rm osc}\,.
\end{align}
Since $2(j-j')$ is a nonzero even number whenever $j\neq j'$, we conclude that
\begin{align}
    \pure{k+1}_{\rm osc}S_{j,j'}\pure{k}_{\rm osc} \eq 0\,.
\end{align}
The same argument, now sandwiching between $\pure{k}_{\rm osc}$ on the left and $\pure{k+1}_{\rm osc}$ on the right, gives
\begin{align}
    \pure{k}_{\rm osc}S_{j,j'}\pure{k+1}_{\rm osc} \eq 0\,.
\end{align}
Combining these relations, we obtain
\begin{align} \label{eq:proj_S_kk'}
    \Xi_{k,k+1}S_{j,j';\alpha,\alpha'}\Xi_{k,k+1} \eq 0\,.
\end{align}
\vspace{1em}

We now apply this observation to the symmetrized blocks.
For a pair of sectors $\mathcal{H}_{q,j}$ and $\mathcal{H}_{q',j'}$ related by the accidental symmetry,
\begin{align}
    \widetilde F_{q,j}
    \eq F_{q,j} + S_{j'\rightarrow j}(F_{q',j'})\,,
\end{align}
with
\begin{align}
    S_{j'\rightarrow j}(F_{q',j'})
    \eq
    \sum_{\alpha'}
    S_{j',j;\alpha_0,\alpha'}F_{q',j'}
    S_{j',j;\alpha_0,\alpha'}^\dagger\,.
\end{align}
Using \cref{eq:F_support_kk1}, we can insert $\Xi_{k,k+1}$ on both sides of
$F_{q',j'}$:
\begin{align}
    F_{q',j'} \eq \Xi_{k,k+1}F_{q',j'}\Xi_{k,k+1}\,.
\end{align}
Therefore,
\begin{align}
    S_{j'\rightarrow j}(F_{q',j'})\Xi_{k,k+1}
    &\eq
    \sum_{\alpha'}
    S_{j',j;\alpha_0,\alpha'}
    \Xi_{k,k+1}F_{q',j'}\Xi_{k,k+1}
    S_{j',j;\alpha_0,\alpha'}^\dagger
    \Xi_{k,k+1}
    \,\,\,\eq0\,,
\end{align}
where in the last line we used \cref{eq:proj_S_kk'}.
Hence,
\begin{align}
    \widetilde F_{q,j}\Xi_{k,k+1}
    \eq F_{q,j}\Xi_{k,k+1}
    \eq F_{q,j}\,.
\end{align}
The same argument applies to every paired sector, while for unpaired sectors
we have $\widetilde F_{q,j}=F_{q,j}$ by definition.
Summing over all $q$ and $j$, we conclude that
\begin{align}
    \widetilde F\Xi_{k,k+1}
    \eq \sum_{q,j}F_{q,j}
    \eq F\,.
\end{align}
Equivalently,
\begin{align}
    F
    \eq\widetilde F
    \Big[\mathbb{I}_{\rm qubits}\otimes\big(\pure{k}_{\rm osc}+\pure{k+1}_{\rm osc}\big)\Big]\,,
\end{align}
which proves that $\widetilde F$ agrees with $F$ on the required oscillator subspace.

\newpage
\appsec{Subsystem Universality in each sector: general version of \cref{lem:subspace_control}}{app:subsystem_control}
Here, we summarize and restate results originally from \cite{Schirmer_etal_2001_controllability}, which determine whether specific classes of Hamiltonians are (semi)-universal.
\Cref{lem:subspace_control}, which we use to prove \cref{thm:full_characterization}, is a special case of the following lemma:
\begin{lemma}\label{lem:subspace_control_general}
Let $\ket{y}: y=1,\dotsc, d$ be an orthonormal basis for $\C^d$. 
Consider the two Hamiltonians
\begin{align}
    B &= \sum_{y=1}^{d} b_y\,\pure{y}\,,\\[8pt]
    A &= \sum_{y=1}^{d-1} a_y\Big(\ket{y+1}\bra{y}+\ket{y}\bra{y+1}\Big)\,,
\end{align}
and assume all $a_y\in\R$ are nonzero.
Also, define the ''energy gaps'' $\Delta_y:=b_y-b_{y+1}$, and the second-order finite differences,
\begin{align}
    \Delta^2a_y := \begin{cases}
        2a_1^2 - a_2^2 & y=1\\[2pt]
        2a_y^2 - a_{y-1}^2 - a_{y+1}^2 & y=2,\dotsc,d-2\\[2pt]
        2a_{d-1}^2 - a_{d-2}^2 & y=d-1
    \end{cases}\,.
\end{align}
Then, the group of unitaries generated from  Hamiltonians $A$ and $B$ contains $\SU(d)$ -- equivalently, $\su(d)\subseteq\alg_{\R}\{iA,iB\}$, the real Lie algebra generated by $iA$ and $iB$ -- if any of the following conditions hold:
\begin{itemize}
    \item[1.] (Anharmonic system) At least one of the boundary energy gaps on the is different from all others, that is either
    \begin{itemize}
        \item[(i)]$\Delta_1\neq0$ and $|\Delta_y|\neq|\Delta_1|$ for all $y>1$, or
        \item[(ii)]$\Delta_{d-1}\neq0$ and $|\Delta_y|\neq|\Delta_{d-1}|$ for all $y<d-1$.
    \end{itemize}
    \item[2.] (Harmonic system) All energy gaps of $B$ are equal and nonzero, i.e. $\Delta_y=\Delta\neq0$ for all $y$, and at least one of boundary second-order finite differences is different from all others, that is either
    \begin{align*}
     \text{(i)}\quad\Delta^2a_y &\neq \Delta^2a_1 &&\text{for all}\,\,y>1\\
     \text{(ii)}\quad \Delta^2a_y &\neq \Delta^2a_{d-1} &&\text{for all}\,\,y<d-1\,.
    \end{align*}
\end{itemize}
\end{lemma}

Here, $B$ is interpreted as the intrinsic Hamiltonian of a given system, while $A$ is called a dipole interaction Hamiltonian, which induces a particular ordering of the energy eigenstates and describes transitions between adjacent states on this energy ladder.
The first type of restriction (anharmonic) describes a system where the energy gap $\Delta_y=b_y-b_{y+1}$ at either the top or bottom of the energy ladder  differs from all others.
Regardless of the values of transition amplitudes $\{a_y\}$, this condition is enough to guarantee semi-universality on $\C^d$.
The second type of restriction (harmonic) describes a system with constant energy gaps, but where the transition amplitudes introduce sufficient ``anharmonicity'', via second order differences, to achieve semi-universality.

In fact, condition (2.) can in a sense be reduced to condition (1.) as follows.
If $B$ is harmonic, i.e. satisfies $\Delta_y:=b_y-b_{y+1}=\Delta\neq0$ constant for all $y$, then one can obtain a new Hamiltonian $\widetilde{B}:=\sum_{y=1}^{d}\widetilde{b}_y\pure{y}$, such that $i\widetilde{B}\in\alg_{\R}\{iA,iB\}$, via
\begin{align}
    \Big[\big[iA,iB\big],iA\Big] = 2i\Delta\widetilde{B}\,,
\end{align}
in which case the energy gaps of $\widetilde{B}$ satisfy
\begin{align}
    \widetilde{\Delta}_y:=\widetilde{b}_y-\widetilde{b}_{y+1} \eq \Delta^2a_y\,.
\end{align}
Therefore, if the second-order finite differences of $A$ satisfy condition (2.), then $A$ and $\widetilde{B}$ describe an ``anharmonic'' system which satisfies condition (1.).

\Cref{lem:subspace_control} is a special case of \cref{lem:subspace_control_general}: in the assumptions of \cref{lem:subspace_control}, anti-Hermitian operators $i(A_++A_-)$ and $A_+-A_-$ generate via their commutator an ``anharmonic'' Hamiltonian $\widetilde{B}$, i.e.,
\begin{align}
    i\widetilde{B} := \frac{1}{2}\Big[A_+-A_-,\,i(A_++A_-)\Big]\,.
\end{align}

\newpage
\appsec{Relative Phase Constraints for $\V_z^+$ and $\mathcal{W}_z$}{app:Jz2_phases}
For the group $\V_z^+$ generated by Hamiltonians $\HTC$, $J_z$, and $J_z^2$, the constraint on relative phases analogous to \cref{eq:thm_phase_constraint} in \cref{thm:full_characterization} is, using the calculations in \cref{app:charge_vectors} and \cref{app:Jz2_charge_vector}:
\begin{align}
 \begin{split}
    \theta_{q,j} &\eq \theta_z\Tr\big(\pi\qj(J_z)\big) + \theta_2\Tr\big(\pi\qj\left(J_z^2\right)\big)\\[8pt]
    &\eq\begin{cases}
        \left(q-\dfrac{n}{2}+j+1\right) \\[4pt]
        \times
        \left[\left(q-\dfrac{n}{2}-j\right)
        \dfrac{\theta_z}{2} + \left(j(2j+1)-2j\left(q-\dfrac{n}{2}\right)+\left(q-\dfrac{n}{2}\right)(2q-n+1)\right)\dfrac{\theta_2}{6}\right]& :\,\,q<\dfrac{n}{2}+j\\[16pt]
        0 & :\,\,q\geq \dfrac{n}{2}+j\,,
    \end{cases}
 \end{split}
 \label{eq:Jz2_phase_constraint}
\end{align}
where $\theta_z,\theta_2\in[0,4\pi)$ and the equations hold modulo $2\pi$.

\vspace{4mm}
Similarly, for the group $\mathcal{W}_z$ generated by Hamiltonians $\HTC$, $J_z$, and $a^{\dag}a$, the analogous phase constraint is:
\begin{align}\label{eq:Wz_phase_constraint}
 \begin{split}
    \theta\qj &\eq \theta_z\Tr\big(\pi\qj(\Zhat)\big) + \theta_{\text{osc}}\Tr\big(\pi\qj(\Nhat)\big)\\[8pt]
    &\eq\begin{cases}
        \left(q-\dfrac{n}{2}+j+1\right)\times
        \left[\left(q-\dfrac{n}{2}-j\right)
        \dfrac{\theta_z}{2} + \left(q-\dfrac{n}{2}+j\right)\dfrac{\theta_{\text{osc}}}{2}\right]& :\,\,q<\dfrac{n}{2}+j\\[16pt]
        (2j+1)\left(q-\dfrac{n}{2}\right)\theta_{\text{osc}} & :\,\,q\geq \dfrac{n}{2}+j\,,
    \end{cases}
 \end{split}
\end{align}
where $\theta_z,\theta_{\text{osc}}\in[0,4\pi)$ and the equations hold modulo $2\pi$.

\end{document}